\documentclass{article}

\usepackage[english]{babel}
\usepackage[letterpaper,top=2cm,bottom=2cm,left=3cm,right=3cm,marginparwidth=1.75cm]{geometry}
\usepackage[utf8]{inputenc}
\usepackage[T1]{fontenc}
\usepackage{setspace}
\usepackage{palatino}

\usepackage{array}
\usepackage{enumitem}

\usepackage{amsmath,amsthm,amssymb,amsfonts,amscd,keyval}
\usepackage{mathtools,mathrsfs}
\mathtoolsset{mathic=true}
\usepackage{bbm}
\usepackage{adjustbox}
\usepackage{tensor}
\usepackage{braket}
\usepackage{slashed}
\usepackage{tikz-cd}
\usepackage{circuitikz}
\usepackage{multirow}
\usepackage{dsfont}
\usepackage{comment}
\usepackage{graphicx}
\usepackage{subcaption}
\usepackage{float}
\usepackage{abstract}
\usepackage{titlesec}
\usepackage{titletoc}
\usepackage[backend=biber,  style=ieee,
citestyle=numeric-comp, maxcitenames=3, maxbibnames=99 ]{biblatex}
\usepackage[colorlinks=true, allcolors=blue]{hyperref}
\usepackage{microtype} 
\usepackage{url}
\usepackage{tikz}
\usetikzlibrary{arrows.meta,calc}

\usepackage{authblk}

\newcommand{\comments}[1]{}
\DeclareMathOperator{\rk}{rk}
\DeclareMathOperator{\Ima}{im}

\DeclareMathOperator{\tr}{Tr}

\newcommand{\CZ}{\mathrm{C}Z}

\DeclareMathOperator{\CNOT}{\mathrm{CNOT}}
\newcommand{\CCZ}{\mathrm{C}\mathrm{C}Z}
\DeclareMathOperator{\supp}{supp}
\DeclareMathOperator{\C}{\mathcal{C}}

\DeclareMathOperator{\F}{\mathcal{F}}
\DeclareMathOperator{\HH}{\mathsf{H}}
\DeclareMathOperator{\G}{\mathsf{G}}

\newcommand{\mCZ}{\mathrm{C}^{r-1}Z}

\numberwithin{equation}{section} 
\theoremstyle{definition}
\newtheorem{definition}{Definition}[section]
\newtheorem{example}{Example}

\theoremstyle{plain}
\newtheorem{remark}{Remark}
\newtheorem{claim}{Claim}
\newtheorem{theorem}[definition]{Theorem}
\newtheorem{proposition}[definition]{Proposition}
\newtheorem{lemma}[definition]{Lemma}
\newtheorem{corollary}[definition]{Corollary}

\title{Non-Abelian sheaf quantum LDPC codes: \\good and magical}
\author{Zimu Li}
\author{Fuchuan Wei}
\author{Zhengyi Han}
\author{Zi-Wen Liu}
\affil{Yau Mathematical Sciences Center, Tsinghua University}
\date{}
\begin{document} 

\maketitle

\begin{abstract}
Non-Abelian quantum codes connect quantum error correction, phases of matter, and computational resources. In this work, we develop a general framework for constructing non-Abelian quantum low-density parity-check (qLDPC) codes by gauging sheaf codes via cup products and use it to establish code families with constant encoding rate and linear distance, simultaneously achieving optimal asymptotic scaling in both parameters. In particular, we characterize explicit logical representatives of the full gauged code space, and provide a fundamental treatment of non-Abelian code distance in the sense of Knill--Laflamme based on combining expansion with cleaning to establish protection against arbitrary errors. We further construct an almost-good family with entire code spaces exhibiting long-range magic. Moreover, we use linearly addressable cup product to directly generate linearly many  encoded magic states from gauging measurements of Clifford operators. These results extend good qLDPC codes beyond the Pauli stabilizer setting, lay the foundation for pursuing the no low-energy trivial magic conjecture and shed light on non-Abelian quantum phases and orders beyond geometric locality.
\end{abstract}

\tableofcontents


\section{Introduction}

Quantum low-density parity-check (qLDPC) codes are of major interest across practical fault-tolerant quantum computing, quantum many-body physics, quantum complexity theory, and mathematics. For practical quantum computing, the potential for low-overhead fault tolerance~\cite{Gottesman2014,Bravyi2024Memory} and compatibility with emerging experimental platforms such as reconfigurable atom arrays make qLDPC codes especially appealing~\cite{Xu2024ConstantOverhead,Bluvstein2026Architecture}. They also provide models of stable quantum phases beyond geometric locality~\cite{YinLucas2025,DeRoeck2025Phases} and connect to quantum complexity theory through the difficulty of preparing low-energy states, as exemplified by the no low-energy trivial states (NLTS) theorem~\cite{Anshu_2023}. Mathematically, the study of qLDPC codes has also established rich connections with algebraic topology, sheaf theory, high-dimensional expansion, systolic geometry, and geometric measure theory~\cite{EvraKaufmanZemor2024,FreedmanHastings2021,Li2025Poincare,LSWLL2026Theory,LuFuLiu2026Intrinsic}.
A central quest in the study of qLDPC codes is to improve their encoding rate and distance. Following decades of progress, recent breakthroughs have achieved the asymptotically optimal scaling of constant encoding rate and linear distance~\cite{PK2022Good,DHLV2022,QuantumTanner2022}. These constructions all fall within the sheaf code framework, in which local classical codes define a sheaf on a cell complex, and the associated cochain complex specifies a CSS code~\cite{Dinur2024sheaf,Panteleev2024,Li2025Poincare}.

Non-Abelian codes provide an important extension of the stabilizer setting, motivated in particular by the study of non-Abelian topological order in condensed matter physics~\cite{Kitaev_2003,Nayak2008}. In these systems, non-Abelian anyonic excitations support nonlocal encoding of quantum information and topologically protected operations through fusion and braiding. Recent work has extended this perspective beyond geometric locality, constructing non-Abelian qLDPC codes by gauging symmetries of underlying stabilizer codes~\cite{Christos2026,Zhu2026}. Such constructions open a path to exploring non-Abelian quantum phases on general connectivity. Another  key practical motivation is their fundamental connection to quantum magic and non-Clifford computation. For example, non-Abelian topological codes enable fault-tolerant non-Clifford gates and magic state preparation in two spatial dimensions~\cite{Davydova2025}, while gauging protocols enable encoded magic state preparation in qLDPC codes~\cite{Christos2026,Zhu2026}.

These developments pose natural questions: what encoding rates and distances can non-Abelian qLDPC codes achieve, and can the sheaf code framework be extended to construct such desirable codes? Existing gauging constructions provide important starting points, but establishing rigorous parameter guarantees presents fundamentally new challenges beyond the traditional Pauli stabilizer setting. Graph gauging offers a broadly applicable construction for which general guarantees on encoding rate and distance remain to be established~\cite{Christos2026}, whereas twisted hypergraph product constructions achieve a constant encoding rate and, under an additional assumption, a square-root lower bound on the distance of a protected subsystem~\cite{Zhu2026}. Gauging couples the logical cohomology classes of the parent CSS codes, so the resulting code dimension is not simply inherited and requires an explicit analysis of the logical constraints. Also importantly, establishing the code distance also requires special care, since the familiar characterization in terms of minimum-weight $X$- and $Z$-type logical representatives no longer applies.

Beyond their intrinsic mathematical interest, these questions are also motivated by their implications for quantum Hamiltonian complexity. In particular, good non-Abelian qLDPC codes are expected to supply a key ingredient toward the no low-energy trivial magic (NLTM) conjecture~\cite{Wei2025LRM}, which asks for Hamiltonians whose low-energy states cannot be prepared from arbitrary stabilizer states by shallow unitary circuits. Such circuit images admit efficient classical evaluation of local energies, so excluding them would address a further obstruction to the quantum PCP conjecture under the assumption $\mathsf{NP}\neq\mathsf{QMA}$~\cite{Wei2025LRM,AharonovAradVidick2013}. Studying the protection of nonstabilizerness in good non-Abelian qLDPC codes provides a key starting point for investigating whether magic can persist throughout a low-energy sector.

In this work, we construct non-Abelian qLDPC codes with constant encoding rate and linear distance by gauging good sheaf codes via cup products. Establishing these parameters requires resolving the coupled logical constraints introduced by gauging and proving error protection throughout the resulting non-Abelian code space.
An independent contribution is a fundamental treatment of code distance in the non-Abelian setting grounded in the general Knill--Laflamme condition~\cite{KnillLaflamme1997,KL2000}. We develop a systematic framework for establishing distance directly from the defining requirements of quantum error correction, addressing arbitrary low-weight errors throughout the entire code space. For gauged sheaf codes, we combine this framework with expansion properties and the cleaning of logical representatives to prove linear distance. This provides a rigorous basis and new approaches for studying quantum codes beyond the Pauli stabilizer setting.

Our main construction result is summarized as follows.
\begin{theorem}[Informal]\label{thm:main}
	Let $\mCZ$ denote a multi-controlled-$Z$ operator on $r$ qubits with $r-1$ controls. For any integer $r \geq 2$, there exists a family of non-Abelian quantum LDPC codes with parameters $[\![N,\Theta(N),\Theta(N)]\!]$ whose check generators are Pauli and $\mCZ$ operators.
\end{theorem}

Our construction uses (almost) good sheaf qLDPC codes as the backbone and dresses their $X$-type stabilizer generators with $\mCZ$ gates. We define the gauge constraints on logical cohomology classes through cup products and identify an explicit orthonormal basis of the non-Abelian code space in Theorem~\ref{thm:code_space_sheaf}. The code  dimension equals the number of tuples of logical cohomology classes satisfying these constraints, which we estimate by computing cup products of polarized logical representatives proposed in Ref.~\cite{LSWLL2026Theory}. We then study the generalized notion of distance of non-Abelian codes from a fundamental perspective, establishing rigorous bounds via the Knill--Laflamme contion in Theorem~\ref{thm:code_distance_sheaf}. The distance and expansion properties of the parent sheaf codes control the weights of correctable errors~\cite{PK2022Good,DHLV2022,Dinur2024sheaf,GJ2026,BLN2026,CHLT2026}. These results open a path to studying new non-Abelian quantum phases and orders with high encoding rates and sparse connectivity beyond geometric lattices. For clarity, we first introduce these results in the more tractable setting of HGP codes in Section~\ref{sec:gauged_HGP}. We give an explicit example of a gauged toric code with a $22$-dimensional code space and a rigorously proved distance in Example~\ref{example:gauged_toric}. We then develop the details and present our main results in the sheaf code setting in Section~\ref{sec:gauged_sheaf}.

We further demonstrate that our non-Abelian codes combine desirable error-correcting parameters with strong magic properties in a fundamental manner. First, using the code induction scheme of Refs.~\cite{LSWLL2026Theory,LLL2026nontrivial}, we construct a code family with constant encoding rate and almost-good distance and show through a dimension obstruction argument~\cite{Wei2025LRM} that all code states exhibits long-range magic. These results lay a foundation for pursuing the NLTM conjecture. In addition, we show that gauging measurements of Clifford operators \cite{Zhu2026,Christos2026} can be employed to generate $\Theta(N)$ many separate magic states. Our construction is based on the recent progress on the linearly addressable logical multi-controlled-$Z$ on good quantum locally testable codes (qLTCs)~\cite{LLL2609}.


\section{Preliminaries}

\subsection{HGP codes and sheaf codes}\label{sec:sheaf}

We briefly introduce hypergraph product (HGP) codes and sheaf codes, which are special instances of CSS codes and serve as the building blocks of our non-Abelian codes. More details can be found in, e.g., \cite{Li2025Poincare,LSWLL2026Theory,LLL2026nontrivial}. Throughout the paper, $\mathbb{F}_q$ denotes a finite field of characteristic $2$. That is, $q = 2^s$ for some large constant $s$. Roughly speaking, codes over $\mathbb{F}_q$ are realized by grouping $s$ physical qubits at each site. This is a necessary condition for obtaining product-expanding classical codes in the construction of sheaf codes with (nearly) optimal parameters \cite{PK2023RobustlyTestable,DHLV2022,KP2025Extendable,Dinur2024sheaf,LLL2026nontrivial,GJ2026,BLN2026,CHLT2026}.

\begin{definition}\label{def:cochain}
	A \emph{cochain complex} is a sequence of vector spaces $C^i$ and \emph{coboundary operators} $\delta^i$ between them, arranged as
	\begin{equation}
		\begin{tikzcd}
			\cdots \arrow[r, "\delta^{i-1}"] &
			C^i \arrow[r, "\delta^{i}"] &
			C^{i+1} \arrow[r, "\delta^{i+1}"] &
			C^{i+2} \arrow[r, "\delta^{i+2}"] &
			\cdots
		\end{tikzcd}
	\end{equation}
	such that $\delta^{i+1} \delta^i = 0$ for every $i$, called the \emph{coboundary condition}. We always assume that $C^i$ is a finite-dimensional vector space over a finite field $\mathbb{F}_q$. By taking transposes,
	\begin{equation}
		\begin{tikzcd}
			\cdots &
			C_i \arrow[l, "\partial_{i} "'] &
			C_{i+1} \arrow[l, "\partial_{i+1} "'] &
			C_{i+2} \arrow[l, "\partial_{i+2} "'] &
			\cdots \arrow[l, "\partial_{i+3} "']
		\end{tikzcd}
	\end{equation} 
	we obtain a \emph{chain complex} with $C_i := (C^i)^\ast \cong C^i$ and $\partial_i := (\delta^{i-1})^\ast = (\delta^{i-1})^T$ as the \emph{boundary operator}. Elements of $\ker \delta^i$ are called $i$-th \emph{cocycles}, and elements of $\Ima \delta^{i-1}$ are called $i$-th \emph{coboundaries}. Analogously, we can define \emph{cycles} and \emph{boundaries} on chain complexes.  We would also abbreviate a (co)chain complex to $(C_\bullet,\partial_\bullet)$ ($(C^\bullet,\delta^\bullet)$).
\end{definition}

\begin{definition}\label{def:CSS}
	A \emph{Calderbank--Shor--Steane (CSS) code} is a stabilizer code given by a 2-cochain complex over a finite field $\mathbb{F}_q$: 
	\begin{align}
		C^{i-1} \xrightarrow{\delta^{i-1}} C^i \xrightarrow{\delta^i} C^{i+1}
	\end{align}
	with $\delta^{i} = H_Z$ and $\delta^{i-1} = H_X^T$ defining the stabilizer generators. The $X$-type and $Z$-type logical operators are represented by elements from the cohomology $H^i(C) = \ker \delta^i / \Ima \delta^{i-1}$ and the homology $H_i(C) = \ker \partial_i / \Ima \partial_{i+1}$, respectively. The \emph{code dimension} $k$, which is the number of encoded logical qudits, is given by 
	\begin{align}
		k = \dim H^i(C) = \dim C^i - \rk \delta^i - \rk \delta^{i-1}
		= \dim C_i - \rk \partial_i - \rk \partial_{i+1} = \dim H_i(C).
	\end{align}
	Let $|x|$ denote the Hamming weight of $x$. The \emph{code distance} is given by $d = \min\{d_X, d_Z\}$, where
	\begin{align}
		d_X = \min_{x \in \ker \delta^i \setminus \Ima \delta^{i-1}} \vert x \vert, \quad 
		d_Z = \min_{z \in \ker \partial_i \setminus \Ima \partial_{i+1}} \vert z \vert.
	\end{align}
	Here $d_X$ ($d_Z$) is the \emph{cosystolic} (\emph{systolic}) distance.
\end{definition}

\begin{definition}\label{def:expansions}
	Given a $t$-dimensional (co)chain complex, the \emph{(co)cycle expansions} are defined by
	\begin{align}
		\epsilon_{\delta}(i) \coloneq \min_{x\in C^i\setminus\ker\delta^i} \frac{|\delta^i x|}{\operatorname{dist}(x,\ker\delta^i)}, \ 0 \leq i \leq t-1, \quad
		\epsilon_{\partial}(j) \coloneq \min_{z\in C_j\setminus\ker\partial_j} \frac{|\partial_j z|}{\operatorname{dist}(z,\ker\partial_j)}, \ 1 \leq j \leq t.
	\end{align}	
	When $1 \leq i,j \leq t-1$, the \emph{(co)boundary expansions} are defined by
	\begin{align}
		\beta_{\delta}(i) \coloneq \min_{x\in C^i\setminus \Ima\delta^{i-1}} \frac{|\delta^i x|}{\operatorname{dist}(x,\Ima\delta^{i-1})}, \qquad
		\beta_{\partial}(j) \coloneq \min_{z\in C_j\setminus \Ima\partial_{j+1}} \frac{|\partial_j z|}{\operatorname{dist}(z,\Ima\partial_{j+1})}.
	\end{align}	
	By definition, $\epsilon_{\delta}(i) \geq \beta_{\delta}(i)$ and $\epsilon_{\partial}(j) \geq \beta_{\partial}(j)$.
\end{definition}

As an example, for a connected graph, let $C^0 = \mathbb{F}_2^V$ be the vertex space and let $C^1 = \mathbb{F}_2^E$ be the edge space. The transposed graph incidence matrix $\delta^0: C^0 \rightarrow C^1$ defines a 1-dimensional cochain complex. Since $\ker \delta^0$ is 1-dimensional and spanned by the all-ones vector on the vertices,
\begin{align}
	\epsilon_{\delta}(0) = \min_{x\in C^0\setminus\ker\delta^0} \frac{|\delta^0 x|}{\operatorname{dist}(x,\ker\delta^0)}
\end{align}
is exactly the edge expansion or Cheeger constant of the graph. 

Lower bounds on these expansion constants play an important role in the study of high-dimensional expanders and locally testable codes \cite{PK2022Good,PK2023RobustlyTestable,DHLV2022,Dinur2024sheaf,GJ2026,BLN2026,CHLT2026}. In two dimensions, a weaker notion called small-set (co)boundary expansion \cite{PK2021,DHLV2022,Anshu_2023} is also useful.

\begin{definition}[\cite{DHLV2022}] \label{def:small_set}
	A 2-dimensional cochain complex is an \emph{$(\alpha,\beta,\gamma)$-small-set-coboundary expander} if for any $x \in C^1$ such that $\vert x \vert < \alpha \dim C^1$, there is some $u \in C^0$ such that 
	\begin{align}
		\vert \delta^1 x \vert \geq \beta \vert x + \delta^0 u \vert \text{ and } \vert u \vert \leq \gamma \vert x \vert. 
	\end{align}
	Small-set-boundary expanders are defined similarly.
\end{definition}

When the expansion parameter is positive, any $x$ with sufficiently small support that lies in $\ker \delta^1$ must be a coboundary, yielding a lower bound on the cosystolic distance.


\medskip 
Given two classical codes represented by 1-dimensional cochain complexes $H_i: \mathbb{F}_q^{M_i} \rightarrow \mathbb{F}_q^{N_i}$, $i = 1,2$, their 2-dimensional \emph{hypergraph product} is the following cochain complex:
\begin{align}\label{eq:2D_HGP}
	\mathbb{F}_q^{N_1 \times N_2} \xleftarrow{\delta^1}
	\mathbb{F}_q^{N_1 \times M_2} \oplus \mathbb{F}_q^{M_1 \times N_2} \xleftarrow{\delta^0} 
	\mathbb{F}_q^{M_1 \times M_2} 
\end{align}
with
\begin{align}
	\delta^1 = \begin{pmatrix} 
		I_{N_1} \otimes H_2 & H_1 \otimes I_{N_2}
	\end{pmatrix}, \quad
	\delta^0 = \begin{pmatrix} H_1 \otimes I_{M_2} \\ I_{M_1} \otimes H_2 \end{pmatrix}
\end{align}
It defines a CSS code with $H_Z = \delta^1$, $H_X^T = \delta^0$, and code parameters $[\![N, k, d]\!]$ given by~\cite{Zemor2014}
\begin{align}\label{eq:2D_HGP_parameters}
	\begin{aligned}
		& N = N_1 M_2 + M_1 N_2, 
		&& k = k_1 k_2^T + k_1^T k_2, \\
		& d_X = \min\{d_1, d_2 \},
		&& d_Z = \min\{d_1^T, d_2^T \},
	\end{aligned}
\end{align}
where $k_i, d_i$ ($k_i^T, d_i^T$) are the dimension and distance of the code defined by $\ker H_i$ ($\ker H_i^T$).

Similarly, the 3-dimensional hypergraph product uses three classical codes:
\begin{align}\label{eq:3D_HGP}
	\begin{aligned}
		\mathbb{F}_q^{N_1 \times N_2 \times N_3} \xleftarrow{\delta^2}
		& \mathbb{F}_q^{M_1 \times N_2 \times N_3} \oplus \mathbb{F}_q^{N_1 \times M_2 \times N_3} \oplus \mathbb{F}_q^{N_1 \times N_2 \times M_3}   \xleftarrow{\delta^1} \\
		& \mathbb{F}_q^{N_1 \times M_2 \times M_3} \oplus \mathbb{F}_q^{M_1 \times N_2 \times M_3} \oplus \mathbb{F}_q^{M_1 \times M_2 \times N_3} \xleftarrow{\delta^0} 
		\mathbb{F}_q^{M_1 \times M_2 \times M_3},
	\end{aligned}
\end{align}
where
\begin{align}\label{eq:HGP_coboundary}
	\delta^1 = \begin{pmatrix} 
		0 & I_{M_1} \otimes I_{N_2} \otimes H_3 & I_{M_1} \otimes H_2 \otimes I_{N_3} & \\ 
		I_{N_1} \otimes I_{M_2} \otimes H_3 & 0 & H_1 \otimes I_{M_2} \otimes I_{N_3} \\
		I_{N_1} \otimes H_2 \otimes I_{M_3} & H_1 \otimes I_{N_2} \otimes I_{M_3} & 0 
	\end{pmatrix},\quad
	\delta^0 = \begin{pmatrix} H_1 \otimes I_{M_2} \otimes I_{M_3} \\ I_{M_1} \otimes H_2 \otimes I_{M_3} \\ I_{M_1} \otimes I_{M_2} \otimes H_3	\end{pmatrix}
\end{align}
and $\delta^2 = (H_1 \otimes I_{N_2} \otimes I_{N_3} \quad I_{N_1} \otimes H_2 \otimes I_{N_3} \quad I_{N_1} \otimes I_{N_2} \otimes H_3 )$. Any three consecutive terms of \eqref{eq:3D_HGP} define a CSS code. For example, let $H_Z = \delta^1$ and $H_X^T = \delta^0$. Its code parameters $[\![N, k, d]\!]$ are~\cite{Zeng_2019}
\begin{align}\label{eq:HGP_parameters}
	\begin{aligned}
		& N = N_1 M_2 M_3 + M_1 N_2 M_3 + M_1 M_2 N_3, 
		&& k = k_1 k_2 k_3^T + k_1 k_2^T k_3 + k_1^T k_2 k_3, \\
		& d_X = \min\{d_1 d_2, d_1 d_3, d_2 d_3\},
		&& d_Z = \min\{d_1^T, d_2^T, d_3^T \}.
	\end{aligned}
\end{align}
The construction generalizes to higher dimensions.

Toric codes are standard examples of HGP codes: each constituent classical code is defined by the coboundary operator of a cycle graph, so $N_i = M_i$ is the number of vertices, or equivalently edges, in that cycle. The transpose $H_i^T$, mapping from edges to vertices, is the standard graph incidence matrix. In this work, we take $H_i$ to be defined by general graphs, so that the resulting combinatorial cell complexes support cup and cap products.


\medskip 
In order to build non-Abelian codes with (almost) good parameters, the base qLDPC codes should also be (almost) good. We now introduce the sheaf code framework, which encompasses HGP codes and these good code constructions. Let $X$ be any cell complex. A basic example is a square tessellation of the 2D torus. We avoid most abstract definitions, but it is helpful to keep in mind that the complexes $X$ used to construct (almost) good qLDPC codes differ substantially from topological manifolds \cite{KB1993,GG2012,FreedmanHastings2021}. We provide more concrete constructions in Section \ref{sec:cubical}. For now, suppose $X$ is given, then we construct a combinatorial cochain complex as follows:
\begin{enumerate}
	\item Let $X(i)$ be the collection of $i$-dimensional cells; for example, for the 2D torus, $i = 0,1,2$. Viewing these cells as a formal basis, we define the vector space $C^i = C^i(X,\mathbb{F}_q) \coloneq \mathbb{F}_q^{\vert X(i) \vert}$.
	
	\item The containment of a lower-dimensional cell $\sigma$ in a higher-dimensional cell $\tau$ is denoted by $\sigma \prec \tau$. Then the coboundary operator $\delta^i: C^i(X,\mathbb{F}_q) \rightarrow C^{i+1}(X,\mathbb{F}_q)$ is simply defined by 
	\begin{align}\label{eq:poset_coboundary}
		\delta^i (\sigma) = \sum_{\sigma \prec \tau \in X(i+1)} \tau. 
	\end{align}
	For $\sigma \in X(i), \pi \in X(i+2)$, we require an even number of cells $\tau \in X(i+1)$ such that $\sigma \prec \tau \prec \pi$. This guarantees the coboundary condition $\delta^{i+1} \delta^i = 0$.
	
	\item The cell complex $X$ is also assumed to be sparse: each cell $\sigma$ is contained in a constant number of cells $\tau$, and each $\tau$ contains a constant number of cells $\sigma$ of one lower dimension. This makes $\delta^i$ sparse, as required for an LDPC code.
\end{enumerate}
The standard cubical complex of the torus satisfies these assumptions, and the resulting cochain complex defines the toric code.

To construct codes with better parameters, we need a globally expanding complex $X$ in high dimension. The choice of a system of local coefficients is also important.

\begin{definition}\label{def:presheaf}
	Given $X$, a \emph{system of local coefficients} $\mathcal{F}$ consists of two types of data. First, to each $\sigma \in X$ we assign a vector space $\mathcal{F}_\sigma$. Second, for every $\sigma \preceq \tau$, we specify a linear map $\mathcal{F}_{\sigma,\tau}: \F_\sigma \rightarrow \F_\tau$, with these maps satisfying the compatibility condition
	\begin{align}
		\sigma \preceq \tau \preceq \pi  \implies  \mathcal{F}_{\tau,\pi} \circ \mathcal{F}_{\sigma,\tau} = \mathcal{F}_{\sigma, \pi}. 
	\end{align}
	We require $\mathcal{F}_{\sigma,\sigma}$ to be the identity map. In category theory, $\mathcal{F}$ is also known as a \emph{covariant functor}. 
\end{definition}

Given $X$ and $\mathcal{F}$, we generalize $C^\bullet(X,\mathbb{F}_q)$ to $C^\bullet(X, \mathcal{F})$ as follows:
\begin{enumerate}
	\item Let
	\begin{align}\label{eq:sheaf_cochain_space}
		C^i(X, \mathcal{F}) \coloneq \bigoplus_{\sigma \in X(i)} \F_\sigma.
	\end{align}
	For comparison, $C^i(X, \mathbb{F}_q) = \mathbb{F}_q^{\vert X(i) \vert} = \bigoplus_{\sigma \in X(i)} \mathbb{F}_q$. In this sense,  $C^\bullet(X,\mathbb{F}_q)$ corresponds to the \emph{constant sheaf with scalar local coefficients}.
	
	\item For any $x \in C^i(X, \mathcal{F})$, $x(\sigma) \in \F_\sigma$ is its local coefficient vector at $\sigma \in X(i)$. The coboundary operator $\delta^i: C^i(X, \mathcal{F}) \rightarrow C^{i+1}(X, \mathcal{F})$ is defined by 
	\begin{align}
		\delta^i( x(\sigma) ) = \sum_{\tau \in X(i+1), \tau \succ \sigma} \mathcal{F}_{\sigma, \tau} ( x(\sigma) ).
	\end{align} 
	The coboundary condition $\delta^{i+1} \circ \delta^i = 0$ follows from the compatibility in Definition \ref{def:presheaf} and the even-incidence condition above.
\end{enumerate}

\begin{definition}
	Let $X$ be a $t$-dimensional cell complex with a sheaf $\mathcal{F}$. A \emph{sheaf code} is a CSS code defined by extracting any three consecutive terms from the following cochain complex:
	\begin{equation}
		\begin{tikzcd}
			\cdots \arrow[r] & C^{i-1}(X, \mathcal{F}) \arrow[r,"\delta^{i-1}"] & 
			C^i(X, \mathcal{F}) \arrow[r,"\delta^{i}"] &
			C^{i+1}(X, \mathcal{F}) \arrow[r] & \cdots 
		\end{tikzcd}
	\end{equation}
\end{definition}

A concrete way to realize $\F$ in Definition \ref{def:presheaf} uses parity-check or generator matrices of classical codes, which we call local codes. Details are given in Section \ref{sec:cubical}.


\subsection{Cup and cap products on cubical complexes}\label{sec:cubical}

We now give explicit constructions of $X$ and $\F$. We then describe how to compute cup and cap products on the resulting sheaf codes, which are central to the construction of our non-Abelian codes. As a warm-up, we start with the 1D case.

Let $G_0 = (V_0, E_0)$ be an $n$-regular graph with $\vert V_0 \vert = n'$. The node degree $n$ is fixed. Let $\HH$ be a finite abelian group with $\vert{\HH}\vert = l$ elements. We define a \emph{lift} of $G_0$ via $\HH$: 
\begin{enumerate}
	\item For each vertex $v$, create a copy $(h,v)$ for each of the $l$ group elements $h \in \HH$.
	
	\item For any edge $e$ connecting $u$ and $v$, we prescribe an \emph{orientation}, such as $(u,v)$. Let $\gamma: E_0 \rightarrow \HH$ assign to each edge $(u,v)$ an element of $\HH$. We set $\gamma_{(v,u)} = \gamma_{(u,v)}^{-1}$.
	
	\item Construct a new graph $\hat{G}_0$ with vertex set $\{(h,v)\}$, joining the vertices $(h',u), (h,v)$ whenever $(u,v)$ is an oriented edge of the base graph $G_0$ and $h' = \gamma_{(u,v)} \cdot h$.
\end{enumerate}

We now use this graph lift to define a 1-dimensional complex $X$:
\begin{align}
	X(0) \coloneq \{ (h,v; b) \in \HH \times V_0 \times \{0,1\} \}.  
\end{align}
Here, $b = 0,1$ is a binary label. The set $X(0)$ consists of two copies of the vertex set of $\hat{G}_0$, which facilitates computations on sheaf codes \cite{LSWLL2026Theory}. Let us abbreviate $(h,v)$ by $g$ and rewrite $(h,v; b)$ as $[g; b]$. Elements in $X(0)$ are called $0$-cubes.

Analogously, $X(1)$ consists of 1-cubes that connect 0-cubes. We prescribe a local ordering of the $n$ edges incident with each vertex in $G_0$ and define $A \coloneq (a^1, \ldots,a^n)$ by
\begin{align}\label{eq:1D_generator_action}
	a^{\mu} \cdot (h,v) \coloneq (\gamma_{(v, v')} \cdot h,v')
\end{align} 
where $(v,v')$ is the $\mu$-th edge incident with $v$. Then
\begin{align}
	\begin{aligned}
		X(1) \coloneq & \Big\{ \Big( [g; 0], \ [a \cdot g; 1 ]  \Big): [g; 0] \in X(0), a \in A \Big\}.
	\end{aligned}
\end{align}
Here, we apply $a$ to $g$ via Eq.~\eqref{eq:1D_generator_action} and flip the binary bit $0$ to $1$. For simplicity, we denote each 1-cube by $[g;a]$. The resultant 1D complex $X$ is the \emph{double cover} of $G_0$.

For example, let $G_0 = K_3$ be the triangle with $n = 2$. We define the local ordering by counting counterclockwise. Suppose the lift group $\HH$ is trivial, i.e., contains only the identity element. Then $X(0)$ simply doubles $V_0$ and $X(1)$ consists of the red segments shown below. The double cover is the closed cycle $C_6$ with 6 vertices:
\begin{equation}\label{eq:K_3_double}
	K_3 \times K_2 = C_6 = 	
	\resizebox{0.2\linewidth}{!}{%
		\begin{tikzpicture}[x=0.75pt,y=0.75pt,yscale=-1,xscale=1, baseline=(current bounding box.center)]
			\draw  [line width=1.5]  (255.33,87.96) -- (136.33,18) -- (374.33,18) -- cycle ;
			\draw    (136.33,18) -- (136.33,203.04) ;
			\draw    (374.33,18) -- (374.33,203.04) ;
			\draw    (255.33,87.96) -- (255.33,273) ;
			\draw [line width=1.5]    (137.33,203.04) -- (256.33,273) ;
			\draw [line width=1.5]    (256.33,273) -- (375.33,203.04) ;
			\draw [line width=1.5]  [dash pattern={on 5.63pt off 4.5pt}]  (136.33,203.04) -- (374.33,203.04) ;
			\draw [color={rgb, 255:red, 208; green, 2; blue, 27 }  ,draw opacity=1 ]   (136.33,18) -- (255.33,273) ;
			\draw [color={rgb, 255:red, 208; green, 2; blue, 27 }  ,draw opacity=1 ]   (255.33,87.96) -- (136.33,203.04) ;
			\draw [color={rgb, 255:red, 208; green, 2; blue, 27 }  ,draw opacity=1 ]   (255.33,87.96) -- (374.33,203.04) ;
			\draw [color={rgb, 255:red, 208; green, 2; blue, 27 }  ,draw opacity=1 ]   (374.33,18) -- (255.33,273) ;
			\draw [color={rgb, 255:red, 208; green, 2; blue, 27 }  ,draw opacity=1 ] [dash pattern={on 4.5pt off 4.5pt}]  (374.33,18) -- (136.33,203.04) ;
			\draw [color={rgb, 255:red, 208; green, 2; blue, 27 }  ,draw opacity=1 ] [dash pattern={on 4.5pt off 4.5pt}]  (136.33,18) -- (374.33,203.04) ;
		\end{tikzpicture}
	}
\end{equation}

Random sampling of $\gamma_{(v,u)}$ over $\HH$ has been shown to produce graph lifts with favorable expansion properties \cite{Marcus2015I,Chandrasekaran2017,Hall_2018,Agarwal2016,Jeronimo2021}, which can be used to construct \emph{high dimensional expanders} (HDXs) \cite{DHLV2022,PK2022Good,QuantumTanner2022,Dinur2024sheaf,RSV2019,GJ2026,BLN2026,CHLT2026}. We only introduce ne approach that generalizes the above construction of the 1-dimensional complex $X$ to $t$ dimensions:
\begin{enumerate}
	\item Consider the $t$-dimensional Cartesian product: 
	\begin{align}
		(h,v_1, \ldots,v_t; b_1, \ldots,b_t) \in \HH \times V_0 \times \cdots \times V_0 \times \{0,1\} \times \cdots \times \{0,1\}.  
	\end{align}
	Here, $(b_1, \ldots,b_t)$ is a binary string. Let us abbreviate $(h,v_1, \ldots,v_j, \ldots,v_t)$ by $g$ and define
	\begin{align}
		X(0) \coloneq \{ [g; (b_j)_{j \in [t]} ] = (h,v_1, \ldots,v_t;b_1, \ldots,b_t) \} = \HH \times V_0^t \times \{0,1\}^t = \G \times \{0,1\}^t
	\end{align}
	as the collection of $0$-cubes or vertices. 	
	
	\item To define $X(1)$, we replace $A$ with $t$ families $A_j = (a_j^1, \ldots,a_j^n)$ with $j = 1, \ldots,t$. They act on $h$ and the corresponding components of the Cartesian product as follows:
	\begin{align}\label{eq:generator_action}
		a_j^{\mu} \cdot (h,v_1, \ldots,v_j, \ldots,v_t; b_1, \ldots,b_j, \ldots,b_t) 
		= (\gamma_{(v_j, v_j')} \cdot h,v_1, \ldots,v_j', \ldots,v_t; b_1, \ldots,1-b_j, \ldots,b_t).
	\end{align} 
	Here, $(v_j,v_j')$ is the $\mu$-th edge incident with $v_j$. Then
	\begin{align}
		\begin{aligned}
			X(1) \coloneq & \Big\{ \Big( [g; (b_j)_{j \in [t]} ], \ [a_i \cdot g; \{1-b_i\} \cup (b_j)_{j \neq i} ]  \Big): [g; (b_j)_{j \in [t]} ] \in X(0), a_i \in A_i \Big\}.
		\end{aligned}
	\end{align}
	For simplicity, we write
	\begin{align}
		[g; a_i, (b_j)_{j \neq i} ] = \Big( [g; b_i = 0, (b_j)_{j \neq i} ], [a_i \cdot g; b_i = 1, (b_j)_{j \neq i} ] \Big). 
	\end{align}
	
	\item More generally, let $S \subset [t]$ satisfy $\vert S \vert = p$. A $p$-cube $[g; (a_i)_{i \in S}, (b_j)_{j \in S^c} ]$ in $X(p)$ is defined as the product of edges $\prod_{i \in S} [g; a_i, (b_j)_{j \neq i} ]$. By definition, 
	\begin{align}\label{eq:system_size}
		\vert \G \vert = \vert \HH \times V_0 \times \cdots \times V_0 \vert = \vert \HH \vert \cdot \vert V_0 \vert^t,
	\end{align}
	Since for each $i \in [t]$, $\vert A_i \vert = n$, we have $\vert X(p) \vert = \binom{t}{p} 2^{t-p} n^p \vert \G \vert$. If we place physical qudits on $1$-cubes, the system size is $N = t 2^{t-1} n \vert \G \vert$.
\end{enumerate}

Throughout this paper, we would stick with 2-dimensional examples: for any $g = (h,v_1,v_2)$, the square below is a fundamental building block of $X$. For example, the $0$-cube $[g;0,0]$ is contained in the 1-cubes $[g;a_1,0]$ and $[g;0,a_2]$, both of which belong to the 2-cube $[g;a_1,a_2]$. The same construction extends to higher-dimensional cubes. Accordingly, $X$ is referred to as a \emph{cubical complex}. 

\begin{figure}[H]
	\centering
	\resizebox{0.45\linewidth}{!}{%
		\begin{tikzpicture}[x=0.75pt,y=0.75pt,yscale=-1,xscale=1]
			
			\draw  [color={rgb, 255:red, 0; green, 0; blue, 0 }  ,draw opacity=0.5 ] (184,189.69) -- (379.33,189.69) -- (379.33,385.02) -- (184,385.02) -- cycle ;
			
			\draw (120,378) node [anchor=north west][inner sep=0.75pt]  [font=\Large] [align=left] {$\displaystyle [ g;0,0]$};
			\draw (395,378) node [anchor=north west][inner sep=0.75pt]  [font=\Large] [align=left] {$\displaystyle [ a_{1} \cdot g;1,0]$};
			\draw (388,183) node [anchor=north west][inner sep=0.75pt]  [font=\Large] [align=left] {$\displaystyle [ a_{1} a_{2} \cdot g;1,1]$};
			\draw (90,183) node [anchor=north west][inner sep=0.75pt]  [font=\Large] [align=left] {$\displaystyle [ a_{2} \cdot g;0,1]$};
			\draw (237,386) node [anchor=north west][inner sep=0.75pt]  [font=\Large,color={rgb, 255:red, 208; green, 2; blue, 27 }  ,opacity=1 ] [align=left] {$\displaystyle [ g; a_{1} ,0]$};
			\draw (230,166) node [anchor=north west][inner sep=0.75pt]  [font=\Large,color={rgb, 255:red, 208; green, 2; blue, 27 }  ,opacity=1 ] [align=left] {$\displaystyle [ a_{2} \cdot g; a_{1} ,1]$};
			\draw (110,279) node [anchor=north west][inner sep=0.75pt]  [font=\Large,color={rgb, 255:red, 208; green, 2; blue, 27 }  ,opacity=1 ] [align=left] {$\displaystyle [ g;0, a_{2}]$};
			\draw (388,279) node [anchor=north west][inner sep=0.75pt]  [font=\Large,color={rgb, 255:red, 208; green, 2; blue, 27 }  ,opacity=1 ] [align=left] {$\displaystyle [ a_{1} \cdot g;1, a_{2}]$};
			\draw (238,279) node [anchor=north west][inner sep=0.75pt]  [font=\Large,color={rgb, 255:red, 245; green, 166; blue, 35 }  ,opacity=1 ] [align=left] {$\displaystyle [ g; a_{1} ,a_{2}]$};
		\end{tikzpicture}
	}
	\caption{A 2-cube $[g;a_1,a_2]$ with vertices (0-cubes) and edges (1-cubes).}
	\label{fig:2-cube}
\end{figure}
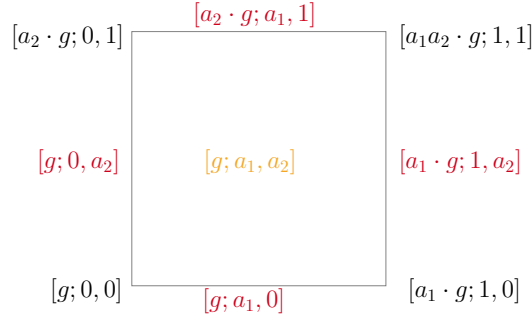

Several assumptions on $X$ in Section \ref{sec:sheaf} are straightforward to verify; for example, sparsity follows because $n$ is constant. In particular, let $G_0$ be a triangle. With a trivial lift, the complex $X$ defined by the $t$-dimensional product based on $G_0$ is isomorphic to the $t$-dimensional torus with $6^t$ vertices. The corresponding cochain complex $C^\bullet(X,\mathbb{F}_q)$ gives the toric code as in Figure~\ref{fig:toric_cubical}. As a reminder, although we always use these simple and heuristic examples, in the formation of sheaf codes with (almost) good parameters, nontrivial lifts are inevitable.


\medskip
We now use local codes to define the sheaf $\F$. For simplicity, we still consider the 2-dimensional case with the following cochain complex:
\begin{equation}\label{eq:DLV_scalar}
	\begin{tikzcd}
		C^0(X,\mathbb{F}_q) \arrow[r, "\delta^{0}"] &
		C^1(X,\mathbb{F}_q) \arrow[r, "\delta^{1}"] &
		C^2(X,\mathbb{F}_q).
	\end{tikzcd}
\end{equation}
We choose two classical codes $\C_i \subseteq \mathbb{F}_q^n$ generated by $h_i \in \mathbb{F}_q^{m_i \times n}$ with $m_i \leq n$, $i = 1,2$. That is, $\C_i = \Ima h_i^T$. Conceptually, $h_i$ is a local generator matrix of $\C_i$ or a local parity-check matrix of the dual code $\ker h_i = \C_i^\perp$. Then we assign each $p$-cube $\sigma = [g; (a_i)_{i \in I}, (b_j)_{j \in J}]$ a vector space $\mathcal{F}_\sigma = \bigotimes_{j \in J} \mathbb{F}_q^{m_j}$. Since $\sigma$ possess $\vert J \vert = t - p$ binary labels, the $(t-p)$-fold tensor product has one factor for each label. Note that $\mathcal{F}_\pi = \mathbb{F}_q$ is a scalar field when $\pi$ is a $t$-cube. Here, $t = 2$, and the local linear maps are defined as follows:
\begin{enumerate}
	\item For $\theta = [a_2 \cdot g;0,1]$ and $\sigma = [g; 0,a_2]$, 
	\begin{align}
		\mathcal{F}_{\theta,\sigma} = I_{m_1} \otimes h_2^T(a_2, \text{-}) \in \mathbb{F}_q^{m_1 \times m_1} \otimes \mathbb{F}_q^{1 \times m_2},
	\end{align} 
	where the $n$ elements $a_2 \in A_2$ label the columns of $h_2$ and $h_2^T(a_2, \text{-})$ is the transpose of the corresponding column. Equivalently, let $a_2$ also denote the corresponding standard basis vector of $\mathbb{F}_q^n$. The row vector $h_2^T(a_2, \text{-})$ can then be written as $(h_2 a_2)^T$.
	
	\item For $\sigma = [g;0,a_2]$ and $\tau = [g; a_1,a_2]$, 
	\begin{align}
		\mathcal{F}_{\sigma,\tau} =  (h_1 a_1)^T \otimes I_{m_2} \in \mathbb{F}_q^{1 \times m_1} \otimes \mathbb{F}_q^{m_2 \times m_2}.
	\end{align} 
\end{enumerate}
By Definition~\ref{def:presheaf}, $\mathcal{F}_{\sigma,\pi} = \mathcal{F}_{\tau,\pi} \circ \mathcal{F}_{\sigma,\tau}$. In the above example, $\sigma = [g;0,a_2]$ and $\tau = [g;a_1,a_2]$. Then
\begin{align}
	\F_{\sigma,\pi} = \big(  I_{m_1} \otimes (h_2 a_2)^T \big) \circ \big( (h_1 a_1)^T \otimes I_{m_2} \big) = (h_1 a_1)^T \otimes (h_2 a_2)^T .
\end{align}
We call the coboundary operators incorporating local codes \emph{sheaved coboundary operators}. It is also easy to see that the 1-dimensional sheaved coboundary operator of any graph defines a classical Tanner code.


\medskip
We now introduce the cup and cap products. Their formal definitions use more abstract notions, so we give only the formulas needed for computation; see, e.g., \cite{Li2025Poincare,LSWLL2026Theory,LLL2026nontrivial} for further details. The cup product of a $p_1$-cube and a $p_2$-cube is either a $(p_1+p_2)$-cube or zero. We illustrate these formulas in 2 dimensions using Figure~\ref{fig:2-cube}. For $0$-cubes, we have, e.g.,
\begin{align}\label{eq:cup_00}
	[g; 0,0] \smile [g; 0,0] = [g; 0,0], \quad [g; 0,0] \smile [a_1 \cdot g; 1,0] = 0.
\end{align}
For $0$-cubes and $1$-cubes,
\begin{align}\label{eq:cup_01}
	[g; 0,0] \smile [g; 0,a_2] = [g; 0,a_2], \quad [g; 0,a_2] \smile [g; 0,0] = 0.
\end{align}
The only nontrivial cup products between $1$-cubes are
\begin{align}\label{eq:1-cube_cup}
	[g; a_1,0] \smile [a_1 \cdot g; 1,a_2] = [g;a_1,a_2], \quad
	[g; 0,a_2] \smile [a_2 \cdot g; a_1,1] = [g;a_1,a_2].
\end{align}
As Eq.~\eqref{eq:cup_01} shows, the cup product depends on the order of its factors; for example, $[a_1 \cdot g; 1,a_2] \smile [g; a_1,0] = 0$. On the other hand, the cap product maps a $p_1$-cube and a $(p_1 + p_2)$-cube to a $p_2$-cube or $0$. Examples of nontrivial cap products are
\begin{align}
	& [g; 0,0] \frown [g; a_1,a_2] = [g;a_1,a_2], \\
	& [g; a_1,0] \frown [g; a_1,a_2] = [a_1 \cdot g;1, a_2], \quad 
	[g; 0,a_2] \frown [g; a_1,a_2] = [a_2 \cdot g;a_1, 1].
\end{align}
Cup and cap products of disjoint cubes are always zero. Formally, these cup and cap products are bilinear operators on $C^\bullet(X,\mathbb{F}_q)$ and $C_\bullet(X,\mathbb{F}_q)$, where each cube is treated as a one-hot vector in the cochain or chain space.

For sheaf codes, vectors $x \in C^\bullet(X,\F)$ are supported on cubes $\sigma$ with local coefficient vectors $v \in \F_\sigma$. Let $\supp(x)$ denote the \emph{geometric support} of $x$, namely, the collection of cubes with nonzero local coefficient vectors. Then, for example, \eqref{eq:1-cube_cup} generalizes to
\begin{align}\label{eq:1-cube_cup_sheaf}
	\smile: \ C^1(X,\mathcal{F}) \times C^1(X,\mathcal{G}) \rightarrow C^2(X,\mathcal{F} \otimes \mathcal{G}),
\end{align}
where $\mathcal{F}$ and $\mathcal{G}$ can be different sheaves and the tensor product $\mathcal{F} \otimes \mathcal{G}$ is defined by
\begin{align}
	(\mathcal{F} \otimes \mathcal{G})_{\sigma,\tau} \coloneq \mathcal{F}_{\sigma,\tau}  \otimes \mathcal{G}_{\sigma,\tau}: \ \mathcal{F}_{\sigma} \otimes \mathcal{G}_{\sigma} \rightarrow \mathcal{F}_{\tau} \otimes \mathcal{G}_{\tau}.
\end{align}	
When $\mathcal{F}$ and $\mathcal{G}$ are defined by local codes, their tensor product is determined by the Khatri--Rao product of the corresponding local generator (parity-check) matrices.

\begin{definition}[Khatri--Rao product] \label{def:KR_product}
	Let $A \in \mathbb{F}_q^{m_A \times n}$, $B \in \mathbb{F}_q^{m_B \times n}$. We denote by $A \ast B$ the $m_A m_B \times n$ matrix obtained by taking the tensor product of corresponding columns. This is a special case of the \emph{Khatri--Rao product}.
\end{definition} 

For example, consider $[g; a_1,0] \otimes \zeta_2$ and $[a_1 \cdot g; 1,a_2] \otimes \zeta_1 \in C^1(X,\mathcal{F})$ with $\zeta_i \in \mathbb{F}_q^{m_i}$. Their cup product is the 2-cube $[g;a_1,a_2]$ with local coefficient
\begin{align}\label{eq:cup_sheaf_2}
	[(h_2 a_2)^T  \otimes (h_1 a_1)^T] (\zeta_2 \otimes \zeta_1) = ( (h_2 a_2)^T \zeta_2 ) \cdot (  (h_1 a_1)^T \zeta_1 ) \in \mathbb{F}_q \otimes \mathbb{F}_q = \mathbb{F}_q.
\end{align}
We first take the cup product of the supporting cubes and then combine the local coefficients using tensor products.

Similarly, given $[g; 0,0] \otimes \Lambda$ with $\Lambda \in \mathbb{F}_q^{m_1} \otimes \mathbb{F}_q^{m_2}$ and $[g; a_1,a_2] \otimes c$ with $c \in \mathbb{F}_q$, we have
\begin{align}
	\frown: \ C^0(X,\F) \times C_2(X,\F) \rightarrow C_2(X, \mathbb{F}_q),
\end{align}
whose value in this example is the 2-cube $[g; a_1,a_2]$ with local coefficient
\begin{align}
	((h_1 a_1)^T \otimes (h_2 a_2)^T \Lambda) \cdot c \in \mathbb{F}_q.
\end{align}
We discuss other cases later.

We conclude this subsection by recording some important properties of the cup and cap products:

\begin{proposition}\label{prop:cup_cap}
	The following properties hold for the cup and cap products:
	\begin{enumerate}
		\item They satisfy the Leibniz rule:
		\begin{align}
			\delta(x_1 \smile x_2) = (\delta x_1) \smile x_2 + x_1 \smile (\delta x_2). 
		\end{align}
		where $x_i$ are cochains. For any $\xi \in C_\bullet(X, \mathcal{F})$, the cap product satisfies
		\begin{align}
			\partial(x \frown \xi) = (\delta x) \frown \xi + x \frown (\partial \xi),
		\end{align}
		
		\item The cup product is associative:
		\begin{align}
			(x_1 \smile x_2) \smile x_3 = x_1 \smile (x_2 \smile x_3).
		\end{align}
	\end{enumerate}
\end{proposition}

In particular, when $x \in C^p(X,\F), \xi \in C_p(X,\F)$, $x \frown \xi \in C_0(X,\mathbb{F}_q)$. We can further sum over the components of $x \frown \xi$ to obtain a scalar in $\mathbb{F}_q$. The resulting scalar defines a \emph{pairing}: 
\begin{align}\label{eq:paring}
	\langle x, \ \xi \rangle \coloneq \sum_{\sigma \in X(p_1)} \left\langle x(\sigma), \ \xi(\sigma) \right\rangle.
\end{align}
We also use the notation $\int_\xi x$ for the pairing as it can be interpreted as a discrete integral. For any $y \in C^{p - 1}(X,\mathcal{F})$, the identity
\begin{align}
	\langle \delta y, \ \xi \rangle = \langle y, \ \partial \xi \rangle \iff \int_\xi \delta y = \int_{\partial \xi} y 
\end{align}
is the discrete form of Stokes' theorem.


\subsection{Quantum gates over $\mathbb{F}_q$ and the Knill--Laflamme condition}\label{sec:F_q}

Let $q = p^s$, where $p$ is a prime. A $q$-dimensional qudit has computational basis $\{ \ket{x} : x \in \mathbb{F}_q \}$. Let $\omega \coloneq e^{\frac{2\pi i}{p}}$ and let
\begin{align}
	\tr_{\mathbb{F}_q/\mathbb{F}_p} : \mathbb{F}_q \to \mathbb{F}_p
\end{align}
be the \emph{trace map}. For every $\alpha,\beta \in \mathbb{F}_q$, the Pauli operators $X^\alpha$ and $Z^\beta$ are defined by
\begin{align}\label{eq:qudit-Pauli}
	X^\alpha \ket{x} = \ket{x+\alpha}, \qquad
	Z^\beta \ket{x} = e^{\frac{2\pi i}{p} \tr_{\mathbb{F}_q/\mathbb{F}_p}(\beta x)} \ket{x}.
\end{align}
Equivalently,
\begin{align}
	X^\alpha = \sum_{x\in\mathbb{F}_q} \ket{x+\alpha}\bra{x}, \qquad
	Z^\beta = \sum_{x\in\mathbb{F}_q} e^{\frac{2\pi i}{p} \tr_{\mathbb{F}_q/\mathbb{F}_p}(\beta x)} \ket{x}\bra{x}.
\end{align}
Notably, in prime-power dimension there is not just one distinguished nontrivial $X$ or $Z$: the Pauli operators are naturally indexed by $\alpha,\beta\in\mathbb{F}_q$. In particular, the notation $X^\alpha$ and $Z^\beta$ should be regarded as a family of operators indexed by elements of $\mathbb{F}_q$, rather than as ordinary exponentiation by arbitrary field elements. These operators satisfy
\begin{align}\label{eq:qudit-Pauli-commutation}
	Z^\beta X^\alpha = e^{\frac{2\pi i}{p} \tr_{\mathbb{F}_q/\mathbb{F}_p}(\alpha\beta)} X^\alpha Z^\beta.
\end{align}
Throughout the paper, $p = 2$ and $q = 2^s$. Therefore
\begin{align}
	Z^\beta X^\alpha = (-1)^{\tr_{\mathbb{F}_q/\mathbb{F}_2}(\alpha\beta)} X^\alpha Z^\beta,
\end{align}
and $(X^\alpha)^2 = (Z^\beta)^2 = I$.

For $\boldsymbol{\alpha} = (\alpha_1,\ldots,\alpha_n)\in\mathbb{F}_q^n$ and
$\boldsymbol{\beta} = (\beta_1,\ldots,\beta_n)\in\mathbb{F}_q^n$, we write
\begin{align}
	X^{\boldsymbol{\alpha}} \coloneq \bigotimes_{j=1}^n X^{\alpha_j}, \qquad
	Z^{\boldsymbol{\beta}} \coloneq \bigotimes_{j=1}^n Z^{\beta_j}.
\end{align}
There is an important difference between working over $\mathbb{F}_2$ and over the extension field $\mathbb{F}_q$. Suppose that $\C \subseteq \mathbb{F}_q^n$ is an $\mathbb{F}_q$-linear code and that $\boldsymbol{c}_1,\ldots,\boldsymbol{c}_k$ is an $\mathbb{F}_q$-basis of $\C$. It is not sufficient, in general, to take only $\{	X^{\boldsymbol{c}_1},\ldots,X^{\boldsymbol{c}_k}\}$ as group generators if one wants the Pauli operators $X^{\boldsymbol{c}}$ for every $\boldsymbol{c}\in C$. Indeed, for a general $\alpha \in \mathbb{F}_q$, the expression $\left(X^{\boldsymbol{c}}\right)^\alpha$ is not defined, and the presence of $X^{\boldsymbol{c}}$ in a group does not imply the presence of $X^{\alpha \boldsymbol{c}}$. The same statement holds for $Z$-type operators. Consequently, one must include the operators
\begin{align}
	X^{\alpha \boldsymbol{c}_i}, \quad Z^{\alpha \boldsymbol{c}_j}, \ \alpha \in \mathbb{F}_q,
\end{align}
when describing the $\mathbb{F}_q$ Pauli operators. We will emphasize this point again in Section \ref{sec:k_HGP}.


We now discuss the spectra and measurement outcomes of qudit Pauli operators. Let
\begin{align}
	\chi(\alpha) \coloneq e^{\frac{2\pi i}{p} \tr_{\mathbb{F}_q/\mathbb{F}_p}(\alpha)}, \quad \alpha \in \mathbb{F}_q,
\end{align}
be the standard \emph{character} of $\mathbb{F}_q$. That is, we view $\mathbb{F}_q$ as the additive abelian group $(\mathbb{F}_q, +)$, with the unitary representation
\begin{align}
	\rho: (\mathbb{F}_q, +) \rightarrow U(\mathbb{C}^q), \quad \rho(\alpha) \coloneq X^\alpha,
\end{align}
and each character corresponds to a 1-dimensional irreducible representation of $(\mathbb{F}_q, +)$. We have $\chi(\alpha + \beta) = \chi(\alpha) \chi(\beta)$, and Schur orthogonality gives
\begin{align}\label{eq:Schur_orthogonal}
	\sum_{\beta \in \F_q}\chi(\alpha \beta) =
	\begin{cases}
		q,& \alpha = 0,\\ 0,& \alpha \neq 0.
	\end{cases}
\end{align}
These facts give a common orthonormal eigenbasis for $\{X^\alpha\}$:
\begin{align}\label{eq:X_basis}
	\ket{\widetilde{\beta}} \coloneq \frac{1}{\sqrt q} \sum_{\gamma \in \F_q} \chi(-\beta \gamma) \ket{\gamma}, \ \beta \in \mathbb{F}_q
\end{align}
with $X^\alpha \ket{\widetilde{\beta}} = \chi(\alpha \beta) \ket{\widetilde{\beta}}$. Similarly, the states $\ket{\gamma}$ form a common orthonormal eigenbasis for $\{Z^\beta\}$.

When $p = 2$, the values of $\chi$ are $\pm 1$, and the corresponding eigenspaces can have dimension greater than one. For example, $\frac{I - X^\alpha}{2}$ projects onto the eigenspace with eigenvalue $-1$. In a complete $X$-basis measurement, the rank-one projector $\ket{\widetilde{\beta}} \bra{\widetilde{\beta}}$ corresponds to the outcome $\beta \in \mathbb{F}_q$.


\medskip
Over $\mathbb{F}_q$, a physical $\CZ$ gate is defined by
\begin{align}\label{eq:CZ} 
	\begin{aligned} 
		\CZ = \sum_{x_1,x_2 \in \mathbb{F}_q} e^{\frac{2\pi i}{p}\tr_{\mathbb{F}_q/\mathbb{F}_p} (x_1 \cdot x_2) }
		\ket{x_1,x_2} \bra{x_1,x_2}. 
	\end{aligned} 
\end{align} 
Since $p$ is a prime, we identify $\mathbb{F}_p$ with $\mathbb{Z}/p\mathbb{Z}$ and use the standard representatives $0,\ldots,p-1$ in the exponent, which makes $e^{\frac{2\pi i}{p} \tr_{\mathbb{F}_q/\mathbb{F}_p} (x_1 \cdot x_2) } \in \mathbb{C}$ well-defined. If $q = p = 2$, Eq.~\eqref{eq:CZ} reduces to the qubit $\CZ$ gate. Multi-controlled-$Z$ gates are defined similarly. More precisely, the $(r-1)$-controlled-$Z$ gate on $r$ qudits is
\begin{align}\label{eq:multiCZ}
	\begin{aligned}
		\C^{r-1}Z = \sum_{x_1,\ldots,x_r\in\mathbb{F}_q} e^{\frac{2\pi i}{p} \tr_{\mathbb{F}_q/\mathbb{F}_p} (x_1x_2\cdots x_r)}
		\ket{x_1,\ldots,x_r} \bra{x_1,\ldots,x_r}.
	\end{aligned}
\end{align}
In particular, $r=2$ gives the $\CZ$ gate in Eq.~\eqref{eq:CZ}, while $r=3$ gives the qudit $\CCZ$ gate \cite{Wills2024magic,He2025addressable,Nguyen2025CCZ}.

Similarly, the qudit analogue of the $\CNOT$ gate is $\CNOT \ket{x_1,x_2} = \ket{x_1,x_1+x_2}$ for $x_1,x_2\in\mathbb{F}_q$. Equivalently,
\begin{align}\label{eq:CNOT}
	\CNOT = \sum_{x_1,x_2\in\mathbb{F}_q} \ket{x_1,x_1+x_2}\bra{x_1,x_2} = \sum_{x\in\mathbb{F}_q} \ket{x}\bra{x}\otimes X^x.
\end{align}
Its inverse is given by $\CNOT^{-1}\ket{x_1,x_2} = \ket{x_1,x_2-x_1}$. In particular, with $q=2^s$, subtraction and addition coincide and hence $\CNOT^{-1}=\CNOT$. For every $\alpha,\beta\in\mathbb{F}_q$, the Pauli operators transform as
\begin{align}\label{eq:CNOT_Pauli}
	\CNOT X_1^\alpha \CNOT & = X_1^\alpha X_2^\alpha, & \CNOT X_2^\alpha \CNOT & = X_2^\alpha, \\
	\CNOT Z_1^\beta \CNOT & = Z_1^\beta, & \CNOT Z_2^\beta \CNOT & = Z_1^\beta Z_2^\beta.
\end{align}

\begin{theorem}[Knill--Laflamme condition \cite{KnillLaflamme1997,KL2000}] \label{thm:Pauli_KL}
	Let $\mathcal{Q} \subseteq (\mathbb{C}^q)^{\otimes n}$ have distance $d$, and let $w \ge 0$ be an integer. Let
	\begin{align}
	    	\mathcal{P}_w \coloneq \{ X^{\boldsymbol{\alpha}} Z^{\boldsymbol{\beta}}: |\boldsymbol{\alpha}| + |\boldsymbol{\beta}| \leq w \}
	\end{align}
	The set $\mathcal{P}_w$ contains \emph{all} Pauli strings of weight at most $w$. Then the following are equivalent:
	\begin{enumerate}[label=(\roman*)]
		\item $d \geq 2w+1$.
		
		\item Every pair $U,V \in \mathcal{P}_w$ obeys
		\begin{equation}\label{eq:Pauli_KL}
			\bra{i} U^\dagger V \ket{j} = \gamma_{UV} \delta_{ij}
		\end{equation}
		in one, and hence every, orthonormal code basis $\{\ket{i}\}$, with constants independent of $i,j$.
	\end{enumerate}
\end{theorem}

\begin{lemma}[Cleaning Lemma \cite{Bravyi2009Cleaning,Kalachev2022Cleaning}]\label{thm:Cleanning}
	For an $\mathbb{F}_q$-linear CSS code of Definition~\ref{def:CSS}, let $M \subseteq [n]$ satisfy $|M| < d$. Then:
	\begin{enumerate}[label=(\roman*)]
		\item For every $X$-type logical representative $x \in \ker H_Z = \ker \delta^i$, there is some $u \in C^{i-1}$ such that
		\begin{equation}\label{eq:Xclean}
			x' = x + \delta^{i-1} u, \qquad x'|_M = 0
		\end{equation}
		
		\item For every $Z$-type logical representative $z \in \ker H_X = \ker \partial_i$, there is some $v \in C_i$ such that
		\begin{equation}\label{eq:Zclean}
			z' = z + \partial_{i+1} v, \qquad z'|_M = 0.
		\end{equation}
	\end{enumerate}
	Namely, $\supp(x')$ and $\supp(z')$ are disjoint from $M$.
\end{lemma}


\section{Gauging HGP codes with scalar coefficients}\label{sec:gauged_HGP}

As a warm-up, we formulate the method to gauge HGP codes with scalar coefficients $C^\bullet(X,\mathbb{F}_q)$, i.e., $X$ is obtained with a trivial lift and the sheaf is constant. We show how this result recovers the well-known gauged toric codes and quantum doubles \cite{Kitaev_2003,twisted_quantum_double,Christos2026}. Then we present our main results of gauged sheaf codes in Section \ref{sec:gauged_sheaf}.

\subsection{Parameters of gauged HGP}\label{sec:k_HGP}

We only consider 2-dimensional HGP codes for brevity. To be precise, let $G_0$ be an $n$-regular connected graph. We also assume the connectedness of its double cover $G_0 \times K_2$ defined in \eqref{eq:K_3_double}. It has vertex and edge sets $V, E$ with $|V| = 2 |V_0|$ and $|E| = 2|E_0|$. For brevity, the symbols $V,E,V_0,E_0$ will also denote their cardinalities. Let $\delta: \mathbb{F}_q^{V} \to \mathbb{F}_q^{E}$ be the coboundary operator of the double cover of the graph as in Figure~\ref{eq:K_3_double}. In \eqref{eq:2D_HGP}, let $H_1 = H_2 = \delta$. Then the homological product defines a 2-dimensional HGP code. The associated cochain complex of \eqref{eq:2D_HGP} can also be written as
\begin{align}
	C^0(X,\mathbb{F}_q) \xrightarrow{\delta^0} C^1(X,\mathbb{F}_q) \xrightarrow{\delta^1} C^2(X,\mathbb{F}_q),
\end{align}
where $X$ is defined by the 2-fold Cartesian product of $G_0 \times K_2$.

To define a gauged HGP code, we need to find a nonzero 2-cycle $\xi \in C_2(X,\mathbb{F}_q)$, i.e., $\partial_2 \xi = 0$. By Eq.~\eqref{eq:2D_HGP_coboundary}
\begin{align}
	\partial_2 = \begin{pmatrix} I_E \otimes \partial \\ \partial \otimes I_E \end{pmatrix} \implies \xi \in \ker \partial_2 = \ker \partial \otimes \ker \partial
\end{align}
Without local codes, $\ker \partial$ is spanned by cycles in the base graph. For example, the base graph of a toric code is simply a cycle and thus $\xi$ is uniquely defined to support on every 2-cube of the torus. There will be more choices of $\xi$ for general graphs.

As a reminder, finding $\xi$ for sheaf codes with nontrivial local coefficients turns out to be extremely difficult and we will discuss more details in Section \ref{sec:gauged_sheaf}. For now, we only assume that $\xi$ is a generic 2-cycle. Let $v$ denote a standard basis vector of $C^0(X,\mathbb{F}_q)$. It is supported on a 0-cube geometrically. Similarly, let $p$ denote a standard basis vector of $C^0(X,\mathbb{F}_q)$. The $X$- and $Z$-stabilizer generators of HGP codes over $\mathbb{F}_q$ are
\begin{align}
	X_{\alpha \delta v}, \  Z_{\beta \partial p}, \quad \alpha, \beta \in \mathbb{F}_q,
\end{align}
where $X_{\alpha \delta v}$ applies the qudit Pauli $X^\alpha$ operator to each 1-cube in the support $\supp(\delta^0 v)$. As a reminder, we always abbreviate $\delta^0 v$ by $\delta v$ in the following context. The $Z$-stabilizers are defined similarly.

We now prepare three identical 2D HGP codes with scalar coefficients, labeled by $A,B,C$. Since $p = 2$, $e^{\frac{2\pi i}{p}} = -1$ and the dressed or gauged $X$ stabilizer generators are defined by
\begin{align}
	& \textcolor{red}{\mathcal{A}_{\alpha v}^A} = X_{\alpha \delta v}^A \sum_{x_B, x_C \in C^1}
	(-1)^{ \tr_{\mathbb{F}_q/\mathbb{F}_2}( \alpha \int_\xi v \smile x_B \smile x_C) } \ket{x_B,x_C} \bra{x_B,x_C}, \label{eq:X_A} \\
	& \textcolor{green}{\mathcal{A}_{\alpha v}^B} = X_{\alpha \delta v}^B \sum_{x_A, x_C \in C^1}
	(-1)^{ \tr_{\mathbb{F}_q/\mathbb{F}_2}( \alpha \int_\xi x_A \smile v \smile x_C) } \ket{x_A,x_C} \bra{x_A,x_C}, \label{eq:X_B} \\
	& \textcolor{blue}{\mathcal{A}_{\alpha v}^C} = X_{\alpha \delta v}^C \sum_{x_A, x_B \in C^1}
	(-1)^{ \tr_{\mathbb{F}_q/\mathbb{F}_2}( \alpha \int_\xi x_A \smile x_B \smile v) } \ket{x_A,x_B} \bra{x_A,x_B}, \label{eq:X_C}
\end{align}
As mentioned in Section \ref{sec:F_q}, writing $(-1)^{\alpha \int_\xi v \smile x_B \smile x_C}$ is illegal as $\int_\xi v \smile x_B \smile x_C \in \mathbb{F}_q$ may not have a natural embedding into $\mathbb{Z}$.

It is also important to verify that each $X$-stabilizer generator is dressed by a constant number of $\CZ$ gates. The operators $\textcolor{red}{\mathcal{A}_{\alpha v}^A}$ and $\textcolor{red}{\mathcal{A}_{v}^A}$ have the same weight when $\alpha \neq 0$, so we set $\alpha = 1$ here. In the cubical complex language, $v = [g;b_1,b_2]$ in 2D. Since the underlying cubical complex is sparse, we have the following observations:
\begin{enumerate}
	\item For $\textcolor{red}{\mathcal{A}_v^A}$ with $v = [g;0,0]$, by Eq.~\eqref{eq:cup_01} and \eqref{eq:1-cube_cup}, we know that for any $a_1,a_2$,
	\begin{align}\label{eq:cup_vee}
		[g;0,0] \smile [g; a_1,0] \smile [a_1 \cdot g; 1,a_2] = [g;a_1,a_2]
		= [g;0,0] \smile [g; 0,a_2] \smile [a_2 \cdot g; a_1,1].
	\end{align}
	There are $2n^2 = O(1)$ nontrivial cup products. As long as $\xi$ has a nontrivial support on these 2-cubes, the corresponding stabilizer is dressed by a physical $\CZ$ gate.
	
	\item If $v = [g;0,1], [g;1,0]$, or $[g;1,1]$, the cup products in Eq.~\eqref{eq:cup_vee} are always trivial and the corresponding $\textcolor{red}{\mathcal{A}_v^A}$ stabilizer is bare. In total, at most $\frac{1}{4}$ of the $\textcolor{red}{\mathcal{A}_v^A}$ stabilizers are truly dressed; each takes $2n^2$ $\CZ$ gates. 
	
	\item For a similar reason, at most $\frac{1}{4}$ of the $\textcolor{blue}{\mathcal{A}_v^C}$ stabilizers are dressed when $v = [g;1,1]$. Each takes $2n^2$ $\CZ$ gates.
	
	\item For $\textcolor{green}{\mathcal{A}_v^B}$, the dressed fraction is at most $\frac{1}{2}$ with $v = [g;0,1]$ and $[g;1,0]$. Each takes $n^2$ $\CZ$ gates according to the following identities:
	\begin{align}\label{eq:cup_eve}
		\hspace{-6mm} [g; a_1,0] \smile [a_1 \cdot g;1,0] \smile [a_1 \cdot g; 1,a_2] = [g;a_1,a_2]
		= [g; 0,a_2] \smile [a_2 \cdot g;0,1] \smile [a_2 \cdot g; a_1,1].
	\end{align}
\end{enumerate}

For any fixed $v$, as long as $\xi$ has support on the local view of $v$, i.e., 2-cubes that contain $v$, $\textcolor{red}{\mathcal{A}_v^A},\textcolor{green}{\mathcal{A}_v^B},\textcolor{blue}{\mathcal{A}_v^C}$ do not commute The $Z$ stabilizers are undressed $Z^A_{\beta\partial p},Z^B_{\beta\partial p},Z^C_{\beta \partial p}$. The entire collection of stabilizers contains both Pauli and Clifford operators and is non-Abelian. 

\begin{figure}[H]
	\centering
	\includegraphics[width=0.65\textwidth]{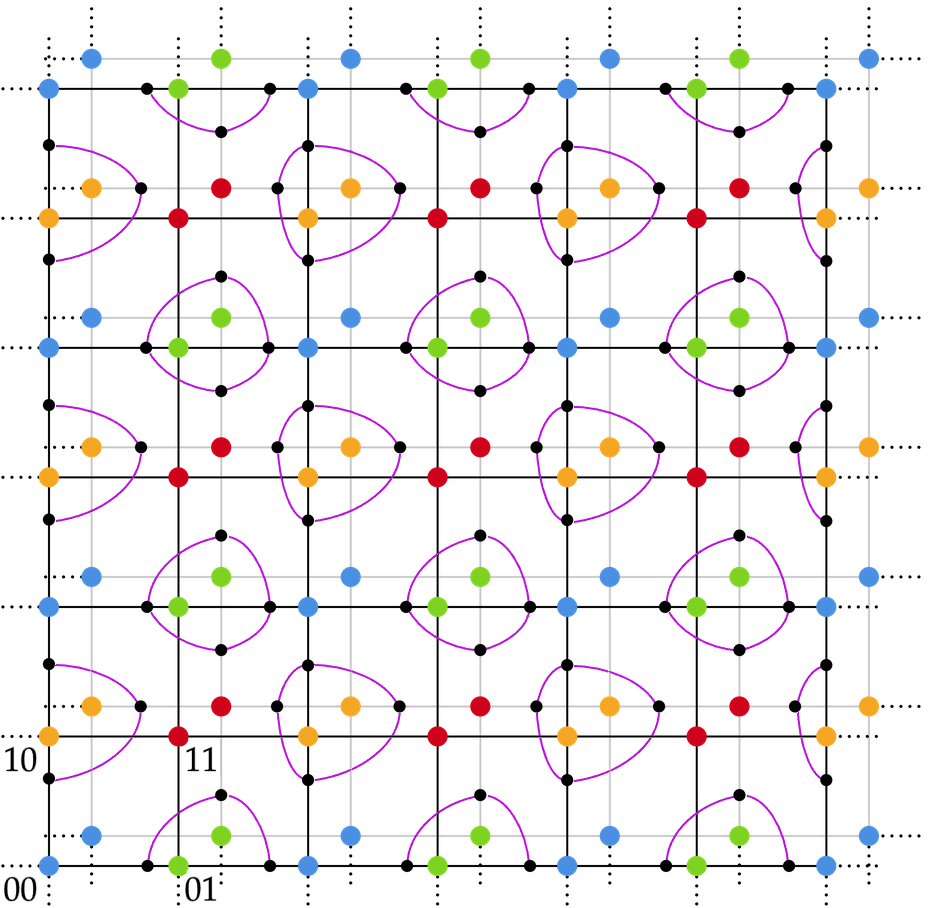}	
	\caption{Cup product $\CZ$ gate on 2D toric codes defined by the hypergraph product of the double cover of a triangle. Each purple curve represents a physical $\CZ$ gate. The binary bits $00,01,10,11$ on each $0$-cube are illustrated by their color: blue, green, brown, red. The configuration is slightly different from the conventional $\CZ$ gate \cite{Wang_2024,Breuckmann2024Cups,Christos2026} due to the definition of cubical complex. Nevertheless, the resulting logical operation is still a $\CZ$ action.}
	\label{fig:toric_cubical}
\end{figure} 


\begin{definition}
	Let $\mathbf{x} = (x_A, x_B, x_C)$ be any tuple of logical representatives of the three HGP code blocks. Then $\mathbf{x}$ is said to be \emph{gauged} if
	\begin{align}\label{eq:gauged_reps_HGP}
		\int_{\xi} x_B \smile x_C, \quad \int_{\xi} x_A \smile x_B, \quad \int_{\xi} x_A \smile x_C = 0.
	\end{align}
\end{definition}

We emphasize that Eq.~\eqref{eq:gauged_reps_HGP} can be derived from more fundamental identities in Definition~\ref{def:gauged_reps} and Theorem \ref{thm:code_space_sheaf}. On the other hand, we say $\mathbf{x}, \mathbf{x}'$ are inequivalent if there is at least one $x_i \nsim x_i'$ with respect to the corresponding coboundaries. If $\mathbf{x}$ is gauged, then by the Leibniz rule, any of its deformations via coboundaries is still gauged. As a result, the notion of a \emph{tuple of gauged cohomological classes} $[\mathbf{x}] = ([x_A],[x_B],[x_C]) \in H^1(X,\mathbb{F}_q) \times H^1(X,\mathbb{F}_q) \times H^1(X,\mathbb{F}_q)$ is well-defined. 

Let $L_G \subset H^1(X,\mathbb{F}_q) \times H^1(X,\mathbb{F}_q) \times H^1(X,\mathbb{F}_q)$ consist of gauged tuples of cohomological classes. As a caveat, it is only a subset but not a subspace. Then the following theorem holds.

\begin{theorem}\label{thm:code_space_HGP}
	Let $\mathbf{x} = (x_A, x_B, x_C)$ be an arbitrary tuple of gauged logical representatives and let $\mathbf{u} = (u_A, u_B, u_C)$ be an arbitrary vector from $C^0(X,\mathbb{F}_q) \times C^0(X,\mathbb{F}_q) \times C^0(X,\mathbb{F}_q)$. Given the above Clifford stabilizers, the non-Abelian code space is spanned by orthonormal logical states of the following form
	\begin{align}\label{eq:gauged_HGP_states}	
		\begin{aligned}
			\ket{\Psi^\xi_{\mathbf{x}}} = & \mathcal{N} \sum_{u_A,u_B,u_C \in C^0} c_{\mathbf{x}}^\xi (\mathbf{u}) 
			\Big\vert x_A + \delta u_A, x_B + \delta u_B, x_C + \delta u_C \Big\rangle.
		\end{aligned}
	\end{align}
	where $\mathcal{N}$ is a normalization factor and $c_{\mathbf{x}}^\xi (\mathbf{u}) = (-1)^{ \tr_{\mathbb{F}_q/\mathbb{F}_2} (F_{\mathbf{x}}^\xi (\mathbf{u}))}$ such that
	\begin{align}
		\begin{aligned}
			F_{\mathbf{x}}^\xi(\mathbf{u}) = & \int_{\xi} u_A \smile x_B \smile x_C + \int_{\xi} x_A \smile u_B \smile x_C + \int_{\xi} x_A \smile x_B \smile u_C \\
			& + \int_{\xi} u_A \smile \delta u_B \smile x_C + \int_{\xi} x_A \smile u_B \smile \delta u_C +  \int_{\xi} \delta u_A \smile x_B \smile u_C \\
			& + \int_{\xi} u_A \smile \delta u_B \smile \delta u_C.
		\end{aligned}
	\end{align}
	For cohomologically equivalent representatives $\mathbf{x} \sim \mathbf{x}'$, the associated logical states $\ket{\Psi^\xi_{\mathbf{x}}}$ and $\ket{\Psi^\xi_{\mathbf{x}'}}$ only differ by a phase factor; thus the code dimension equals the size of $L_G$.
\end{theorem}

We postpone the proof to Section \ref{sec:gauged_sheaf} in the general setting of sheaf codes. Intuitively, if $c_{\mathbf{x}}^\xi (\mathbf{u}) \equiv 1$ in Eq.~\eqref{eq:gauged_HGP_states}, the quantum state falls back into the original HGP code space. We also show in Example~\ref{example:gauged_toric} that Eq.~\eqref{eq:gauged_HGP_states} yields a $22$-dimensional code space, consistent with the torus ground state degeneracy of the $D_4$ quantum double.

We can extend the above construction beyond 2 dimensions. Taking $(t+1)$ copies of some $t$-dimensional HGP codes, we dress $X$ stabilizers by $\mathrm{C}^t Z$ gates as follows:
\begin{align}
	& \mathcal{A}_v^{(1)} = X_{\delta v}^{(1)} \sum_{x^{(2)},\ldots, x^{(t+1)} \in C^1}
	(-1)^{ \tr_{\mathbb{F}_q/\mathbb{F}_2}( \int_\xi v \smile x^{(2)} \smile \cdots \smile x^{(t+1)} ) } 
	\ket{ x^{(2)},\ldots, x^{(t+1)} } \bra{ x^{(2)},\ldots, x^{(t+1)} }, \label{eq:X_1} \\
	& \qquad \vdots \notag \\
	& \mathcal{A}_v^{(t+1)} = X_{\delta v}^{(t+1)} \sum_{x^{(1)},\ldots, x^{(t)} \in C^1}
	(-1)^{ \tr_{\mathbb{F}_q/\mathbb{F}_2}( \int_\xi x^{(1)} \smile \cdots \smile x^{(t)} \smile v ) } 
	\ket{ x^{(1)},\ldots, x^{(t)} } \bra{ x^{(1)},\ldots, x^{(t)} }. \label{eq:X_t}
\end{align}
The $Z$ stabilizers are simply $Z^{(1)}_{\beta \partial p},...,Z^{(t)}_{\beta \partial p},Z^{(t+1)}_{\beta \partial p}$. As before, the collection of stabilizers involves non-Clifford operators and Theorem \ref{thm:code_space_HGP} still holds. Let $d$ be the code distance of the HGP code and let $\epsilon = \epsilon_\delta(0)$ be the $0$-cocycle expansion of the HGP cochain complex (see Definition \ref{def:expansions}). Then we have the following bound on the gauged HGP code distance.

\begin{theorem}\label{thm:code_distance_HGP}
	The gauged HGP code has distance
	\begin{align}
		d_{\text{gauge}} \geq c(t,n) \cdot \epsilon \cdot d
	\end{align}
	where $c(t,n)$ is a function that depends only on the node degree $n$ and the dimension $t$ of the cochain complex.
\end{theorem}

Again, the proof is postponed to the general case of sheaf codes. We also provide an elementary example about gauged toric codes in Example \ref{example:gauged_toric}.


\subsection{Gauged toric codes and quantum double}\label{sec:toric}

We compute several examples of gauged HGP codes with explicit code parameters. To begin with,  we start from constructing \emph{canonical logical representatives} for a two-dimensional HGP code. This is used to count the number of gauged tuples of cohomological classes, which yields the gauged code dimension. For $i=1,2$, write
\begin{align}
	\mathbb{F}_q^{M_i} \xrightarrow{H_i} \mathbb{F}_q^{N_i}, \quad
	r_i = \rk H_i, \quad
	k_i = M_i-r_i = \dim \ker H_i,\quad k_i^T = N_i-r_i = \dim \ker H_i^T.
\end{align}
We use exactly the four groups of vectors introduced in ~\cite{Fu2025nogo,LSWLL2026Theory}:
\begin{align}
	e_1^{(i)},\ e_2^{(p)},\qquad \zeta_1^{(q)},\ \zeta_2^{(j)},\qquad
	f_1^{(i)},\ f_2^{(p)},\qquad \eta_1^{(q)},\ \eta_2^{(j)}.
\end{align}
We first construct $e_1^{(i)}$ and $f_1^{(i)}$. Applying elementary column reduction to $H_1$, and permuting its rows if necessary,  we obtain
\begin{align}\label{eq:HGP_systematic_output}
	\widetilde{H}_1 =
	\begin{pmatrix}
		I_{N_1-k_1^T} & 0 \\
		J_1 & 0
	\end{pmatrix}, \qquad J_1\in\mathbb{F}_q^{k_1^T \times (N_1-k_1^T)}.
\end{align}
When solving the linear system $x^T H_1 = 0$, the first $N_1-k_1^T$ rows of $\widetilde{H}_1$ are called \emph{pivot} rows and the last $k_1^T$ rows are \emph{free} rows. For $i\in[k_1^T]$, define
\begin{align}\label{eq:HGP_e1_definition}
	e_1^{(i)} = \begin{pmatrix} 0 \\ \varepsilon_i \end{pmatrix} \in \mathbb{F}_q^{N_1},
\end{align}
where $\varepsilon_i$ is the $i$-th standard basis vector of $\mathbb{F}_q^{k_1^T} \subset \mathbb{F}_q^{N_i}$. Namely, $e_1^{(i)}$ is the one-hot vector on the $i$-th free variable. By basic linear algebra, these vectors span a complement of $\Ima H_1$. 

Next observe that
\begin{align}\label{eq:J_I_kernel}
	\begin{pmatrix} J_1 & I_{k_1^T} \end{pmatrix}
	\begin{pmatrix}
		I_{N_1-k_1^T} & 0 \\
		J_1 & 0
	\end{pmatrix} = J_1+J_1=0.
\end{align}
Then we define $f_1^{(i)}$ to be the transpose of the $i$-th row of $\begin{pmatrix}J_1 & I_{k_1^T}\end{pmatrix}$. By definition,
\begin{align}
	f_1^{(i)} \in \ker H_1^T, \qquad \langle e_1^{(i)},f_1^{(i')}\rangle=\delta_{i,i'}.
\end{align}
We define $e_2^{(p)}$ and $f_2^{(p)}$, $p\in[k_2^T]$, by the same procedure applied to $H_2$.

On the other hand, we apply row reduction to $H_i$ and choose one-hot representatives
\begin{align}
	\eta_1^{(q)} \in \mathbb{F}_q^{k_1} \subset \mathbb{F}_q^{M_1}, \qquad \eta_2^{(j)} \in \mathbb{F}_q^{k_2} \subset \mathbb{F}_q^{M_2},
\end{align}
and choose the kernel bases of $\ker H_i$ compatibly so that
\begin{align}\label{eq:HGP_zeta_xi_pairing}
	\langle\zeta_1^{(q)},\eta_1^{(q')}\rangle = \delta_{q,q'}, \qquad \langle\zeta_2^{(j)},\eta_2^{(j')}\rangle = \delta_{j,j'}.
\end{align}

Recall that the HGP code is defined on $C^1$ consisting of two sectors:
\begin{align}
	C^1(X,\mathbb{F}_q) = (\mathbb{F}_q^{N_1} \otimes \mathbb{F}_q^{M_2}) \oplus (\mathbb{F}_q^{M_1} \otimes \mathbb{F}_q^{N_2}).
\end{align}
with
\begin{align}\label{eq:2D_HGP_coboundary}
	\delta^1 = \begin{pmatrix} 
		I_{N_1} \otimes H_2 & H_1 \otimes I_{N_2}
	\end{pmatrix}, \quad
	\delta^0 = \begin{pmatrix} H_1 \otimes I_{M_2} \\ I_{M_1} \otimes H_2 \end{pmatrix}
\end{align}
It turns out that two families of canonical $X$-logical representatives are
\begin{align}\label{eq:HGP_canonical_X}
	e_1^{(i)} \otimes \zeta_2^{(j)} \in \mathbb{F}_q^{N_1} \otimes \mathbb{F}_q^{M_2}, & \quad i \in [k_1^T], \ j\in[k_2], \\
	\zeta_1^{(q)} \otimes e_2^{(p)} \in \mathbb{F}_q^{M_1} \otimes \mathbb{F}_q^{N_2}, & \quad q \in [k_1], \ p\in[k_2^T],
\end{align}
It is easy to show that they are cocycles but not coboundaries. Their total number is
\begin{align}
	k_1^T k_2 + k_1 k_2^T = \dim H^1(X,\mathbb{F}_q),
\end{align}
so the cohomological classes of Eq.~\eqref{eq:HGP_canonical_X} form a basis of the cohomology.

Transposing the cochain complex and applying the same construction gives the canonical $Z$-logical representatives
\begin{align}\label{eq:HGP_canonical_Z}
	f_1^{(i)} \otimes \eta_2^{(j)} \in \mathbb{F}_q^{N_1} \otimes \mathbb{F}_q^{M_2}, & \quad i\in[k_1^T], \ j \in [k_2], \\
	\eta_1^{(q)} \otimes f_2^{(p)} \in \mathbb{F}_q^{M_1} \otimes \mathbb{F}_q^{N_2}, & \quad q\in[k_1], \ p \in [k_2^T].
\end{align}
The preceding pairings give
\begin{align}
	\left\langle e_1^{(i)}\otimes\zeta_2^{(j)}, f_1^{(i')} \otimes \eta_2^{(j')} \right\rangle & = \delta_{i,i'}\delta_{j,j'}, \\
	\left\langle \zeta_1^{(q)} \otimes e_2^{(p)}, \eta_1^{(q')} \otimes f_2^{(p')} \right\rangle & = \delta_{q,q'}\delta_{p,p'},
\end{align}
and pairings between different sectors vanish. 

\begin{example}\label{example:gauged_toric}
	As an immediate example, the toric code has exactly two canonical $X$-logical representatives
	\begin{align}
		x_V = e_1 \otimes \zeta_2, \quad x_H = \zeta_1 \otimes e_2.
	\end{align}
	To be precise, let $G_0 \times K_2$ be the double cover of a cycle $G_0$. When $V_0$ is odd, the double cover is still a connected cycle of $V = 2V_0$ vertices as in Eq.~\eqref{eq:K_3_double}. For the toric code constructed by the corresponding 2-dimensional hypergraph product, $\zeta_2 = \zeta_1 \in \ker H_i = \ker \delta \subset \mathbb{F}_2^{2V_0}$ is the all-ones vector with $k_i = \dim \ker H_i = 1$. The vectors $e_1, e_2$ are arbitrary one-hot vectors in $\mathbb{F}_2^{2E_0}$ as $k_i^T = \dim \ker H_i^T = 1$. Pictorially, $x_V$ represents the vertical non-contractible loop and $x_H$ is the horizontal non-contractible loop.
	
	We now check the gauge conditions for toric codes:
	\begin{align}
		\int_\xi x_B \smile x_C = \int_\xi x_A \smile x_B = \int_\xi x_A \smile x_C = 0.
	\end{align}
	Since $x_A,x_B,x_C$ are cocycles, any of them can be written as $ax_V + bx_H + \delta^0 u$ for some $u \in C^0$. By multi-linearity of cup products and the computation rules introduced in Section \ref{sec:cubical}, 
	\begin{align}
		\int_\xi x_B \smile x_C = \int_\xi (ax_V + b x_H) \smile  (c x_V + d x_H) = ad + bc,
	\end{align}
	so we have the following possibilities for the choices of $[(a,b),(c,d)]$:
	\begin{align}
		\begin{aligned}
			& [(0,0),(0,0)], \quad [(1,1),(1,1)], \\
			& [(1,0),(0,0)], \quad [(0,1),(0,0)], \quad [(0,0),(1,0)], \quad[(0,0),(0,1)], \\ 
			& [(1,1),(0,0)], \quad [(0,0),(1,1)], \quad [(1,0),(1,0)], \quad [(0,1),(0,1)].
		\end{aligned}
	\end{align}
	The total number of combinations of $[(a_1,b_1),(a_2,b_2),(a_3,b_3)]$ satisfying all three cup product conditions equals $22$, in agreement with the ground-state degeneracy of the $D_4$ quantum double on the torus\cite{Kitaev_2003,twisted_quantum_double}.
	
	Let $d = \Theta(\sqrt{N})$ denote the ordinary toric code distance. The proof method of Theorem \ref{thm:code_distance_sheaf} indicates that the gauged toric code has distance $\geq \frac{d}{4} = \Theta(\sqrt{N})$. The proof uses the cleaning lemma, which we present later. What we want to emphasize is the property of cocycle expansions. As mentioned, the toric code is formed by a product of cycles and it is well-known that they do not admit good expansion properties, i.e., the $\epsilon$ in Theorem \ref{thm:code_distance_HGP} cannot be a constant for toric codes. Nevertheless, the actual distance still attains  $\Theta(\sqrt{N})$. The reason goes as follows: let
	\begin{align}
		E = (\mathcal{X}_A, \mathcal{Z}_A; \mathcal{X}_B, \mathcal{Z}_B; \mathcal{X}_C, \mathcal{Z}_C).
	\end{align}
	be a tuple of Pauli errors on the gauged toric codes. We show in Theorem  \ref{thm:code_distance_sheaf} that $\mathcal{X}_i$ are always coboundaries, i.e., $\mathcal{X}_i = \delta w_i$ with $i = A,B,C$. For a cubical complex with a constant cocycle expansion in Definition \ref{def:expansions}, we can choose $w_i$ such that its weight is bounded by $\mathcal{X}_i$ and this further enables the use of the cleaning lemma.
	
	For toric codes, we do not have such convenience. Instead, suppose $|\mathcal{X}_i| < \frac{d}{2}$. Then we can always find some $u_i$ such that $\delta^0 u_i = \mathcal{X}_i$ and
	\begin{align}
		X_{\geq u_i}(2) \coloneq \{ \pi \in X(2): \pi \text{ contains some support of } u_i\}
	\end{align}	
	can always be cleaned. The reason is simple: we represent the torus as a $d \times d$ grid. On any row of the grid, if $u_i$ is not a constant vector, then $\delta^0 u_i$ contributes at least two nonzero entries to $\mathcal{X}_i$ in that row. Since we have $d$ rows and since $|\mathcal{X}_i| < \frac{d}{2}$, there are at least $\frac{3d}{4}$ rows on which $u_i$ is constant. By the same argument, there are at least $\frac{3d}{4}$ columns on which $u_i$ is constant. If $u_i \equiv 0$ there, then we simply place the canonical logical representatives on the corresponding columns and rows. Otherwise, by adding the all-ones vector in $C^0$, we can flip the value of $u_i$ over $\mathbb{F}_2$. In any case, the cup products in \eqref{eq:rep_error} vanish and all the other arguments in Theorem \ref{thm:code_distance_sheaf} hold to give the $\Theta(\sqrt{N})$ distance. 
\end{example}


\medskip
Suppose $G_0$ and its double cover are connected $n$-regular graphs. Let $\{[L_i]\}$ denote a basis of the cohomology via canonical $X$-logical representatives. As before
\begin{align}
	[x_A] = \left[ \sum_{i_A} a_{i_A} L_{i_A} \right], \quad [x_B] = \left[ \sum_{i_B} b_{i_B} L_{i_B} \right], \quad [x_C] = \left[ \sum_{i_C} c_{i_C} L_{i_C} \right]
\end{align}
and we have $H_1 = H_2$. Then $k_1 = k_2$ and $k_1^T = k_2^T$. Similarly to the case of toric codes, the logical representatives $L_i$ can be divided into two groups such that
\begin{align}
	L_i^H \smile L_j^V = \tau_{ij} = L_j^V \smile L_i^H, \quad 1 \leq i,j \leq \frac{\dim H^1(X,\mathbb{F}_q)}{2} = \frac{k}{2}, 
\end{align}
where $\tau_{ij}$ is a 2-cube. 

Assume $n$ is even, we can take $\xi$ to be supported on each 2-cube. As a result,
\begin{align}
	\int_{\xi} x_A \smile x_B = \sum_{1 \leq i_A,j_A,i_B,j_B \leq k/2} a_{i_A} b_{j_B} + a_{j_A} b_{i_B}, \quad a_{i_A}, b_{j_B}, a_{j_A}, b_{i_B} \in \mathbb{F}_q.
\end{align}
and the global constraints are
\begin{align}
	\sum_{1 \leq i_A,j_A,i_B,j_B \leq k/2} a_{i_A} b_{j_B} + a_{j_A} b_{i_B} = 0, \\
	\sum_{1 \leq i_B,j_B,i_C,j_C \leq k/2} b_{i_B} c_{j_C} + b_{j_B} c_{i_C} = 0, \\
	\sum_{1 \leq i_A,j_A,i_C,j_C \leq k/2} a_{i_A} c_{j_C} + a_{j_A} c_{i_C} = 0.
\end{align}
It is however intricate to obtain a closed-form formula for the gauged code dimension of general hypergraph products. Nevertheless, a simple counting argument indicates that the gauged code dimension has the same scaling as that of the original CSS code. To be precise, given $\dim H^1(X,\mathbb{F}_q) = k_1^T k_2 + k_1 k_2^T = k$, suppose $k_1^T k_2 = \Theta(k)$ without loss of generality. Then we merely take $x_A,x_B,x_C$ to lie in the span of logical representatives in $\mathbb{F}_q^{N_1} \otimes \mathbb{F}_q^{M_2}$. Then the cup product must always vanish and the gauge condition holds automatically. Therefore, there are at least $(q^{k_1^T k_2})^3$ gauged tuples of cohomological classes in $L_G$ and the gauged code dimension is $\geq (q^{k_1^T k_2})^3 = q^{\Theta(k)}$. The same argument is used in the setting of gauged sheaf codes to achieve good coding rate.


\section{Gauging sheaf codes}\label{sec:gauged_sheaf}

It is natural to generalize the dressed $X$ stabilizer from HGP codes to sheaf codes as  
\begin{align}
	& \textcolor{red}{\mathcal{A}_{\alpha v}^A} 
	= X_{\alpha \delta v}^A \sum_{x_B, x_C \in C^1(X,\F)}
	(-1)^{ \tr_{\mathbb{F}_q/\mathbb{F}_2}( \alpha \int_\xi v \smile x_B \smile x_C) }
	\ket{x_B,x_C} \bra{x_B,x_C}, \\
	& \textcolor{green}{\mathcal{A}_{\alpha v}^B} 
	= X_{\alpha \delta v}^B \sum_{x_A, x_C \in C^1(X,\F)}
	(-1)^{ \tr_{\mathbb{F}_q/\mathbb{F}_2}( \alpha \int_\xi x_A \smile v \smile x_C) }
	\ket{x_A,x_C} \bra{x_A,x_C},  \\
	& \textcolor{blue}{\mathcal{A}_{\alpha v}^C} 
	= X_{\alpha \delta v}^C \sum_{x_A, x_B \in C^1(X,\F)}
	(-1)^{ \tr_{\mathbb{F}_q/\mathbb{F}_2}( \alpha \int_\xi x_A \smile x_B \smile v) }
	\ket{x_A,x_B} \bra{x_A,x_B}, 
\end{align}
where $\alpha \in \mathbb{F}_q$ and $v \in C^0(X,\F)$ is a standard basis vector that is supported on a single 0-cube with a local standard basis vector in $\mathbb{F}_q^{m_1 \times m_2}$. However, before we initiate the generalization, we have to take a close look at the cup product $v \smile x_B \smile x_C$. It is defined by the map
\begin{align}
	C^0(X,\F) \times C^1(X,\F) \times C^1(X,\F) \rightarrow C^2(X, \F^{\otimes 3})  
\end{align}
analogous to \ref{eq:1-cube_cup_sheaf}. Recall that $\F^{\otimes 3}$ is defined by $\F_{\sigma,\tau}^{\otimes 3}: \F_\sigma^{\otimes 3} \rightarrow \F_\tau^{\otimes 3}$. On the other hand, for a well-defined pairing or integral in Eq.~\eqref{eq:paring}, $\xi$ must be a 2-cycle from $C_2(X, \F^{\otimes 3})$. Let $h_i \in \mathbb{F}_q^{m_i \times n}$ with $i = 1,2,3$ be the local parity-check matrices that define $\F$. For each $i = 1,2,3$, let 
\begin{align}
	h_i \ast h_i \ast h_i = h_i^{\ast 3}: \mathbb{F}_q^{n} \rightarrow \mathbb{F}_q^{m_i^3}
\end{align}
be obtained by taking column-wise tensor products (Khatri--Rao product in Definition \ref{def:KR_product}). Solving for $\xi$ needs $\ker  h_i^{\ast 3} \neq 0$ \cite{LSWLL2026Theory}. However, to achieve an (almost) good qLDPC code, a simple yet necessary condition is that the local codes $\ker h_i$ have to have good local distance and this requires $m_i$ to be large enough. Usually, $m_i^3 = O(n^3)$ and it is unlikely that $\ker  h_i^{\ast 3}$ can still survive. To be more precise, the framework needs two-way product-expanding local classical codes \cite{PK2023RobustlyTestable,DHLV2022,Dinur2024sheaf,KP2025Extendable} which are proven to exist by probabilistic methods and it is unclear whether or not $\ker h_i^{\ast 3}$ is trivial in the random sense. A recent work \cite{LLL2026nontrivial} demonstrates the product-expansion on punctured Reed--Solomon codes, whose parity-check matrices admit a Khatri--Rao product with nontrivial kernel. In order to achieve good qLTC codes, projective Reed--Solomon codes were anticipated to be useful \cite{Panteleev2024,KP2025Extendable} and demonstrated in \cite{GJ2026,BLN2026,CHLT2026}. They also support the existence of $\xi$ in the good qLTC setting \cite{LLL2609}. We will use these RS codes in the following construction of (almost) good qLDPC non-Abelian codes.   


\subsection{Gauged code space and dimension}\label{sec:k_sheaf}

We prepare $t+1$ pieces of $t$-dimensional (co)chain complexes: 
\begin{align}\label{eq:t_sheaf_codes}
	\begin{aligned}
		& C^0(X,\F) \xrightarrow{\delta^0} C^1(X,\F) \xrightarrow{\delta^1} C^2(X,\F) \xrightarrow{\delta^2} \cdots \xrightarrow{\delta^{t-1}} C^t(X,\F), \\
		& \qquad \vdots \\
		& C^0(X,\F) \xrightarrow{\delta^0} C^1(X,\F) \xrightarrow{\delta^1} C^2(X,\F) \xrightarrow{\delta^2} \cdots \xrightarrow{\delta^{t-1}} C^t(X,\F), \\
		& C^0(X,\mathcal{G}) \xrightarrow{\delta^0} C^1(X,\mathcal{G}) \xrightarrow{\delta^1} C^2(X,\mathcal{G}) \xrightarrow{\delta^2} \cdots \xrightarrow{\delta^{t-1}} C^t(X,\mathcal{G}),
	\end{aligned}
\end{align}
where the first $t$ pieces are identical and the last one uses a different sheaf $\mathcal{G}$. Both $\F$ and $\mathcal{G}$ are generated by two-way product expanding punctured RS codes \cite{LLL2026nontrivial}, but in different ways. This is due to the necessity to preserve the global coding rate explained below. Gauging good qLTCs with projective RS codes \cite{GJ2026,BLN2026,CHLT2026} exhibits the same pattern. For simplicity, we outline the process at the end of this subsection. 

For now, let $\{\C_i\}_{i=1}^t$ be a collection of two-way product expanding punctured RS codes of length $n$. Let $\dim \C_i = m_i \leq \frac{n}{2t}$ and let $h_i \in \mathbb{F}_q^{m_i \times n}$ be the generator matrices of $\C_i$, i.e., $\C_i = \Ima h_i^T$. They are Vandermonde matrices and are used to generate $\F$. As for $\mathcal{G}$, it is generated by $(h_1',\ldots,h_{t-1}',h_t)$ where $h_i'$ is the generator matrix of  
\begin{align}
	\C_i' \coloneq (\C_i^{\ast t})^\perp, \ i=1,\ldots,t-1
\end{align}
where $\C_i^{\ast t}$ is the code generated by the $t$-fold Schur product of $h_i$ and $(\C_i^{\ast t})^\perp$ is its dual. Given $m_i < \frac{n}{2t}$ and the fact that $h_i$ are Vandermonde matrices, $h_i^{\ast t}$ is still a Vandermonde matrix with repeated rows and its rank is $< \frac{n}{2}$. The dual code and thus the generator matrix $h_i'$ have rank $m_i' > \frac{n}{2}$. To summarize, $\F$ and $\mathcal{G}$ are defined by
\begin{align}
	(h_1,\ldots,h_{t-1},h_t), \quad (h_1',\ldots,h_{t-1}',h_t)
\end{align}
with ranks
\begin{align}\label{eq:local_rank}
	m_1,\ldots,m_{t-1},m_t < \frac{n}{2t}, \quad m_1',\ldots,m_{t-1}' > \frac{n}{2}.
\end{align}
For simplicity, we take $m_i \equiv m < \frac{n}{2t}$ and $m_i' \equiv m' > \frac{n}{2}$. With product-expanding local codes, it is proved in \cite{DHLV2022,Dinur2024sheaf} that the first $t$ identical copies of sheaf codes have parameters
\begin{align}\label{eq:sheaf_code_parameter_1}
	[\![ N, \Omega(1), \Theta(N/\mathrm{polylog}\,N)  ]\!]
\end{align}
and the last one has almost good parameters
\begin{align}\label{eq:sheaf_code_parameter_2}
	[\![ N, \Theta(N), \Theta(N/\mathrm{polylog}\,N)  ]\!]
\end{align}
If $t = 2$, the distance becomes optimal $d = \Theta(N)$. Good qLTCs are achieved via distinct cubical complexes \cite{GJ2026,BLN2026,CHLT2026}. The tensor product sheaf $\F^{\otimes t} \otimes \ \mathcal{G}$ is generated by
\begin{align}
	(h_1)^{\ast t} \ast h_1', \cdots, (h_{t-1})^{\ast t} \ast h_{t-1}', \ (h_t)^{\ast (t+1)}.
\end{align}
The last matrix has rank $< n$ and thus $\ker (h_t)^{\ast (t+1)} \neq 0$. The others are products between codes and dual codes; the resulting kernel must contain the all-ones vector. All these facts ensure the existence of a nontrivial cycle $\xi \in C_t(X,\F^{\otimes t} \otimes \ \mathcal{G})$ and it can be chosen to support on every $t$-cube of $X$ \cite{LLL2026nontrivial}.

The dressed $X$-stabilizers are defined below: for $1 \leq i \leq t$,
\begin{align}\label{eq:A_i_sheaf}
	\mathcal{A}_{\alpha v}^{(i)} = X_{\alpha \delta v}^{(i)} \sum_{ \substack{x^{(j)} \in C^1(X,\F), j \neq i, 1 \leq j \leq t \\ x^{(t+1)} \in C^1(X,\mathcal{G})} } 
	& (-1)^{ \tr_{\mathbb{F}_q/\mathbb{F}_2}( \left\langle x^{(1)} \smile \cdots \smile x^{(i-1)} \smile \alpha v \smile  x^{(i+1)} \smile \cdots \smile x^{(t+1)}, \ \xi \right\rangle )}  \\
	& \ket{ x^{(1)},\ldots,x^{(i-1)},x^{(i+1)},\ldots, x^{(t+1)} } \bra{ x^{(1)},\ldots,x^{(i-1)},x^{(i+1)},\ldots, x^{(t+1)} }, \notag 
\end{align}
where $v \in C^0(X,\F)$ is defined on a single 0-cube with a local standard basis vector in $\mathbb{F}_q^{m_1 \times \cdots \times m_{t}}$. For $i = t+1$,
\begin{align}\label{eq:A_t_sheaf} 
	\mathcal{A}_{\alpha v}^{(t+1)} = X_{\alpha \delta v}^{(t+1)} \sum_{x^{(1)},\ldots,x^{(t)} \in C^1(X,\F)}
	& (-1)^{ \tr_{\mathbb{F}_q/\mathbb{F}_2}( \left\langle x^{(1)} \smile \cdots \smile x^{(t)} \smile \alpha v, \ \xi \right\rangle )}
	\ket{ x^{(1)},\ldots, x^{(t)} } \bra{ x^{(1)},\ldots,x^{(t)}}, 
\end{align} 

As for the weights of these stabilizer generators, for example, $\mathcal{A}_v^{(1)}$ with $v$ based on some $[g;0,\ldots,0]$ has $t! n^t$ nontrivial cup products with $t$ pieces of $1$-cubes. Taking local coefficients into consideration (see Section \ref{sec:cubical}), there are at most $(t-1) m^{t-1} + m'^{t-2} m$ degrees of freedom on $x^{(2)}, \ldots, x^{(t)},x^{(t+1)}$ and a local $t$-controlled-$Z$ circuit would act on physical qudits indexed by these 1-cubes with local coefficients Regardless of the choice of $\xi$, since $n,m,m',t$ are constants, the stabilizer weights are bounded. Conversely, since each physical qudit corresponds to a coordinate in $x^{(i)}$, each of them is involved in a constant number of cup products that depend only on $n,m,m',t$ again. The $Z$-stabilizers are undressed and the physical qudits are decomposed into a constant number of qubits because $q$ is a constant. As a result, the non-Abelian code is an LDPC code.

\begin{lemma}\label{lemma:coefficient}
	Let $\xi$ be a fixed nontrivial cycle in $C_t(X,\F^{\otimes t} \otimes \mathcal{G})$ and let $\mathbf{x} = (x^{(i)})$ with $i \in [t+1]$ be fixed logical representatives of the sheaf code blocks. For arbitrary vectors $\mathbf{u} = (u^{(i)}) \in C^0(X,\F) \times \cdots \times C^0(X,\F) \times C^0(X,\mathcal{G})$, we define the function 
	\begin{align}
		c_{\mathbf{x}}^\xi (\mathbf{u}) = (-1)^{ \tr_{\mathbb{F}_q/\mathbb{F}_2} (F_{\mathbf{x}}^\xi (\mathbf{u}))} \in \mathbb{C}
	\end{align}
	by
	\begin{align}
		\begin{aligned}
			F_{\mathbf{x}}^\xi(\mathbf{u}) = & \sum_{i \in [t+1]} \left\langle x^{(1)} \smile \cdots \smile x^{(i-1)} \smile u^{(i)} \smile  x^{(i+1)} \smile \cdots \smile x^{(t+1)}, \ \xi \right\rangle \\
			& +  \sum_{i < j \in [t+1]} \left\langle x^{(1)} \smile \cdots \smile u^{(i)} \smile \cdots \smile \delta u^{(j)}  \smile \cdots \smile x^{(t+1)}, \ \xi \right\rangle \\
			& + \cdots + \left\langle u^{(1)} \smile \delta u^{(2)} \smile \cdots \smile \delta u^{(t+1)}, \ \xi \right\rangle. 
		\end{aligned}
	\end{align}
	Then the following identities hold:
	\begin{enumerate}
		\item For any $w^{(i)} \in C^0(X,\F)$ with $1 \leq i \leq t$ or any $w^{(t+1)} \in C^0(X,\mathcal{G})$, 
		\begin{align}\label{eq:coefficient_1}
			\begin{aligned}
				& F_{\mathbf{x}}^\xi(u^{(1)},\ldots, u^{(i)} + w^{(i)}, \ldots, u^{(t+1)}) \\
				= & F_{\mathbf{x}}^\xi(\mathbf{u}) + \left\langle  (x^{(1)} + \delta u^{(1)} ) \smile \cdots \smile w^{(i)}  \smile \cdots \smile (x^{(t+1)} + \delta u^{(t+1)} ), \ \xi \right\rangle.
			\end{aligned}
		\end{align}
		
		\item For equivalent representatives $\mathbf{x} = (x^{(i)} ) \sim \mathbf{x}' = (x^{(i)} + \delta w^{(i)})$, let $\mathbf{w} = (w^{(i)} )$. Then
		\begin{align}\label{eq:coefficient_2}
			c_{\mathbf{x}}^\xi (\mathbf{u}) = (-1)^{\tr_{\mathbb{F}_q/\mathbb{F}_2} (F_{\mathbf{x}}^\xi (\mathbf{w}))} c_{\mathbf{x}'}^\xi (\mathbf{u} + \mathbf{w}).
		\end{align}
		Here $(-1)^{\tr_{\mathbb{F}_q/\mathbb{F}_2} (F_{\mathbf{x}}^\xi (\mathbf{w}))}$ is a global phase factor independent of $\mathbf{u}$.
	\end{enumerate}
\end{lemma}

Again, we abbreviate $\delta^0 u^{(i)}$ by $\delta u^{(i)}$. As a reminder, Lemma \ref{lemma:coefficient} and Theorem \ref{thm:code_space_sheaf} in the following hold for general choices of sheaves as long as $\xi$ exists. One can set $\F = \mathcal{G}$, ask them to be trivial or be generated by the two-way product expanding local codes. A special choice of sheaves in the 2-dimensional case is made in Section \ref{sec:m_sheaf}. Different choices of sheaves influence the code dimension and distance as will be revealed in the following results. 

\begin{proof}[Proof of Lemma \ref{lemma:coefficient}]
	We consider three pieces of $2$-cochain complexes for simplicity:
	\begin{align}
		\begin{aligned}
			& F_{\mathbf{x}}^\xi(\mathbf{u}) = F_{\mathbf{x}}^\xi(u_A,u_B,u_C) \\
			= & \left\langle u_A \smile x_B \smile x_C, \ \xi \right\rangle + \left\langle x_A \smile u_B \smile x_C, \ \xi \right\rangle + \left\langle x_A \smile x_B \smile u_C, \ \xi \right\rangle \\
			& + \left\langle u_A \smile \delta u_B \smile x_C, \ \xi \right\rangle 
			+ \left\langle x_A \smile u_B \smile \delta u_C, \ \xi \right\rangle
			+ \left\langle \delta u_A \smile x_B \smile u_C, \ \xi \right\rangle \\
			& + \left\langle u_A \smile \delta u_B \smile \delta u_C, \ \xi \right\rangle 
		\end{aligned}
	\end{align}
	Or equivalently, as integrals:
	\begin{align}
		\begin{aligned}
			& \int_{\xi} u_A \smile x_B \smile x_C + \int_{\xi} x_A \smile u_B \smile x_C + \int_{\xi} x_A \smile x_B \smile u_C \\
			& + \int_{\xi} u_A \smile \delta u_B \smile x_C + \int_{\xi} x_A \smile u_B \smile \delta u_C +  \int_{\xi} \delta u_A \smile x_B \smile u_C \\
			& + \int_{\xi} u_A \smile \delta u_B \smile \delta u_C. 
		\end{aligned}
	\end{align}
	By the Leibniz rules and Stokes' theorem in Proposition \ref{prop:cup_cap}, we utilize a discrete version of integration by parts to obtain
	\begin{align}\label{eq:swap_delta}
		\int_{\xi} u_A \smile \delta u_B \smile \delta u_C = \int_{\xi} \delta u_A \smile  u_B \smile \delta u_C =	\int_{\xi} \delta u_A \smile \delta u_B \smile u_C.
	\end{align}
	Let $w_C \in C^0(X,\mathcal{G})$. Since the cup product is multi-linear, by Eq.~\eqref{eq:swap_delta}, we have
	\begin{align}
		\begin{aligned}
			F_{\mathbf{x}}^\xi (u_A,u_B,u_C + w_C) = &
			\int_{\xi} x_A \smile x_B \smile w_C
			+ \int_{\xi} x_A \smile u_B \smile \delta w_C + \int_{\xi} \delta u_A \smile x_B \smile w_C \\
			& + \int_{\xi} u_A \smile \delta u_B \smile \delta w_C + F_{\mathbf{x}}^\xi(u_A,u_B,u_C) \\
			= & \int_{\xi} (x_A + \delta u_A) \smile (x_B + \delta u_B) \smile w_C + F_{\mathbf{x}}^\xi(u_A,u_B,u_C).
		\end{aligned}
	\end{align}
	Similar arguments hold for $w_A, w_B$.
	
	Now, let $(x_A, x_B, x_C) \sim (x_A + \delta w_A, x_B + \delta w_B, x_C + \delta w_C)$ for some $(w_A,w_B, w_C) \in C^0(X,\F) \times C^0(X,\F) \times C^0(X,\mathcal{G})$. By Eq.~\eqref{eq:coefficient_1},
	\begin{align}
		& F_{\mathbf{x}'}^\xi(u_A + w_A, u_B + w_B, u_C + w_C) \notag \\
		= & \int_{\xi} (u_A + w_A) \smile (x_B + \delta w_B) \smile (x_C + \delta w_C) + \int_{\xi} (x_A + \delta w_A) \smile (u_B + w_B) \smile (x_C + \delta w_C) \notag \\
		& + \int_{\xi} (x_A + \delta w_A) \smile (x_B + \delta w_B) \smile (u_C + w_C) \\
		& + \int_{\xi} (u_A + w_A) \smile \delta (u_B + w_B) \smile (x_C + \delta w_C) + \int_{\xi} (x_A + \delta w_A) \smile (u_B + w_B) \smile \delta (u_C + w_C) \notag \\
		& + \int_{\xi} \delta (u_A + w_A) \smile (x_B + \delta w_B) \smile (u_C + w_C) 
		+ \int_{\xi} (u_A + w_A) \smile \delta (u_B + w_B) \smile \delta (u_C + w_C). \notag 
	\end{align}
	The expansion equals 
	\begin{align}
		\begin{aligned}
			& \int_{\xi} w_A \smile x_B \smile x_C + \int_{\xi} x_A \smile w_B \smile x_C + \int_{\xi} x_A \smile x_B \smile w_C \\
			& + \int_{\xi} w_A \smile \delta w_B \smile x_C + \int_{\xi} x_A \smile w_B \smile \delta w_C +  \int_{\xi} \delta w_A \smile x_B \smile w_C \\
			& + \int_{\xi} w_A \smile \delta w_B \smile \delta w_C +  F_{\mathbf{x}}^\xi(u_A,u_B,u_C) 
			= F_{\mathbf{x}}^\xi(w_A,w_B,w_C) + F_{\mathbf{x}}^\xi(u_A,u_B,u_C).
		\end{aligned}
	\end{align}
	Since $\tr_{\mathbb{F}_q/\mathbb{F}_p}$ is linear, 
	\begin{align}
		c_{\mathbf{x}}^\xi (\mathbf{u}) =  (-1)^{\tr_{\mathbb{F}_q/\mathbb{F}_2} (F_{\mathbf{x}}^\xi (\mathbf{w}))} c_{\mathbf{x}'}^\xi (\mathbf{u} + \mathbf{w}).
	\end{align}
	Here the phase depends on $\mathbf{x}$ and $\mathbf{w}$, but is independent of $\mathbf{u}$.
\end{proof}

\begin{definition}\label{def:gauged_reps}
	Let $\boldsymbol{\eta} = (\eta^{(i)}) \in C^0(X,\F) \times \cdots \times C^0(X,\F) \times C^0(X,\mathcal{G})$ be 0-cocycles. Let $\mathbf{x} = (x^{(i)})$ with $i \in [t+1]$ be logical representatives of the sheaf code blocks. Then $\mathbf{x}$ is said to be \emph{gauged} if for any $i$ and for any $\boldsymbol{\eta}$,
	\begin{align}\label{eq:gauged_reps}
		\left\langle x^{(1)} \smile \cdots \smile x^{(i-1)} \smile \eta^{(i)} \smile  x^{(i+1)} \smile \cdots \smile x^{(t+1)}, \ \xi \right\rangle = 0.
	\end{align}
\end{definition}

In Section \ref{sec:k_HGP}, for HGP codes based on connected graphs with a constant sheaf of scalar coefficients, $\ker \delta^0$ is $1$-dimensional and spanned by $\eta$ that supports on every single vertex of the complex. By Eq.~\eqref{eq:cup_01}, $\eta \smile x = x \smile \eta = x$ for any $x \in C^1(X,\mathbb{F}_q)$. Explicitly, for $3$ copies of a 2-cochain complex, Eq.~\eqref{eq:gauged_reps} degenerates to Eq.~\eqref{eq:gauged_reps_HGP}:
\begin{align}
	\int_{\xi} x_B \smile x_C, \quad \int_{\xi} x_A \smile x_B, \quad \int_{\xi} x_A \smile x_C = 0.
\end{align}

As before, a \emph{tuple of gauged cohomological classes} $[\mathbf{x}] \in H^1(X,\F) \times \cdots \times H^1(X,\F) \times H^1(X,\mathcal{G})$ is well-defined. Let
\begin{align}
	L_G \coloneq \Big\{ [\mathbf{x}] \in H^1(X,\F) \times \cdots \times H^1(X,\F) \times H^1(X,\mathcal{G}): [\mathbf{x}] \text{ is gauged} \Big\}.
\end{align}
For any tuple $\mathbf{x}$ of explicit gauged representatives, let $\mathbf{u} = (u^{(i)})$ be an arbitrary vector from $C^0(X,\F) \times \cdots \times C^0(X,\F) \times C^0(X,\mathcal{G})$ and let
\begin{align}\label{eq:code_state}
	\begin{aligned}
		\ket{\Psi_{\mathbf{x}}^\xi} \coloneq & \mathcal{N} \sum_{\mathbf{u} \in C^0(X,\F) \times \cdots \times C^0(X,\F) \times C^0(X,\mathcal{G})} c_{\mathbf{x}}^\xi (\mathbf{u}) 
		\Big\vert \cdots, x^{(i)} + \delta u^{(i)} , \cdots \Big\rangle
	\end{aligned}
\end{align}
where $\mathcal{N}$ is the normalization factor and $c_{\mathbf{x}}^\xi (\mathbf{u}) = (-1)^{ \tr_{\mathbb{F}_q/\mathbb{F}_2} (F_{\mathbf{x}}^\xi (\mathbf{u}))}$. The following theorem holds.

\begin{theorem}\label{thm:code_space_sheaf}
	Quantum states $\ket{\Psi_{\mathbf{x}}^\xi}, \ket{\Psi_{\mathbf{x}'}^\xi}$ induced by equivalent gauged representatives $\mathbf{x} \sim \mathbf{x}'$ only differ by a phase factor. Moreover, up to phase factors, tuples of gauged cohomological classes induce an orthonormal basis $\{ \ket{\Psi_{\mathbf{x}}^\xi} \}$ of the non-Abelian code defined by the above gauged $X$-stabilizers and undressed $Z$-stabilizers. The code dimension equals the size of $L_G$.
\end{theorem}
\begin{proof}
	We still consider three pieces of $2$-cochain complexes. First, we have to show that Eq.~\eqref{eq:code_state} is well-defined. If $\eta_A, \eta_B, \eta_C$ are $0$-cocycles, i.e., $\delta \eta_i = 0$, then
	\begin{align}
		\ket{x_A + \delta u_A, x_B + \delta u_B, x_C + \delta u_C} = \ket{x_A + \delta u_A + \delta \eta_A, x_B + \delta u_B + \delta \eta_B, x_C + \delta u_C + \delta \eta_C},
	\end{align}
	To avoid ambiguity, the coefficients have to satisfy
	\begin{align}\label{eq:code_state_compatible}
		c_{\mathbf{x}}^\xi (\mathbf{u}) = c_{\mathbf{x}}^\xi (\mathbf{u} + \boldsymbol{\eta})  
	\end{align}
	for $0$-cocycles. By Lemma \ref{lemma:coefficient},
	\begin{align}\label{eq:eta_expansion}
		& F_{\mathbf{x}}^\xi (u_A + \eta_A, u_B + \eta_B, u_C + \eta_C) \notag \\
		= &	F_{\mathbf{x}}^\xi (u_A + \eta_A, u_B + \eta_B, u_C) + \int_\xi  (x_A + \delta (u_A+\eta_A) ) \smile (x_B + \delta (u_B+\eta_B) ) \smile \eta_C \notag \\
		= & F_{\mathbf{x}}^\xi (u_A + \eta_A, u_B, u_C) 
		+ \int_\xi (x_A + \delta (u_A+\eta_A) ) \smile \eta_B \smile (x_C + \delta u_C ) \notag \\
		& + \int_\xi  (x_A + \delta (u_A+\eta_A) ) \smile (x_B + \delta (u_B+\eta_B) ) \smile \eta_C \\
		= & F_{\mathbf{x}}^\xi (u_A, u_B, u_C) 			
		+ \int_\xi \eta_A \smile (x_B + \delta u_B) \smile (x_C + \delta u_C ) 
		+ \int_\xi (x_A + \delta (u_A+\eta_A) ) \smile \eta_B \smile (x_C + \delta u_C ) \notag \\
		& + \int_\xi  (x_A + \delta (u_A+\eta_A) ) \smile (x_B + \delta (u_B+\eta_B) ) \smile \eta_C. \notag
	\end{align}
	Provided that $\boldsymbol{\eta}$ consists of 0-cocycles, it can be further simplified:
	\begin{align}
		\begin{aligned}
			& F_{\mathbf{x}}^\xi (u_A, u_B, u_C) 			
			+ \int_\xi \eta_A \smile (x_B + \delta u_B) \smile (x_C + \delta u_C )
			+ \int_\xi (x_A + \delta u_A ) \smile \eta_B \smile (x_C + \delta u_C ) \\
			& + \int_\xi  (x_A + \delta u_A ) \smile (x_B + \delta u_B ) \smile \eta_C \\
			= \ & F_{\mathbf{x}}^\xi (u_A, u_B, u_C) + \int_\xi \eta_A \smile x_B \smile x_C + \int_\xi x_A  \smile \eta_B \smile x_C + \int_\xi x_A \smile x_B \smile \eta_C \\
		\end{aligned}
	\end{align}
	Since $\eta_A, \eta_B, \eta_C$ can be taken independently and since the trace map is nondegenerate, Eq.~\eqref{eq:code_state_compatible} is equivalent to
	\begin{align}
		\int_\xi \eta_A \smile x_B \smile x_C, \quad \int_\xi x_A  \smile \eta_B \smile x_C, \quad \int_\xi x_A \smile x_B \smile \eta_C = 0
	\end{align}
	Its generalization is exactly Eq.~\eqref{eq:gauged_reps}.
	
	Now, we justify that the states $\ket{\Psi_{\mathbf{x}}^\xi}$ are indeed common $+1$ eigenstates of the stabilizers. Since $\mathbf{x}$ consists of 1-cocycles, $\ket{\Psi_{\mathbf{x}}^\xi}$ must be invariant under the action of $Z$ stabilizers. For any dressed $X$-stabilizer,
	\begin{align}
		\begin{aligned}
			& \mathcal{A}_{\alpha v}^A \sum_{u_A,u_B,u_C} c_{\mathbf{x}}^\xi \ket{x_A + \delta u_A, x_B  + \delta u_B, x_C + \delta u_C} \\
			= & \sum_{u_A,u_B,u_C} c_{\mathbf{x}}^\xi (-1)^{ \tr_{\mathbb{F}_q/\mathbb{F}_2}( \alpha \int_\xi v \smile (x_B + \delta u_B) \smile (x_C +  \delta u_C )) }
			\ket{x_A + \delta u_A + \alpha \delta v, x_B + \delta u_B, x_C + \delta u_C}.
		\end{aligned}
	\end{align}
	On the other hand, the original phase of $\ket{x_A + \delta u_A + \alpha \delta v, x_B + \delta u_B, x_C + \delta u_C}$ is $F_{\mathbf{x}}^\xi (u_A + \alpha v,u_B,u_C)$. By Lemma \ref{lemma:coefficient}, 
	\begin{align}
		\begin{aligned}
			F_{\mathbf{x}}^\xi (u_A + \alpha v,u_B,u_C) = F_{\mathbf{x}}^\xi (u_A ,u_B,u_C) + \alpha \int_\xi v \smile (x_B + \delta u_B) \smile (x_C +  \delta u_C ).
		\end{aligned}
	\end{align}
	Together with the fact that $\tr_{\mathbb{F}_q/\mathbb{F}_p}$ is linear, we see that the state $\ket{\Psi_{\mathbf{x}}^\xi}$ is invariant under the action of $\mathcal{A}_{\alpha v}^A$. Similar arguments can justify the invariance under $\mathcal{A}_{\alpha v}^B$ and $\mathcal{A}_{\alpha v}^C$. 
	
	Conversely, suppose $\ket{\Psi}$ is a state in the non-Abelian code space. To be invariant under $Z$-stabilizers, $\ket{\Psi}$ must be a superposition of basis states given by 1-cocycles. Suppose $\ket{x_A',x_B',x_C'}$ is one term in the expansion. Let $u_A = \sum_v \alpha_v v \in C^0(X,\F)$. Then we define
	\begin{align}
		\mathcal{A}_{u_A}^A \coloneq \prod_v \mathcal{A}_{\alpha_v v}^A = X_{\delta u_A}^A \sum_{x_B \in C^1(X,\F), x_C \in C^1(X,\mathcal{G})}
		(-1)^{ \tr_{\mathbb{F}_q/\mathbb{F}_2}( \int_\xi u_A \smile x_B \smile x_C) } \ket{x_B,x_C} \bra{x_B,x_C}.
	\end{align}
	Acting with $\mathcal{A}_{u_A}^A$, we obtain
	\begin{align}
		& \mathcal{A}_{u_A}^A \ket{x_A', x_B', x_C'} = (-1)^{ \tr_{\mathbb{F}_q/\mathbb{F}_2}( \int_\xi u_A \smile x_B' \smile x_C') } \ket{x_A' + \delta u_A, x_B', x_C'}.
	\end{align}
	Similarly, 
	\begin{align}
		& \mathcal{A}_{u_B}^B \ket{x_A' + \delta u_A, x_B', x_C'} = (-1)^{ \tr_{\mathbb{F}_q/\mathbb{F}_2}( \int_\xi (x_A' + \delta u_A) \smile u_B \smile x_C') } \ket{x_A' + \delta u_A, x_B' + \delta u_B, x_C'}, \\
		& \mathcal{A}_{u_C}^C \ket{x_A' + \delta u_A, x_B' + \delta u_B, x_C'} = (-1)^{ \tr_{\mathbb{F}_q/\mathbb{F}_2}( \int_\xi (x_A' + \delta u_A) \smile (x_B' + \delta u_B)  \smile u_C) } \ket{x_A' + \delta u_A, x_B' + \delta u_B, x_C' + \delta u_C}. \notag
	\end{align}
	Without loss of generality, suppose the coefficient of $\ket{x_A',x_B',x_C'}$ in $\ket{\Psi}$ is $1$; then the invariance forces the state $\ket{x_A' + \delta u_A, x_B' + \delta u_B, x_C' + \delta u_C}$ to be in the expansion with coefficient equal to $c_{\mathbf{x}'}^{\xi}(\mathbf{u})$. Therefore, $\ket{\Psi}$ must lie in the span of states as in Eq.~\eqref{eq:code_state}.
	
	Now, suppose $(x_A, x_B, x_C) \sim (x_A + \delta w_A, x_B + \delta w_B, x_C + \delta w_C)$ for some $(w_A,w_B, w_C) \in C^0(X,\F) \times C^0(X,\F) \times C^0(X,\mathcal{G})$. We have 
	\begin{align}
		& \ket{\Psi_{\mathbf{x}}^\xi} = \sum_{u_A,u_B,u_C} c_{\mathbf{x}}^\xi \ket{x_A + \delta u_A, x_B  + \delta u_B, x_C + \delta u_C}, \\
		& \ket{\Psi_{\mathbf{x}'}^\xi} = \sum_{u_A',u_B',u_C'} c_{\mathbf{x}'}^\xi \ket{x_A' + \delta u_A', x_B'  + \delta u_B', x_C' + \delta u_C'}. 
	\end{align}
	To compare the states, we have to compare coefficients with respect to
	\begin{align}
		\ket{x_A + \delta u_A, x_B  + \delta u_B, x_C + \delta u_C} = \ket{x_A' + \delta u_A', x_B'  + \delta u_B', x_C' + \delta u_C'} \implies \delta u_i = \delta u_i' + \delta w_i
	\end{align}
	for $i = A,B,C$. By Lemma \ref{lemma:coefficient},
	\begin{align}
		c_{\mathbf{x}}^\xi (\mathbf{u}) = (-1)^{\tr_{\mathbb{F}_q/\mathbb{F}_2} (F_{\mathbf{x}}^\xi (\mathbf{w}))} c_{\mathbf{x}'}^\xi (\mathbf{u}') 
		\implies \ket{\Psi_{\mathbf{x}}^\xi} = (-1)^{\tr_{\mathbb{F}_q/\mathbb{F}_2} (F_{\mathbf{x}}^\xi (\mathbf{w}))} \ket{\Psi_{\mathbf{x}'}^\xi},
	\end{align}
	so the difference is a global phase. Together with the previous expansion of $\ket{\Psi}$, we demonstrate that all non-Abelian gauged states are spanned by the states in Eq.~\eqref{eq:code_state} for any explicit representatives of tuples in $L_G$. Since inequivalent representatives must admit distinct bit strings, they define an orthonormal basis for the non-Abelian code.
\end{proof}

To estimate the non-Abelian code dimension, we need to count the number of inequivalent cohomological representatives that satisfy \eqref{eq:gauged_reps}. To achieve a large code space, we need more trivial cup products. Compared with HGP codes in Section \ref{sec:k_HGP}, sheaf codes do not admit the canonical logical representatives through simple tensor products. Nevertheless, when the cubical complex $X$ and the underlying base graph $G_0$ are lifted via the cyclic group $\HH = C_l$, i.e., $\gamma_{(u,v)} \in C_l$ in \eqref{eq:1D_generator_action}, estimating the dimension can be achieved via the \emph{polarized logical representatives} proposed in \cite{LSWLL2026Theory}. As a brief introduction, for a 2-dimensional $X$, an associated sheaf code admits the following types of logical representatives polarized in two directions:
\begin{align}\label{eq:sheaf_basis}
	[v_1;a_1] \otimes \sum_{h,v_2,b_2} [(h,v_2);b_2], \qquad \sum_{h,v_1,b_1} [(h,v_1);b_1] \otimes [v_2;a_2],
\end{align}
where $[v_i;a_i]$ are 1-cubes while $[v_j;b_j]$ are 0-cubes as defined in Section \ref{sec:cubical}. We hide the local coefficient vectors for conciseness and $h \in C_l$ denotes a group element. Analogously to the HGP case, $[v_i;a_i]$ are one-hot vectors while the sums, e.g., $\sum_{h,v_2,b_2} [(h,v_2);b_2]$, are taken from the kernel of the 1-dimensional sheaved coboundary operator. The formal explanation of Eq.~\eqref{eq:sheaf_basis} uses the \emph{K\"{u}nneth formula} over semisimple modules \cite{Hatcher2015AT,Gallier2022}, but we do not go into the algebraic details. 

For our purpose, we simply note that the logical representatives are polarized in different directions and the cup products among logical representatives with the same polarization always vanish by Eq.~\eqref{eq:1-cube_cup}. This is an analogue of the HGP case in the general sheaf setting. In particular, when the local codes' ranks are taken as \eqref{eq:local_rank}, for the last code block in \eqref{eq:t_sheaf_codes}, there are $\Theta(N)$ logical representatives polarized in one direction and there are $\Omega(1)$ polarized logical representatives in the same direction for the first $t$ code blocks \cite{LSWLL2026Theory}. Their cup products equal zero and hence they are gauged. This proves the following corollary.

\begin{corollary}\label{coro:sheaf_dim}
	Take $t+1$ sheaf code family as \eqref{eq:t_sheaf_codes}. When the ranks of the local codes satisfy \eqref{eq:local_rank}, there is a constant $c \in (0,1)$ such that total number of gauged cohomological classes, and hence the gauged code dimension, is lower bounded by $q^{cN}$ as the system size $N$ increases. 
\end{corollary}

\begin{remark}\label{remark:2D_dim}
	In dimension two, to achieve good code distance (cf. \eqref{eq:sheaf_code_parameter_2}), we have to replace $C_l$ by non-Abelian finite groups \cite{PK2022Good,DHLV2022,QuantumTanner2022}. For a suitable fixed degree $n$, the classical arithmetic constructions of Ramanujan graphs provide a fixed finite connected $n$-regular graph $G_0$ and a family of lifts $\widetilde{G}_i$ such that each of them is Ramanujan. The groups \(\HH_i\) used for lifts are non-Abelian with odd cardinality and $|\HH_i| \to \infty$ as $i \to \infty$ \cite{Glasner03,Bekka2008}. This yields an asymptotically good qLDPC code family in two dimensions. Since $\mathrm{char}\mathbb{F}_q = 2$ and $|\HH_i|$ is odd, Maschke's theorem makes $\mathbb{F}_q[\HH_i]$ semisimple \cite{Goodman2009}. Then the K\"{u}nneth formula still holds and we have polarized logical representatives as in \eqref{eq:sheaf_basis} to support Corollary \ref{coro:sheaf_dim} in this case. There are also subtle details about the gauged code distance in two dimensions and we address them in Corollary \ref{coro:2D_dist}. 	
\end{remark}

Theorem~\ref{thm:code_space_sheaf} and the exponential dimension bound in Corollary~\ref{coro:sheaf_dim} extend to the arithmetic cubical complexes with projective Reed--Solomon local codes used for good qLTCs~\cite{GJ2026,BLN2026,CHLT2026}. For fixed $r \geq 3$, take the $r$ sheaves $\mathcal{F}_1,\ldots,\mathcal{F}_r$ and the receiving $2r$-cycle $\xi$ on the $2r$-dimensional complex $X$ of Ref.~\cite{LLL2609}, including its coefficient multiplication in each cup product. Keep the first $r-1$ code blocks in degree two and use the last block in degree three. The logical representatives $x^{(i)}$ then have degrees $3,2,\ldots,2$, and the $v^{(i)}, u^{(i)}$ have degrees $2,1,\ldots,1$. 

The coefficient identities and the proof of Theorem~\ref{thm:code_space_sheaf} use only multilinearity, the Leibniz rule, $\partial \xi = 0$, and the decomposition into coboundary orbits. These properties hold for the good qLTC sheaves and thus the proof still holds. The code dimension is $|L_G|$ in the qLTCs. To estimate the code dimension, a simpler argument holds when the notion of polarized logical representatives is absent: set all logical classes except the first to zero. Then every gauge condition retains a zero factor and thus
\begin{align}
	|L_G| \geq q^{\dim_{\mathbb{F}_q}H^2(X,\mathcal{F}_r)}=q^{\Omega(N)},
\end{align}
where $N$ is the total number of physical qubits after the degree change. As before, we can select local codes to make the last block has constant rate. This proves the same exponential dimension scaling without an abelian-lift decomposition.


\subsection{Gauged code distance}\label{sec:d_sheaf}

Let $d$ be the minimal distance of the sheaf code blocks. Let $\epsilon_\delta^{\F}(0), \epsilon_\delta^{\mathcal{G}}(0)$ be the 0-cocycle expansions in Definition \ref{def:expansions} of the cochain complexes of sheaves $\F$ and $\mathcal{G}$, respectively. Let $\epsilon$ be the minimum. We have the following bound on the non-Abelian code distance. 

\begin{theorem}\label{thm:code_distance_sheaf}
	The gauged sheaf code has distance
	\begin{align}
		d_{\text{gauge}} \geq c(t,n,m,m') \cdot \epsilon \cdot d
	\end{align}
	where $c(t,n,m,m')$ is a function that depends only on the node degree $n$, the dimensions $m,m'$ of the local coefficient spaces and the dimension $t$ of the cochain complex.
\end{theorem}
\begin{proof}[Proof of Theorem \ref{thm:code_distance_sheaf}]
	We only prove the 2-dimensional case as it generalizes easily to high dimensions. By Theorem \ref{thm:Pauli_KL}, we consider an arbitrary Pauli error $E$ that is specified by six $1$-(co)chains for $A,B,C$:
	\begin{align}
		E = (\mathcal{X}_A, \mathcal{Z}_A; \mathcal{X}_B, \mathcal{Z}_B; \mathcal{X}_C, \mathcal{Z}_C).
	\end{align}
	Recall the (co)cycle expansions in Definition \ref{def:expansions}:
	\begin{align}
		\epsilon_{\delta}(i) = \min_{x\in C^i\setminus\ker\delta^i} \frac{|\delta^i x|}{\operatorname{dist}(x,\ker\delta^i)},\quad
		\epsilon_{\partial}(i) = \min_{z\in C_i\setminus\ker\partial_i} \frac{|\partial_i z|}{\operatorname{dist}(z,\ker\partial_i)}.
	\end{align}	
	Since we only gauge the $X$-stabilizers, we will only use the cocycle expansion. Actually, in 2 dimensions, only small-set (co)boundary expansions are known \cite{DHLV2022} and (almost) good (co)cycle expansions are obtained in higher dimensions \cite{Dinur2024sheaf,GJ2026,BLN2026,CHLT2026}. For conciseness, we proceed with cocycle expansion; after the proof, we explain the changes needed when using small-set coboundary expansion.
	
	For local coefficients, let $m_0 = \max\{m,m'\}$. We are going to prove that when the Pauli errors have weights bounded by
	\begin{align}\label{eq:error_weight}
		\vert \mathcal{X}_A \vert, \ \vert \mathcal{X}_B \vert, \ \vert \mathcal{X}_C \vert < \frac{\epsilon}{16 n^2 m_0 } d, \qquad 
		\vert \mathcal{Z}_A \vert, \ \vert \mathcal{Z}_B \vert, \ \vert \mathcal{Z}_C \vert < \frac{1}{2m_0} d,
	\end{align}
	they are always detecable. Here, $\epsilon$ depends on the expansion property of the cubical complex $X$ as well as the local codes used to generate the sheaf. The factor $\frac{1}{16 n^2 m_0}$ accounts for the dimensions of the complex and the local coefficients. A valid upper bound for general $t$ dimensions is given in \eqref{eq:local_view_w_j}.
	
	Given any codeword from the original sheaf codes, we have
	\begin{align}\label{eq:error_expectation}
		\begin{aligned}
			& \bra{x_A', x_B', x_C'} E  \ket{x_A, x_B, x_C} \\
			= &  (-1)^{ \tr_{\mathbb{F}_q/\mathbb{F}_2}( \langle \mathcal{Z}_A, x_A \rangle + \langle \mathcal{Z}_B, x_B \rangle + \langle \mathcal{Z}_C, x_C \rangle )} 
			\bra{x_A', x_B', x_C'} x_A + \mathcal{X}_A, x_B + \mathcal{X}_B, x_C +  \mathcal{X}_C \rangle.
		\end{aligned}
	\end{align}
	Since $x_A,x_B,x_C$ are cohomological representatives, unless $\mathcal{X}_A,\mathcal{X}_B,\mathcal{X}_C$ are $1$-cocycles, the above inner product vanishes. If one of them belongs to $\ker \delta^1 \setminus \Ima \delta^0$, then the $X$-error has weight at least the original CSS code distance and is therefore ruled out by \eqref{eq:error_weight}. To analyze the situation with possibly smaller weight, we consider $X$-errors such that
	\begin{align}\label{eq:error_coboundary}
		\mathcal{X}_A = \delta^0 w_A, \quad \mathcal{X}_B = \delta^0 w_B, \quad \mathcal{X}_C = \delta^0 w_C.
	\end{align}
	That is, they are coboundaries. By the definition of cocycle expansion and the assumption in \eqref{eq:error_weight}, we can find $w_i$ such that
	\begin{align}
		\vert w_A \vert, \ \vert w_B \vert, \ \vert w_C \vert < \frac{1}{16 n^2 m_0} d.  
	\end{align}
	
	For any $E$ with these error bounds, we demonstrate below that it is possible to select inequivalent cohomological representatives $\{x_A\}$, $\{x_B\}$ and $\{x_C\}$ whose associated orthonormal gauged code basis in Theorem \ref{thm:code_space_sheaf} satisfies the KL condition. Accordingly, it holds for any orthonormal basis of the gauged code space. To be precise, for fixed $\mathcal{X}_i, w_i, \mathcal{Z}_j$ satisfying \eqref{eq:error_weight}, we take an arbitrary gauged tuple of logical representatives of $L_G$ in Theorem \ref{thm:code_space_sheaf}. We clean their supports in the following. Let $\ket{\Psi} = \ket{\Psi_{x_A,x_B,x_C}^\xi}$ and $\ket{\Psi'} = \ket{\Psi_{x_A',x_B',x_C'}^\xi}$. By Eq.~\eqref{eq:error_expectation} and \eqref{eq:error_coboundary},
	\begin{align}
		\bra{\Psi'} E \ket{\Psi} \text{ could be nonzero only if } (x_A,x_B,x_C) \sim (x_A',x_B',x_C').
	\end{align}
	With respect to a given orthonormal basis, this forces $(x_A,x_B,x_C) = (x_A',x_B',x_C')$ and thus
	\begin{align}\label{eq:error_expectation_2}
		\begin{aligned}
			\bra{\Psi'} E \ket{\Psi} = 
			\mathcal{N}^2 \sum_{u_A,u_B,u_C} &
			(-1)^{ \tr_{\mathbb{F}_q/\mathbb{F}_2}(  F_{\mathbf{x}}^\xi (u_A + w_A, u_B + w_B, u_C + w_C ) + F_{\mathbf{x}}^\xi (u_A,u_B,u_C)) }  \\
			& \cdot (-1)^{ \tr_{\mathbb{F}_q/\mathbb{F}_2}( \langle \mathcal{Z}_A, x_A + \delta u_A \rangle + \langle \mathcal{Z}_B, x_B + \delta u_B \rangle + \langle \mathcal{Z}_C, x_C + \delta u_C \rangle )}
		\end{aligned}
	\end{align}
	where the normalization factor $\mathcal{N}$ is given by the number of $C^0$ elements, independent of gauged logical states. The goal is to verify that the right-hand side is independent of $\mathbf{x}$, then Theorem \ref{thm:Pauli_KL} gives a lower bound on the code distance.
	
	To this end, the first step is to prove $\langle \mathcal{Z}_i,x_i \rangle = 0$ for $i = A,B,C$. Recall that in Section~\ref{sec:cubical} we defined the geometric support $\supp(x_i)$ of $x_i$ as the collection of 1-cubes with nontrivial local coefficient vectors. Then the inner product vanishes if $\supp(x_i) \cap \supp(\mathcal{Z}_i) = \emptyset$. This can be achieved if $x_i$ is cleaned out in a subset of $C^1(X,\F)$ or $C^1(X,\mathcal{G})$ of size
	\begin{align}
		\vert \supp(\mathcal{Z}_i) \vert \cdot m_0 < \frac{d}{2}.
	\end{align}
	Here, we used the fact that $\vert \supp(\mathcal{Z}_i) \vert \leq \vert \mathcal{Z}_i \vert$ as the local coefficient space has dimension $\leq m_0$ on each geometric site. For general $t$, we replace the upper bound in \ref{eq:error_weight} by $\vert \mathcal{Z}_i \vert < \frac{1}{2m_0^{t-1}} d$.
	
	For $X$-errors, we have the following expansion analogous to Eq.~\eqref{eq:eta_expansion}
	\begin{align}\label{eq:w_expansion}
		& F_{\mathbf{x}}^\xi (u_A + w_A, u_B + w_B, u_C + w_C) \\
		= & F_{\mathbf{x}}^\xi (u_A, u_B, u_C) 			
		+ \int_\xi w_A \smile (x_B + \delta u_B) \smile (x_C + \delta u_C ) 
		+ \int_\xi (x_A + \delta (u_A + w_A) ) \smile w_B \smile (x_C + \delta u_C ) \notag \\
		& + \int_\xi  (x_A + \delta (u_A + w_A) ) \smile (x_B + \delta (u_B + w_B) ) \smile w_C. \notag
	\end{align}
	Note that the components of $\boldsymbol{w} = (w_A, w_B, w_C)$ need not be 0-cocycles, and we further expand:
	\begin{align}
		\begin{aligned}
			& \int_\xi  (x_A + \delta (u_A + w_A) ) \smile (x_B + \delta (u_B + w_B) ) \smile w_C \\
			= & \int_\xi x_A  \smile x_B \smile w_C + \int_\xi x_A \smile u_B \smile \delta w_C
			+ \int_\xi x_A \smile \delta w_B \smile w_C \\
			& + \int_\xi u_A \smile x_B \smile \delta w_C + \int_\xi \delta u_A \smile u_B \smile \delta w_C + \int_\xi u_A \smile \delta w_B \smile \delta w_C \\
			& + \int_\xi \delta w_A \smile x_B \smile w_C + \int_\xi \delta w_A \smile \delta u_B \smile w_C + \int_\xi \delta w_A \smile \delta w_B \smile w_C,
		\end{aligned}
	\end{align}
	and
	\begin{align}
		\begin{aligned}
			& \int_\xi (x_A + \delta (u_A + w_A) ) \smile w_B \smile (x_C + \delta u_C ) \\
			= & \int_\xi x_A \smile w_B \smile x_C + \int_\xi u_A \smile \delta w_B \smile x_C + \int_\xi \delta w_A \smile w_B \smile x_C \\
			& + \int_\xi x_A \smile \delta w_B \smile u_C + \int_\xi u_A \smile \delta w_B \smile \delta u_C + \int_\xi \delta w_A \smile w_B \smile \delta u_C,
		\end{aligned}
	\end{align}
	and
	\begin{align}
		\begin{aligned}
			& \int_\xi w_A \smile (x_B + \delta u_B) \smile (x_C + \delta u_C ) \\
			= & \int_\xi w_A \smile x_B \smile x_C + \int_\xi \delta w_A \smile u_B \smile x_C 
			+ \int_\xi \delta w_A \smile x_B \smile u_C + \int_\xi \delta w_A \smile u_B \smile \delta u_C,
		\end{aligned}
	\end{align}
	where we applied integration by parts several times.
	
	Now, we select terms that involve both logical representatives and the errors (recall $\delta w_i = \mathcal{X}_i$):
	\begin{align}\label{eq:rep_error}
		\begin{aligned}
			& \int_\xi x_A \smile x_B \smile w_C + \int_\xi x_A \smile u_B \smile \mathcal{X}_C
			+ \int_\xi x_A \smile \mathcal{X}_B \smile w_C \\
			+ & \int_\xi u_A \smile x_B \smile \mathcal{X}_C + \int_\xi \mathcal{X}_A \smile x_B \smile w_C \\
			+ & \int_\xi x_A \smile w_B \smile x_C + \int_\xi u_A \smile \mathcal{X}_B \smile x_C + \int_\xi \mathcal{X}_A \smile w_B \smile x_C + \int_\xi x_A \smile \mathcal{X}_B \smile u_C \\
			+ & \int_\xi w_A \smile x_B \smile x_C + \int_\xi \mathcal{X}_A \smile u_B \smile x_C + \int_\xi \mathcal{X}_A \smile x_B \smile u_C. 
		\end{aligned}
	\end{align}
	If these terms vanish for any $(x_A,x_B,x_C)$ with respect to the given $\mathcal{X}_i, w_i$, then the proof is complete. The following conditions are sufficient:
	\begin{enumerate}
		\item The two-term cup products vanish:
		\begin{align}
			& x_A \smile w_B, \ x_A \smile \mathcal{X}_B = 0, \\
			& w_A \smile x_B, \ \mathcal{X}_A \smile x_B, \ x_B \smile w_C, \ x_B \smile \mathcal{X}_C = 0, \\
			& w_B \smile x_C, \ \mathcal{X}_B \smile x_C = 0.
		\end{align}
		Note that $x_i, \mathcal{X}_j, w_j$ are defined on cubes with local coefficients. Their cup products are computed in two steps: taking cup products among the underlying geometrical cubes and then processing the local coefficient vectors as in Eq.~\eqref{eq:cup_sheaf_2}. Although there are nontrivial local structures, we still have a rather simple method to trivialize these cup products: given the support $\supp(w_j)$, let $X_{\geq \supp(w_j)}(t)$ be the collection of $t$-cubes that contain at least one cube in the support of $w_j$. With local coefficients, we have $\vert \supp(w_i) \vert \leq \vert w_i \vert$. Then by \eqref{eq:error_weight},
		\begin{align}
			\vert X_{\geq \supp(w_j)}(2) \vert \leq \vert \supp(w_j) \vert \cdot n^2 < \frac{1}{16 n^2 m_0} d \cdot n^2.
		\end{align}
		By definition, $X_{\geq \supp(w_j)}(2)$ includes at most $4\vert X_{\geq \supp(w_j)}(2) \vert < \frac{1}{4 n^2 m_0} d \cdot n^2$ 1-cubes. As long as the supports of our logical representatives bypass all these 1-cubes, the cup products are zero. Since each $x_i$ would form cup products with $w_{i-1}, w_{i+1}$ when they are defined, the logical representatives have to vanish on a subset of $C^1(X,\F)$ or $C^1(X,\mathcal{G})$ of size bounded by
		\begin{align}\label{eq:cleaning_w_j}
			8 \cdot \vert X_{\geq \supp(w_j)}(2) \vert \cdot m_0^{t-1} < \frac{d}{2}. 
		\end{align}
		By a similar argument, we can also bound $\vert X_{\geq \supp(\mathcal{X}_j)}(2) \vert$. Since $\delta w_j = \mathcal{X}_j$, $X_{\geq \supp(\mathcal{X}_j)}(2) \subseteq X_{\geq \supp(w_j)}(2)$ and \eqref{eq:cleaning_w_j} suffices to determine the size of the region that needs to be cleaned.
		
		\item We still require, e.g., $x_A \smile u_B \smile \mathcal{X}_C = 0$ for arbitrary $u_B \in C^0(X,\F)$. Since $u_B$ is supported on 0-cubes,
		\begin{align}
			\supp( u_B \smile \mathcal{X}_C ) \subseteq \supp(\mathcal{X}_C)
		\end{align} 
		and it suffices when $x_A$ bypasses $\supp(\mathcal{X}_C)$, as already accounted for in the worst-case bound in \eqref{eq:cleaning_w_j}.
	\end{enumerate}
	For general $t$, we replace \eqref{eq:error_weight} by
	\begin{align}\label{eq:local_view_w_j}
		\vert \mathcal{X}_i \vert < \frac{\epsilon}{2 t^2 2^{t-1} \cdot n^t m_0^{t-1}} d
	\end{align}
	as $X_{\geq \supp(w_j)}(t)$ includes $\vert X_{\geq \supp(w_j)}(t) \vert \cdot t 2^{t-1}$ 1-cubes. Combining $Z$- and $X$-errors together, we have to clean the region of size $< d$, which is doable by the cleaning lemma \cite{Kalachev2022Cleaning}. As a result, for any $E_a, E_b$ with half the error bounds, the KL condition holds for $E = E_a^\dagger E_b$ and hence the gauged code distance $\geq \min\{ \frac{\epsilon}{4 t^2 2^{t-1} \cdot n^t m_0^{t-1}} d, \frac{1}{4 m_0^{t-1}} d \}$. It is straightforward to generalize this result to $t$ dimensions. Besides
\end{proof}

\begin{corollary}\label{coro:2D_dist}
	The following distance bounds holds for different gauged code models
	\begin{enumerate}
		\item For the 2-dimensional sheaf codes based on left-right Cayley complex \cite{DHLV2022}, $d = \Theta(N)$.
		
		\item For high dimensional sheaf codes based on abelian lifts \cite{Dinur2024sheaf}, $d = \Theta(N/\mathrm{polylog}\,N)$.
		
		\item For good qLTCs based on  Ramanujan cubical complexes \cite{GJ2026,BLN2026,CHLT2026}, $d = \Theta(N)$.
	\end{enumerate}
\end{corollary}
\begin{proof}
	The last two statements hold immediately by using the corresponding code parameters and cocycle expansions of these codes families. As a reminder, the underlying cubical complexes in \cite{Dinur2024sheaf} and \cite{GJ2026,BLN2026,CHLT2026} are different: the former is obtained by abelian lifts of products of expander graphs while the latter is constructed by arithmetic quotients of product of Bruhat--Tits trees \cite{RSV2019}. Nevertheless, the complexes are sparse, the local counting of $X_{\geq \supp(w_j)}(t)$ are identical in the above proof and cup products are well-defined \cite{Li2025Poincare,LLL2026nontrivial}, so Theorem \ref{thm:code_distance_sheaf} applies to both cases.  
	
	When $t = 2$, the $(\alpha,\beta,\gamma)$-small-set-coboundary expansion says that for any $x \in C^1$ such that $\vert x \vert < \alpha \dim C^1$, there is some $u \in C^0$ such that
	\begin{align}
		\vert \delta^1 x \vert \geq \beta \vert x + \delta^0 u \vert \text{ and } \vert u \vert \leq \gamma \vert x \vert. 
	\end{align}
	In our case, the Pauli $X$ errors under consideration are 1-cocycles, so $\delta^1 x = 0$. When their weights are upper bounded by $\alpha \dim C^1$, we can find $\delta^0 u = x$ such that $\vert u \vert \leq \gamma \vert x \vert \leq \gamma \alpha \dim C^1$. This provides the same type of bound on $\vert u \vert$ as 0-cocycle expansion. The small-set coboundary expansion with constant $(\alpha,\beta,\gamma)$ is proved in \cite{DHLV2022} for good qLDPC codes based on $\mathbb{F}_2$ and can be easily generalized to a sufficiently large $\mathbb{F}_q$, which further ensures the existence of product-expanding punctured RS codes \cite{LLL2026nontrivial}.
\end{proof}

\subsection{Long-range magic throughout the gauged code space}\label{sec:m_sheaf}

We construct an asymptotic family of gauged sheaf codes with long-range magic throughout their code spaces. By the dimension criterion in~\cite{Wei2025LRM}, it suffices to show that the code dimension is not a power of $2$. We establish this by proving that the dimension is divisible by $3$.

Our construction uses the quantum code induction scheme with the covering-space method developed in~\cite{LSWLL2026Theory,LLL2026nontrivial}. As a brief introduction, recall the definition of a cubical complex in Section \ref{sec:cubical}. Taking the same base graph $G_0$, we define cubical complexes $X, X'$. The first one admits a nontrivial lift while the second one uses the trivial lift. Then there is a natural covering map $X \rightarrow X'$ defined by
\begin{align}\label{eq:projection}
	\sigma = [(h,v_1,\ldots, v_t);\,(a_j)_{j\in S},\,(b_j)_{j\notin S}] \mapsto \sigma' = [(v_1,\ldots, v_t);\,(a_j)_{j\in S},\,(b_j)_{j\notin S}],
\end{align}
for any $h \in \HH$ indexing the lifts. The preimage $\{\sigma\}$ is the \emph{fiber} of $\sigma'$. Accordingly, the sheaf cochain complex of $X$ can be viewed as the \emph{covering space} of the HGP cochain complex of $X'$:
\begin{align}
	C^\bullet(X,\F) \longrightarrow C^\bullet(X',\F'),
\end{align}
where $\F$ is the pullback cellular sheaf of $\F'$ via the covering map: for any $\sigma$, we define $\F(\sigma) \coloneq \F'(\sigma')$. The local linear maps are defined similarly. Accordingly, we identify $\F$ with $\F'$. 

We still use three blocks of 2-dimensional sheaf codes to illustrate the proofs. As mentioned after Lemma \ref{lemma:coefficient}, we choose three different sheaves that would facilitate our proofs:
\begin{align}\label{eq:HGP_sheaf}
	\F=(h_1,h_2), \qquad \mathcal{G} = (h_1,h_3), \qquad \mathcal{H} = (h_4,h_2),
\end{align}
where $(\C_1,\C_2)$ generated by $h_i$, $i = 1,2$, are still product-expanding punctured RS codes and $h_3, h_4$ are taken to be the generator matrices of the dual codes $(\C_2 \ast \C_2)^\perp$,  $(\C_1 \ast \C_1)^\perp$, respectively. More details about the local codes is given after \eqref{eq:HGP_rank} when we present concrete computations. For now, we take three sheaf codes:
\begin{align}\label{eq:sheaf_magic}
	\begin{aligned}
		& C^0(X,\F) \xrightarrow{\delta^0} C^1(X,\F) \xrightarrow{\delta^1} C^2(X,\F), \\
		& C^0(X,\mathcal{G}) \xrightarrow{\delta^0} C^1(X,\mathcal{G}) \xrightarrow{\delta^1} C^2(X,\mathcal{G}), \\
		& C^0(X,\mathcal{H}) \xrightarrow{\delta^0} C^1(X,\mathcal{H}) \xrightarrow{\delta^1} C^2(X,\mathcal{H}).
	\end{aligned}
\end{align}
Suppose $[\mathbf{x}] = ([x_A],[x_B],[x_C]) \in L_G$ is a tuple of gauged cohomological classes. For any $s \in \HH$ and any logical representative $x_A \in C^1(X,\F)$, let
\begin{align}\label{eq:permute}
	(s \cdot x_A)(\sigma) \coloneq x_A(s^{-1} \cdot \sigma), \quad s^{-1} \cdot \sigma = [(s^{-1} \cdot h,v_1,v_2);\,(a_j)_{j\in S},\,(b_j)_{j\notin S}].
\end{align}
Namely, we permute the local coefficient vectors of $x_A$ within each fiber $\{\sigma\}$ by the translation action of $\HH$. If $x_A = x_A' + \delta u_A$, then
\begin{align}\label{eq:action_coboundary_commute}
	s \cdot x_A = s \cdot x_A' + s \cdot \delta^0 u_A = s \cdot x_A' + \delta^0 (s \cdot u_A).
\end{align}
The last equality holds immediately when $\HH$ is an abelian group. As a result, $s \cdot [x_A] \coloneq [s \cdot x_A]$ is well-defined on the cohomological classes.

Similar facts hold for $x_B$ and $x_C$. Let $[\mathbf{x}] = ([x_A], [x_B], [x_C]) \in L_G$. The gauge condition says that for any $0$-cocycle $\eta_A \in C^0(X,\F)$,
\begin{align}
	\langle \eta_A \smile x_B \smile x_C, \xi \rangle = 0 
\end{align}
Since $\delta(s \cdot\eta_A) = s \cdot (\delta \eta_A)$, $s \cdot \eta_A$ is also a cocycle. Then by permuting cup products within the fiber of any 2-cube, we obtain  
\begin{align}
	\left\langle \eta_A \smile (s \cdot x_B) \smile (s \cdot x_C), \xi \right\rangle = \left\langle s \cdot (s^{-1} \cdot \eta_A) \smile (s \cdot x_B) \smile (s \cdot x_C), \xi \right\rangle = 0, 
\end{align}
and hence $s \cdot [\mathbf{x}] \coloneq (s \cdot [x_A], s \cdot [x_B], s \cdot [x_C])$ is still gauged. 

\begin{lemma}\label{lemma:group_action}
	Let $\HH$ be the abelian group of odd order used in the definition of $X$. Then it induces a group action on the finite set $L_G$ of gauged cohomological classes. Moreover, fixed points of the group action admit a one-to-one correspondence with the gauged cohomological classes of HGP codes based on $X'$ given by \eqref{eq:projection}.
\end{lemma} 
\begin{proof}
	We have justified that the action is well-defined on $L_G$. Suppose $s \cdot [\mathbf{x}] \equiv [\mathbf{x}]$ for all $s \in \HH$, i.e., $[\mathbf{x}]$ is a fixed point. Let
	\begin{align}
		\overline{x}_i = \frac{1}{|\HH|} \sum_{s \in \HH} s \cdot x_i, \quad i = A,B,C.
	\end{align}
	Then for $s' \in \HH$, $s' \cdot \overline{x}_i = \overline{x}_i$. Since the translation action on $\HH$ is transitive, the local coefficient vectors restricted to any fiber $\{\sigma\}$ must be identical and thus the covering map \eqref{eq:projection} sends $\overline{x}_i$ to a codeword $\overline{x}_i'$ in the HGP code \cite{LSWLL2026Theory}. By assumption, 
	\begin{align}
		s \cdot [\mathbf{x}] \equiv [\mathbf{x}] \implies s \cdot x_i + \delta^0 u_i = x_i
		\implies \overline{x}_i + \delta \overline{u}_i = x_i 
	\end{align}
	for some $\overline{u}_i$. In other words, $[\mathbf{x}]$ can be represented by $\overline{x}_i$, which has a natural counterpart in the HGP codes, denoted by $[\mathbf{x}']$.
\end{proof}

For nonfixed points, $\{s \cdot [\mathbf{x}]\}$ is a nontrivial orbit. By Lagrange's theorem for finite groups, $|\{s \cdot [\mathbf{x}]\}|$ divides the group cardinality. Suppose $|\HH| = \omega^k$ is a power of some prime $\omega$. Let $L_G^{\mathrm{HGP}}$ be the collection of tuples of gauged cohomological classes of these HGP codes. Then
\begin{align}\label{eq:mod_sheaf_HGP}
	| L_G | \equiv | L_G^{\mathrm{HGP}} | \mod \omega.
\end{align}
Eq.~\eqref{eq:mod_sheaf_HGP} largely simplifies the task of showing that the gauged sheaf code dimension is not a power of $q$: we only need to compute $| L_G^{\mathrm{HGP}} |$ for the HGP codes. Moreover, the code induction scheme proposed in \cite{LSWLL2026Theory} can further reduce the task to the computation on three HGP codes of constant-size. To be precise, given any constant-size base graph $G_0$, the Cartesian product defines the cubical complex $X_1'$. We can either apply a constant-size lift in each graph direction to produce the next Cartesian-product complex $X_2'$, or take a large abelian lift of $X_1'$ to produce the high-dimensional expander $X_1$. By doing so inductively, we obtain a double sequence of cubical complexes
\begin{equation}\label{eq:induct_lift}
	\begin{array}{ccccc}
		X_1' & \xrightarrow{\HH_1'} & X_2' & \xrightarrow{\HH_2'} & \cdots \\
		{\scriptstyle \HH_1}\downarrow && {\scriptstyle \HH_2}\downarrow && \\
		X_1 & \dashrightarrow & X_2 & \dashrightarrow & \cdots
	\end{array}
\end{equation}
The upper row consists of HGP complexes, whereas the lower row consists of HDXs. Equipped with the sheaves $\F,\mathcal{G},\mathcal{H}$ and their pullbacks, the double sequence yields asymptotic families of sheaf codes and HGP codes, respectively. Since $X_i'$ and $X_{i+1}'$ are connected by lifts via the group $\HH_i'$, Lemma \ref{lemma:group_action} also holds between HGP codes. For simplicity, let $\HH_i' \equiv C_3$ be the cyclic group of three elements and let $\HH_i = C_3 \times \cdots \times C_3$ be the product with $|\HH_i| = \Theta(\exp(|X_i'|))$ (see \cite{Agarwal2016,Jeronimo2021}), then by Eq.~\eqref{eq:mod_sheaf_HGP} we have
\begin{align}\label{eq:sheaf_HGP_mod_3}
	| L_G^{i+1} | \equiv | (L_G^{\mathrm{HGP}})^{i+1} | \equiv | (L_G^{\mathrm{HGP}})^i | \equiv |L_G^i| \equiv \cdots \equiv | (L_G^{\mathrm{HGP}})^1 | \mod 3.
\end{align}


We now provide an explicit example of initial HGP codes for which $|(L_G^{\mathrm{HGP}})^1|$ is divisible by $3$, yielding an asymptotic family of gauged sheaf codes with constant coding rate and almost good distance whose code states are all long-range magic states. According to \cite{LLL2026nontrivial}, for sufficiently large $n$ and $q$, we can find a pair of two-way product expanding punctured RS codes $(\C_1,\C_2)$ generated by $(h_1,h_2)$ such that
\begin{align}\label{eq:HGP_rank}
	\rk h_1 = \rk h_2 = m = n/10 = \text{ power of } 3
\end{align}
Note that with this choice, $n$ is an even number. We use \eqref{eq:HGP_sheaf}. Usually, $h_3, h_4$ are generator matrices of $(\C_2 \ast \C_2)^\perp$,  $(\C_1 \ast \C_1)^\perp$, respectively. If that is the case, their rank equals $8m + 1$. However, we deliberately discard the last row of each matrix, so the actual rank is $8m$. This allows us to prove that $| (L_G^{\mathrm{HGP}})^1 |$ is divisible by $3$, yet does not affect the existence of $\xi$ and the product-expansion property \cite{LLL2026nontrivial}.  

Given $27m > n$, we define the base graph $G_0$ by taking vertices as elements of $C_{27m}$. By standard results on expander graphs \cite{Alon1986Expander,Alon1994}, there is a subset $\{a^1,\ldots,a^n\} \subset C_{27m}$ such that the graph $G_0$ defined by
\begin{align}\label{eq:HGP_base}
	V_0 = \{v \in C_{27m} \}, \quad E_0 = \{ ( v, v+a^i ): v \in V_0, i = 1,\ldots,n \}
\end{align}
is almost Ramanujan with second largest eigenvalue equal to $O(\sqrt{n \log n})$. Then, the double cover $G_0 \times K_2$ is taken as in Section \ref{sec:cubical} and $X_1'$ is defined by the 2-fold Cartesian product of $G_0 \times K_2$. For simplicity, we denote it by $X'$ with three HGP codes $C^\bullet(X',\F)$, $C^\bullet(X',\mathcal{G})$ and $C^\bullet(X',\mathcal{H})$.

Let $s \in C_{27m}$. Then $s + v \in C_{27m} = V_0$ and $( (s+v), (s+v+a^i)) \in E_0$ show that any translation action of $C_{27m}$ preserves the graph. As $s$ ranges over the translation group, the resulting group action on edges has $n$ orbits, each containing $27m$ edges. Let $\delta: \mathbb{F}_q^{2V_0} \rightarrow \mathbb{F}_q^{2E_0}$ be the coboundary operator of the double cover. Obviously, the induced permutations $s_v$ and $s_e$ satisfy $s_e \cdot \delta = \delta \cdot s_v$.

Now consider the HGP coboundary operators of $C^\bullet(X',\F)$:
\begin{align}\label{eq:HGP_coboundary_magic}
	\delta^1 = \begin{pmatrix} 
		I_{2E_0} \otimes H_2 & H_1 \otimes I_{2E_0}  
	\end{pmatrix}, \quad
	\delta^0 = \begin{pmatrix} H_1 \otimes I_{2mV_0} \\ I_{2mV_0} \otimes H_2 \end{pmatrix}
\end{align}
where $H_i$ is the sheaved coboundary operator of $\delta$ as defined after \eqref{eq:DLV_scalar}. Similarly to the above case of permuting local coefficient vectors within fibers of $X$, $X'$ admits a natural $C_{27m} \times C_{27m}$ action defined by
\begin{align}\label{eq:permute_HGP}
	\big( (s_1,s_2) \cdot x_A' \big)(\sigma) \coloneq x_A'( (-s_1, -s_2) \cdot \sigma), \quad (-s_1,-s_2) \cdot \sigma = [(v_1 - s_1, v_2 - s_2);\,(a_j)_{j\in S},\,(b_j)_{j\notin S}]
\end{align}
Using the representative permutation matrices on vertices and edges, we have
\begin{align}
	\begin{pmatrix} 
		I_{2E_0} \otimes H_2 & H_1 \otimes I_{2E_0}  
	\end{pmatrix} \begin{pmatrix} s_e \otimes s_v & 0 \\ 0 & s_v \otimes s_e	\end{pmatrix} & = 
	s_e \otimes s_e \begin{pmatrix} 
		I_{2E_0} \otimes H_2 & H_1 \otimes I_{2E_0}  
	\end{pmatrix} \\
	\begin{pmatrix} H_1 \otimes I_{2mV_0} \\ I_{2mV_0} \otimes H_2 \end{pmatrix} s_v \otimes s_v & =
	\begin{pmatrix} s_e \otimes s_v & 0 \\ 0 & s_v \otimes s_e \end{pmatrix}  \begin{pmatrix} H_1 \otimes I_{2mV_0} \\ I_{2mV_0} \otimes H_2 \end{pmatrix} 
\end{align}
Let $s = (s_1,s_2)$. If $x_A'^{(1)} = x_A'^{(2)} + \delta u_A'$, then
\begin{align}
	s \cdot x_A'^{(1)} = s \cdot x_A'^{(2)} + s \cdot \delta^0 u_A' = s \cdot x_A'^{(2)} + \delta^0 (s \cdot u_A').
\end{align}
As a result, $s \cdot [\mathbf{x}']$ is well-defined and the image of any gauged cohomological class is still gauged. Therefore, Lemma \ref{lemma:group_action} holds in the present case of HGP codes with a special choice of the base graph $G_0$ and
\begin{align}\label{eq:mod_HGP_all_ones}
	| L_G^{\mathrm{HGP}} | \equiv | L_G^{\mathrm{fixed}} | \mod 3.
\end{align}
where $L_G^{\mathrm{fixed}}$ consists of all tuples of gauged cohomological classes $([x_A'],[x_B'],[x_C'])$ such that $s \cdot x_i' \equiv x_i'$ for certain representatives. As a caveat, the double cover $G_0 \times K_2$ is bipartite and the translation action splits its vertex set into two orbits.

Recall that any cohomological class admits a representative that can be expanded by canonical logical representatives in Section \ref{sec:toric}. Given $s \cdot [x_i'] \equiv [x_i']$, suppose $x_i''$ is a linear combination of canonical logical representatives. As in Lemma \ref{lemma:group_action}
\begin{align}
	\overline{x}_i'' \coloneq \frac{1}{| 27m \times 27m|} \sum_{s \in C_{27m} \times C_{27m}} s \cdot x_i''
\end{align} 
satisfies $s \cdot \overline{x}_i'' \equiv  \overline{x}_i''$ and we simply identify it with $x_i'$. Note that each canonical logical representative is also averaged in the expansion. For example, let $s = (s_1,s_2)$ with fixed $s_1$ and let $e \otimes \zeta$ be one of the canonical logical representatives. Then
\begin{align}
	\sum_{s_2 \in C_{27m}} (s_1,s_2) \cdot e \otimes \zeta = (s_1 \cdot e) \otimes \sum_{s_2 \in C_{27m}} s_2 \cdot \zeta.
\end{align}
Since $H_2 (s_2 \cdot \zeta) = s_2 H_2 \zeta = 0$ and since the local Vandermonde matrices are always of full rank, when $G_0$ and its double cover $G_0 \times K_2$ are taken to be connected, translation invariance guarantees that the local coefficient vectors of $\zeta$ are identical over all 0-cubes $[v_2;b_2]$. 

As for $s_1 \cdot e = [s_1 + v_1; a_1]$, different free variables could merge if they lie in the same orbit of the group action. With all these facts, we write
\begin{align}\label{eq:translation_logical}
	x_i' = \sum_{r_i} x_{i,H}^{(r_i)} + \sum_{s_i} x_{i,V}^{(s_i)}, \quad i = A, B, C,
\end{align}
where we separate each $x_i'$ into the vertical and horizontal parts. We use $r_i$ with $|\{r_i\}|$ being the number of orbits of the free variables. Obviously $|\{r_i\}| \leq n$ and for further computation, their locations need to be clearly determined. For brevity, we only study \eqref{eq:HGP_coboundary_magic} defined by $\F = (h_1, h_2)$. Let $H_1$ be the sheaved coboundary operator of $G_0 \times K_2$ via $h_1$. For any $z \in \mathbb{F}_q^{2 m V_0}$, we have
\begin{align}
	(H_1 z)([v_1;a_1^r]) = (h_1 a_1^r)^T z([v_1;0]) + (h_1 a_1^r)^T z([v_1 + a_1^r;1]) \in \mathbb{F}_q.
\end{align}
After averaging, for any $v_1$,
\begin{align}
	\overline{H_1 z} ([v_1;a_1^r]) = (H_1 \overline{z}) ([v_1;a_1^r]) \equiv (h_1 a_1^r)^T (\overline{z}_0 + \overline{z}_1)
\end{align}
where we use $\overline{z}_0, \overline{z}_1$ to represent the averaged local coefficient vectors. Going through all $a_1^1,\ldots,a_1^n$, $\overline{H_1 z}$ can be represented by a codeword in $\Ima h_1^T = \C_1$. Since each $e \otimes \zeta$ in Eq.~\eqref{eq:translation_logical} is also averaged, it is a coboundary if and only if it can be expressed by $\C_1$ or $\C_2$ codewords in different directions. Since $h_i$ is a Vandermonde matrix, any of its $m$ columns are linearly independent and thus by taking $r_A, s_A = m+1,\ldots,n$, any nonzero expansion in Eq.~\eqref{eq:translation_logical} is a nontrivial cocycle. For the same reason, for code blocks $B$ and $C$, we have
\begin{align}
	& m+1 \leq r_B \leq n, \quad 8m+1 \leq s_B \leq n, \\
	& 8m+1 \leq r_C \leq n, \quad m+1 \leq s_C \leq n.
\end{align} 

For $\left\langle \eta_A \smile x_B \smile x_C, \xi \right\rangle$, permuting the positions of cubes in $X'$ gives 
\begin{align}
	\left\langle \eta_A \smile x_B \smile x_C, \xi \right\rangle = \left\langle (s \cdot \eta_A) \smile (s \cdot x_B) \smile  (s \cdot x_C), \xi \right\rangle 
\end{align}
If $s \cdot x_i \equiv x_i$,
\begin{align}
	\left\langle \eta_A \smile x_B \smile x_C, \xi \right\rangle = 0 \iff \sum_{s \in C_{27m} \times C_{27m}}	\left\langle (s \cdot \eta_A) \smile x_B \smile x_C, \xi \right\rangle = 0.
\end{align}
Note that the condition holds only for an odd number of translations. Otherwise, the right-hand side above becomes trivial.

As before, $\overline{\eta}_A \coloneq \frac{1}{|27m \times 27m|} \sum_{s \in C_{27m} \times C_{27m}} s \cdot \eta_A$ has a constant local coefficient vector, still denoted by $\overline{\eta}_A \in \mathbb{F}_q^{m \times m}$, over all 0-cubes. Let
\begin{align}
	C_A = h_1^T \overline{\eta}_A h_2,\qquad C_B = h_1^T \overline{\eta}_B h_3,\qquad C_C = h_4^T \overline{\eta}_C h_2.
\end{align}
For $i = A,B,C$, let $x_{i,H}(r,\text{-}), x_{i,V}(\text{-},s)$ be the local coefficients of $x_i$ in Eq.~\eqref{eq:translation_logical}. For simplicity, we omit the ``prime'' and let $1 \leq r,s \leq n$. For those $r,s$ that do not appear in the expansion, simply set $x_{i,H}(r,\text{-}), x_{i,V}(\text{-},s) = 0$. In computing cup products, by \eqref{eq:cup_sheaf_2}, their horizontal and vertical evaluations via different local codes are:
\begin{align}\label{eq:uv}
	\begin{array}{c|cc}
		& u_i = (u_{i,rs}) & w_i = (w_{i,rs}) \\ \hline
		A & x_{A,H} h_2 & h_1^T x_{A,V} \\
		B & x_{B,H} h_3 & h_1^T x_{B,V} \\
		C & x_{C,H} h_2 & h_4^T x_{C,V}.
	\end{array}
\end{align}
For example $u_{B,rs} = x_{B,H}(r,\text{-}) (h_3 a_2^{s})$ and $w_{B,rs} = (h_1 a_1^{r})^T x_{B,V}(\text{-},s)$. Then
\begin{align}\label{eq:cup_A}
	& \langle \overline{\eta}_A \smile x_B' \smile x_C', \xi \rangle \notag \\
	= & \sum_{(v_1,v_2) \in V_0^2} \sum_{r,s} (C_A)_{rs} (u_{B,rs}w_{C,rs} + w_{B,rs} u_{C,rs}) 
	= \sum_{r,s} (C_A)_{rs} (u_{B,rs} w_{C,rs} + w_{B,rs} u_{C,rs}) \\
	= & \sum_{r,s} \sum_{k,l} h_{1, kr} \overline{\eta}_{A,kl} h_{2, ls} (u_{B,rs} w_{C,rs} + w_{B,rs} u_{C,rs})
	= \sum_{a,b} \overline{\eta}_{A,ab} \sum_{r,s} h_{1, ar} (u_{B,rs} w_{C,rs} + w_{B,rs} u_{C,rs}) h_{2, bs} \notag
\end{align}
because $V_0^2 = (27m)^2$ is still odd and translation invariance implies that the contribution of each 2-cube to the pairing is identical. We have similar identities for  $\left\langle x_A' \smile \overline{\eta}_B \smile x_C', \xi \right\rangle$ and $\left\langle x_A' \smile x_B' \smile \overline{\eta}_C, \xi \right\rangle$.  Since the local coefficients $\overline{\eta}_i$ are arbitrary, all the pairings vanish exactly when
\begin{align}\label{eq:three_conditions}
	h_1(u_B \odot w_C + w_B \odot u_C) h_2^T & = 0,\\
	h_1(u_A \odot w_C + w_A \odot u_C) h_3^T & = 0,\\
	h_4(u_A \odot w_B + w_A \odot u_B) h_2^T & = 0.
\end{align}
The $\odot$ here is entrywise Hadamard multiplication. 

The local matrices $h_i$ are (weighted) Vandermonde matrices:
\begin{equation}\label{eq:Vandermonde}
	\begin{aligned}
		h_i(\text{-},r) & = (1,\beta_{i,r},\ldots,\beta_{i,r}^{m-1})^T
		\quad (i=1,2), \\
		h_4(\text{-},r) & = \lambda_{1,r}(1,\beta_{1,r},\ldots,\beta_{1,r}^{8m-1})^T,\\
		h_3(\text{-},r) & = \lambda_{2,r}(1,\beta_{2,r},\ldots,\beta_{2,r}^{8m-1})^T.
	\end{aligned}
\end{equation}
with respect to two sets of distinct evaluation points $(\beta_{i,1},\ldots,\beta_{i,n})$, $i = 1,2$ and Lagrangian multipliers $\lambda_{i,r} = \prod_{j \neq r}(\beta_{i,r}-\beta_{i,j})^{-1}$. Then we explicitly expand Eq.~\eqref{eq:cup_A} by using
\begin{align}
	\begin{aligned}
		u_{B,rs} & = \lambda_{2,s} \sum_{j=0}^{8m-1} x_{B,H}(r,j) \beta_{2,s}^j, &
		w_{C,rs} & = \lambda_{1,r} \sum_{i=0}^{8m-1} x_{C,V}(i,s) \beta_{1,r}^i, \\
		w_{B,rs} & = \sum_{i=0}^{m-1} x_{B,V}(i,s) \beta_{1,r}^i, &
		u_{C,rs} & = \sum_{j=0}^{m-1} x_{C,H}(r,j) \beta_{2,s}^j.
	\end{aligned}
\end{align}
and $(C_A)_{rs} = \sum_{a,b=0}^{m-1} \overline{\eta}_{A,ab} \beta_{1,r}^a \beta_{2,s}^b$, which yields:
\begin{align}\label{eq:cup_A_expand}
	\begin{aligned}
		& \sum_{r = m+1}^n \sum_{s = m+1}^n \lambda_{1,r}\lambda_{2,s} \sum_{i,j=0}^{8m-1} \beta_{1,r}^{a+i} \beta_{2,s}^{b+j} x_{B,H}(r,j) x_{C,V}(i,s) \\
		&\quad +
		\sum_{r = 8m+1}^n \sum_{s = 8m+1}^n \sum_{i,j=0}^{m-1} \beta_{1,r}^{a+i} \beta_{2,s}^{b+j} x_{C,H}(r,j) x_{B,V}(i,s) = 0,
		\qquad 0 \leq a,b < m.
	\end{aligned}
\end{align}
Similarly, the pairing involving $\overline{\eta}_{B,ab}$ with $0 \leq a < m$ and $0 \leq b < 8m$ is
\begin{align}\label{eq:cup_B_expand}
	\begin{aligned}
		& \sum_{r = m+1}^n \sum_{s = m+1}^n \lambda_{1,r} \lambda_{2,s} \sum_{i = 0}^{8m-1} \sum_{j = 0}^{m-1} \beta_{1,r}^{a+i} \beta_{2,s}^{b+j}
		x_{A,H}(r,j)x_{C,V}(i,s) \\
		& \quad + \sum_{r = 8m+1}^n \sum_{s=m+1}^n \lambda_{2,s} \sum_{i,j=0}^{m-1} \beta_{1,r}^{a+i} \beta_{2,s}^{b+j} x_{C,H}(r,j) x_{A,V}(i,s) = 0.
	\end{aligned}
\end{align}
For $\overline{\eta}_{C,ab}$ with $0 \leq a < 8m$ and $0 \leq b < m$, we have
\begin{align}\label{eq:cup_C_expand}
	\begin{aligned}
		&\sum_{r=m+1}^n \sum_{s=8m+1}^n \lambda_{1,r} \sum_{i,j=0}^{m-1} \beta_{1,r}^{a+i} \beta_{2,s}^{b+j} x_{A,H}(r,j) x_{B,V}(i,s) \\
		& \quad + \sum_{r=m+1}^n \sum_{s=m+1}^n \lambda_{1,r} \lambda_{2,s} \sum_{i=0}^{m-1} \sum_{j=0}^{8m-1} \beta_{1,r}^{a+i} \beta_{2,s}^{b+j} x_{B,H}(r,j) x_{A,V}(i,s) = 0.
	\end{aligned}
\end{align}

We simplify the above equations by defining six matrices as follows:
\begin{equation}\label{eq:hat_matrices}
	\begin{aligned}
		\widehat{x}_{A,H}(\ell,j) & = \sum_{r>m} \lambda_{1,r}\beta_{1,r}^{\ell} x_{A,H}(r,j), &
		\widehat{x}_{B,H}(\ell,j) & = \sum_{r>m} \lambda_{1,r}\beta_{1,r}^{\ell} x_{B,H}(r,j), \\
		\widehat{x}_{C,H}(\ell,j) & = \sum_{r>8m} \beta_{1,r}^{\ell} x_{C,H}(r,j), \\
		\widehat{x}_{A,V}(i,\ell) & = \sum_{s>m} \lambda_{2,s}\beta_{2,s}^{\ell} x_{A,V}(i,s), &
		\widehat{x}_{C,V}(i,\ell) & = \sum_{s>m} \lambda_{2,s}\beta_{2,s}^{\ell} x_{C,V}(i,s), \\
		\widehat{x}_{B,V}(i,\ell) & = \sum_{s>8m} \beta_{2,s}^{\ell} x_{B,V}(i,s).
	\end{aligned}
\end{equation}
For fixed \(i,j,a,b\), the first sum in Eq.\eqref{eq:cup_A_expand} is
\begin{align}
	\left(\sum_{r>m}\lambda_{1,r}\beta_{1,r}^{a+i}x_{B,H}(r,j)\right) \left(\sum_{s>m}\lambda_{2,s}\beta_{2,s}^{b+j}x_{C,V}(i,s)\right) = \widehat{x}_{B,H}(a+i,j) \widehat{x}_{C,V}(i,b+j)
\end{align}
Rearranging other terms gives
\begin{align}
	& \sum_{i,j=0}^{8m-1} \widehat{x}_{B,H}(a+i,j) \widehat{x}_{C,V}(i,b+j)
	+ \sum_{i,j=0}^{m-1} \widehat{x}_{C,H}(a+i,j) \widehat{x}_{B,V}(i,b+j) = 0,
	&& 0 \leq a,b < m, \label{eq:cup_A_hat} \\
	& \sum_{i=0}^{8m-1} \sum_{j=0}^{m-1} \widehat{x}_{A,H}(a+i,j) \widehat{x}_{C,V}(i,b+j) 
	+ \sum_{i,j=0}^{m-1} \widehat{x}_{C,H}(a+i,j) \widehat{x}_{A,V}(i,b+j) = 0,
	&& \substack{0 \leq a < m \\ 0 \leq b < 8m}, \label{eq:cup_B_hat} \\
	& \sum_{i,j=0}^{m-1} \widehat{x}_{A,H}(a+i,j) \widehat{x}_{B,V}(i,b+j) 
	+ \sum_{i=0}^{m-1} \sum_{j=0}^{8m-1} \widehat{x}_{B,H}(a+i,j) \widehat{x}_{A,V}(i,b+j) = 0,
	&& \substack{0 \leq a < 8m \\ 0 \leq b < m}. \label{eq:cup_C_hat}
\end{align}
Given the range of $i,j,a,b$, the sizes of these hatted matrices are
\begin{equation}\label{eq:hatsizes}
	\begin{array}{c|cc}
		& \widehat{x}_{i,H} & \widehat{x}_{i,V} \\ \hline
		A & (9m-1) \times m & m \times (9m-1) \\
		B & (9m-1) \times 8m & m \times (2m-1) \\
		C & (2m-1) \times m & 8m \times (9m-1).
	\end{array}
\end{equation}
Let $N_m$ be the number of six-tuples of hatted matrices in \eqref{eq:hatsizes} satisfying Eqs.~\eqref{eq:cup_A_hat}--\eqref{eq:cup_C_hat}. We note that, for example, $\widehat{x}_{A,H}(\ell,j) = \sum_{r>m} \lambda_{1,r}\beta_{1,r}^{\ell} x_{A,H}(r,j)$ is defined by applying a submatrix $(\lambda_{1,r}\beta_{1,r}^{\ell})$ of size $(9m - 1) \times 9m$ of the weighted Vandermonde matrix. The corresponding linear map is surjective and thus for any fixed $\widehat{x}_{A,H}$, it has a preimage of size equal to $q^m$. The other five hatted matrices have preimages of sizes equal to $q^m,q^{8m},q^m,q^m,q^{8m}$, respectively. As a result, solving Eqs.~\eqref{eq:cup_A_expand}--\eqref{eq:cup_C_expand} reduces to solving for the hatted matrices with
\begin{align}
	|L_G^{\mathrm{fixed}}| = q^{20m}N_m.
\end{align}

\begin{claim}
	Let $m$ be a power of $3$ and let $q=2^s > n = 10m$ with odd $s$. If $3 | q + 1$, then
	\begin{align}
		|L_G^{\mathrm{HGP}}| \equiv |L_G^{\mathrm{fixed}}| \equiv q^{20m}N_m \equiv 0 \pmod{3}.
	\end{align}
\end{claim}
\begin{proof}
	For any odd $s$, we always have $2^s \equiv -1 \pmod{3}$ and thus the assumption holds. We are going to prove by representation theory that $N_m \equiv 2q-1 \equiv -3 \pmod{q+1}$. Since $3 | q + 1$, the result follows.

	We briefly introduce some basic notions: consider the algebraic group
	\begin{align}
		T(\mathbb{F}_q) \coloneq \left \{ g(u,v) = \begin{pmatrix} u & v \\ v & u+v \end{pmatrix}: u,v \in \mathbb{F}_q,\ \det g = u^2 + uv + v^2 = 1 \right\}.
	\end{align}
	Then for degrees $0 \leq D \leq 9m-2 < q+1$, consider the space of homogeneous polynomials and its dual:
	\begin{align}\label{eq:poly_space}
		V_D = \mathbb{F}_q[x,y]_D = \left\{ \sum_{i=0}^D c_i x^{D-i} y^i: c_i \in \mathbb{F}_q \right\},
		\qquad V_D^\ast = \mathrm{Hom}_{\mathbb{F}_q}(V_D, \mathbb{F}_q).
	\end{align}
	Both have dimension $D+1$. For $g(u,v) \in T(\mathbb{F}_q)$, we define the action
	\begin{align}\label{eq:poly_action}
		(g(u,v) \cdot f)(x,y) = f(ux + vy, vx + (u+v)y).
	\end{align}
	Let $\rho_D(g) \in \mathrm{GL}_{D+1}(\mathbb{F}_q)$ be its matrix representation under the ordered basis $x^D, x^{D-1}y, \ldots,y^D$. The dual basis consists of coefficient extractors $\epsilon_i(f) = c_i$. If \(F = \sum_i F_i \epsilon_i\), then
	\begin{align}\label{eq:dual_action}
		F(f) = \sum_{i=0}^D F_i c_i, \qquad (g \cdot F)(f) = F(g^{-1} \cdot f),
	\end{align}
	Any $\mathbb{F}_q$-vector is algebraically fixed if it is fixed over the algebraic closure $\overline{\mathbb{F}}_q$ under the action of $T(\overline{\mathbb{F}}_q)$. Our case is simpler because we do not need to work with $\overline{\mathbb{F}}_q$ as the largest polynomial degree is much smaller than $q$. Then
	\begin{align}\label{eq:fixed_space}
		\begin{aligned}
			V_D^T & = \{f \in V_D: g \cdot f = f \text{ for every } g \in T(\mathbb{F}_q)\}, \\
			(V_D^*)^T & = \{F \in V_D^*: g \cdot F = F \text{ for every }g\in T(\mathbb{F}_q) \}.
		\end{aligned}
	\end{align}
	A standard result \cite{Milne_2017} about the action \eqref{eq:poly_action} on polynomial space says that
	\begin{align}
		\dim V_D^T = \dim(V_D^*)^T = \begin{cases} 1, & D \text{ is even} \\ 0, & D \text{ is odd} \end{cases}
	\end{align}
	
	In our case, we rewrite the six hatted matrices as elements in the tensor product of the polynomial spaces and their duals. For example, consider $\widehat{x}_{C,H}$ of size $(2m-1) \times m$. We use its $j$-th column to define a functional $L_j: V_{2m-2} \to \overline{\mathbb{F}}_q$, for $0 \leq j < m$. Then for any polynomial $f = \sum_{\ell = 0}^{2m-2} f_\ell x^{2m-2-\ell} y^\ell$ with $0 \leq \ell \leq 2m-2$ and $0 \leq j < m$, we have
	\begin{align}\label{eq:hat_matrix_action}
		\begin{aligned}
			\widehat{x}_{C,H}(\ell,j) & = \sum_{r=8m+1}^{10m} \beta_{1,r}^{\ell} x_{C,H}(r,j) \implies L_j = \sum_{\ell=0}^{2m-2} \widehat{x}_{C,H}(\ell,j) \epsilon_\ell, \\
			L_j(f) & = \sum_{\ell=0}^{2m-2} \widehat{x}_{C,H}(\ell,j) f_\ell = \sum_{r = 8m+1}^{10m} x_{C,H}(r,j) f(1,\beta_{1,r}), \\
			\implies \widehat{x}_{C,H} & \coloneq \sum_{j=0}^{m-1} L_j\otimes x^{m-1-j} y^j \in V_{2m-2}^*\otimes V_{m-1}, 
		\end{aligned}
	\end{align}
	The action is defined by the tensor product of \eqref{eq:poly_action} and \eqref{eq:dual_action}: 
	\begin{align}\label{eq:tensor_action}
		(g,h) \cdot \widehat{x}_{C,H} = \rho_{2m-2}(g)^{-T} \widehat{x}_{C,H} \rho_{m-1}(h)^T
	\end{align}
	The spaces of all these matrices and the dimensions of their fixed subspaces under the actions are
	\begin{align}
		\begin{array}{c|c|c}
			\text{matrix} & \text{space} & \text{dimension fixed by } T\times T \\ \hline
			\widehat{x}_{A,H} & V_{9m-2}^*\otimes V_{m-1} & 0 \\
			\widehat{x}_{A,V} & V_{m-1}\otimes V_{9m-2}^* & 0 \\
			\widehat{x}_{B,H} & V_{9m-2}^*\otimes V_{8m-1} & 0 \\
			\widehat{x}_{B,V} & V_{m-1}\otimes V_{2m-2}^* & 1 \\
			\widehat{x}_{C,H} & V_{2m-2}^*\otimes V_{m-1} & 1 \\
			\widehat{x}_{C,V} & V_{8m-1}\otimes V_{9m-2}^* & 0
		\end{array}
	\end{align}
	
	Now, let $Y(\mathbb{F}_q)$ be the set of solutions of Eqs.~\eqref{eq:cup_A_hat}--\eqref{eq:cup_C_hat}, so $|Y(\mathbb{F}_q)| = N_m$. It is the affine algebraic set defined by $17m^2$ quadratic equations in the affine coordinate space with $166m^2-20m$ coordinates over $\mathbb{F}_q$. It is straightforward to check by the definition of Eq.~\eqref{eq:tensor_action} that the action preserves the set of solutions as multiplication of homogeneous polynomials and the dual pairing are equivariant. Therefore, the following two actions on the solution set and the corresponding invariant sets are well-defined: 
	\begin{align}
		& \qquad\qquad T \times Y \longrightarrow Y,\qquad (g,z) \longmapsto (g,I) \cdot z, \\
		& Y^{T \times \{1\}}(\mathbb{F}_q) = \{z \in Y(\mathbb{F}_q): (g,I) \cdot z = z\text{ for every } g \in T(\mathbb{F}_q)\}, \\
		& \qquad\qquad T \times Y^{T \times \{1\}} \longrightarrow Y^{T \times \{1\}},
		\qquad (h,z) \longmapsto (I,h) \cdot z, \\
		& Y^{T \times T}(\mathbb{F}_q) = \{z \in Y^{T \times \{1\}}(\mathbb{F}_q): (I,h) \cdot z = z \text{ for every } h \in T(\mathbb{F}_q)\}.
	\end{align}
	We use another standard result of representation theory and algebraic group \cite{Milne_2017}: if $T(\mathbb{F}_q)$ acts algebraically on an affine algebraic set $X$, then $|X(\mathbb{F}_q)| \equiv |X^T(\mathbb{F}_q)| \pmod{q+1}$. We apply this result twice to our case and obtain
	\begin{align}
		|Y(\mathbb{F}_q)| \equiv |Y^{T \times \{1\}}(\mathbb{F}_q)| \equiv |Y^{T \times T}(\mathbb{F}_q)| \pmod{q+1}.
	\end{align}
	
	We now turn to study the solutions that are fixed by the action. Let $P \in V_{m-1}^T$ and $Q \in V_{m-1}^T$ be nontrivial polynomials in the corresponding 1-dimensional spaces. By definition, for any $g \in T(\mathbb{F}_q)$, since it acts on the variables, $g \cdot (P^2) = (g\cdot P)^2 = P^2$. Similarly, $g \cdot (Q^2) = Q^2$. A nonzero polynomial has a nonzero square and thus $P^2$ and $Q^2$ span the one-dimensional spaces $V_{2m-2}^T$ and $V_{2m-2}^T$, respectively. They correspond to the unique invariant functionals $\alpha \in (V_{2m-2}^*)^T$ and $\beta \in (V_{2m-2}^*)^T$ with $\alpha(P^2) = \beta(Q^2) = 1$. Then we can express the two nontrivial invariant matrices by
	\begin{align}
		\widehat{x}_{B,V} = b_0 P \otimes \beta,\qquad
		\widehat{x}_{C,H} = c_0 \alpha \otimes Q, \qquad b_0,c_0 \in \mathbb{F}_q.
	\end{align}
	In coordinates, we have the following expansions:
	\begin{align}
		\begin{aligned}
			\widehat{x}_{B,V}(i,j) & = b_0 P_i \beta_j, && 0 \leq i < m, \quad 0 \leq j < 2m-1,\\
			\widehat{x}_{C,H}(i,j) & = c_0 \alpha_i Q_j, && 0 \leq i < 2m-1,\quad 0 \leq j < m.
		\end{aligned}
	\end{align}
	Here $P_i,Q_j$ are monomial coefficients, and $\alpha_i,\beta_j$ are dual-basis coefficients.
	
	Checking Eqs.~\eqref{eq:cup_A_hat}--\eqref{eq:cup_C_hat} again, we find that only the second term in Eq.~\eqref{eq:cup_A_hat} remains. Since this holds for any $0 \leq a < m$ and $0 \leq b < m$, we apply arbitrary testing polynomials $u_1 \in V_{m-1}$ and $u_2 \in V_{m-1}$ to the second sum:
	\begin{equation}\label{eq:twofactorproduct}
		\begin{aligned}
			& \sum_{a,b=0}^{m-1} u_{1,a} u_{2,b}
			\sum_{i,j=0}^{m-1} \widehat{x}_{C,H}(a+i,j) \widehat{x}_{B,V}(i,b+j) \\
			= & b_0 c_0 \left(\sum_{a=0}^{m-1} \sum_{i=0}^{m-1} u_{1,a} P_i \alpha_{a+i} \right)
			\left(\sum_{b=0}^{m-1} \sum_{j=0}^{m-1} u_{2,b} Q_j \beta_{b+j} \right) = b_0 c_0 \alpha(u_1 P) \beta(u_2 Q).
		\end{aligned}
	\end{equation}
	Since $u_1,u_2$ are arbitrary, we must have $b_0 c_0 = 0$ to force the sum to be zero. Consequently, 
	\begin{align}
		\begin{aligned}
			N_m = |Y(\mathbb{F}_q)| & \equiv |Y^{T \times \{1\}}(\mathbb{F}_q)| \equiv |Y^{T \times T}(\mathbb{F}_q)| \\
			& = \#\{(b_0,c_0) \in \mathbb{F}_q^2: b_0 c_0 = 0 \}
			= q + (q-1) = 2q-1 \equiv -3 \pmod{q+1}.
		\end{aligned}
	\end{align}
	Since $3 | q+1$, we can conclude that $N_m \equiv 0 \pmod{3}$.
\end{proof}

In conclusion, we have constructed a constant-size instance of an HGP code such that the gauged code dimension is a multiple of $3$. As a result, it cannot be any power of $q = 2^s$. Applying the code induction scheme in \eqref{eq:induct_lift} yields an asymptotic family of non-Abelian qLDPC codes with almost good parameters whose code spaces are filled with long-range magic states:
\begin{align}
	| L_G^{i+1} | \equiv | (L_G^{\mathrm{HGP}})^{i+1} | \equiv | (L_G^{\mathrm{HGP}})^i | \equiv |L_G^i| \equiv \cdots \equiv | (L_G^{\mathrm{HGP}})^1 | \equiv 0 \mod 3.
\end{align}
In order to attain good distance, we have to apply the non-Abelian lift in the construction of cubical complexes as explained in Remark \ref{remark:2D_dim} and Corollary \ref{coro:2D_dist}. However, the starting point of our method in this section assumes abelian lifts in Eq.~\eqref{eq:action_coboundary_commute}, and this identity does not hold for non-Abelian lifts. We leave this problem as a future research opportunity.


\subsection{Gauging measurements and magic state generation}

Finally, we demonstrate that the standard gauging/ungauging process \cite{Christos2026,Zhu2026} yields magic states by using three sheaf code blocks over a 2-dimensional expander $X$: 
\begin{align}\label{eq:sheaf_ungauge_measurement}
	\begin{aligned}
		& C^0(X,\F) \xrightarrow{\delta^0} C^1(X,\F) \xrightarrow{\delta^1} C^2(X,\F), \\
		& C^0(X,\mathcal{G}) \xrightarrow{\delta^0} C^1(X,\mathcal{G}) \xrightarrow{\delta^1} C^2(X,\mathcal{G}), \\
		& C^0(X,\mathcal{H}) \xrightarrow{\delta^0} C^1(X,\mathcal{H}) \xrightarrow{\delta^1} C^2(X,\mathcal{H}),
	\end{aligned}
\end{align}
where $\mathcal{F}, \mathcal{G}, \mathcal{H}$ are still generated by RS code families to guarantee the existence of $\xi$. Unlike the previous cases, we do not use them to construct the good gauged code, but we show below that they are sufficient to support a gauging measurement and produce magic states in the original code blocks of $B$ and $C$.
We begin with the code state in $B$ and $C$ as
\begin{align}
	\ket{\overline{+}_B,\overline{+}_C} \propto \sum_{y_B \in \ker\delta^1(\mathcal{G}), y_C \in \ker\delta^1(\mathcal{H})} \ket{y_B,y_C},
\end{align}
where $\ker\delta^1(\mathcal{G})$ denotes the kernel of $\delta^1$ of $C^\bullet(X,\mathcal{G})$ and $\ker\delta^1(\mathcal{H})$ is defined similarly. By definition, $\ket{\overline{+}_B,\overline{+}_C}$ is an eigenstate with $+1$ eigenvalue of every $X$ stabilizer of $B$ and $C$, analogous to the product state $\frac{1}{2} (\ket{0} + \ket{1}) \otimes (\ket{0} + \ket{1})$ in the qubit case. We will often use the symbol ``$\propto$'' in the following to avoid explicit normalization factors. 

Then with respect to each standard basis vector in $C^0(X,\F)$, we assign one ancilla qudit. When the local coefficients are trivial, the ancilla qudits can be intuitively thought of as placed on each $0$-cube. Let 
\begin{align}\label{eq:U_{aA}}
	U = \sum_{u_A \in C^0(X,\F)} \sum_{x_A \in C^1(X,\F)} \ket{ u_A, x_A + \delta u_A} \bra{ u_A, x_A}.
\end{align} 
By definition, $U$ is a constant-depth circuit of $\CNOT$ operators acting on each $0$-cube with local coefficients and its coboundary. The joint initial state is
\begin{align}
	U \ket{+_a,0_A}, \qquad \ket{+_a} \propto \sum_{u_A \in C^0(X,\F)} \ket{u_A}
\end{align}
where $\ket{+_a}$ is the uniform superposition state of the ancilla qudits. 

The full initial state of three code blocks with ancilla qudits is
\begin{align}
	\ket{\Psi_0} = \Big( U \ket{+_a,0_A} \Big) \otimes \ket{\overline{+}_B,\overline{+}_C}.
\end{align}
For further use, let $X_{\alpha v}$ denote the single-qudit Pauli $X^\alpha = \sum_{\beta \in \mathbb{F}_q} \ket{\alpha + \beta} \bra{\beta}$ based on the standard basis vector $v \in C^0(X,\F)$. Then by \eqref{eq:CNOT_Pauli}, 
\begin{align}
	& U \Big( \ket{u + w} \bra{u} \Big) U^\dagger \notag \\ 
	= & \sum_{\substack{u_A \in C^0(X,\F) \\ x_A \in C^1(X,\F)}} \ket{ u_A, x_A + \delta u_A} \bra{ u_A, x_A} \Big( \ket{u + w} \bra{u} \Big) \sum_{\substack{u_A' \in C^0(X,\F) \\ x_A' \in C^1(X,\F)}} \ket{ u_A', x_A'} \bra{ u_A', x_A' + \delta u_A'} \notag \\
	= & \sum_{x_A \in C^1(X,\F)} \ket{ u + w, x_A + \delta u + \delta w} \bra{ u, x_A} \sum_{\substack{u_A' \in C^0(X,\F) \\ x_A' \in C^1(X,\F)}} \ket{ u_A', x_A'} \bra{ u_A', x_A' + \delta u_A'}  \\
	= & \sum_{\substack{x_A \in C^1(X,\F)}} \ket{u + w, x_A + \delta u + \delta w} \bra{ u, x_A + \delta u}. \notag
\end{align}
This implies $U X_{\alpha v} U^\dagger = X_{\alpha v} X_{\alpha \delta v}$. Then we apply the following constant-depth circuit of physical $\CCZ$ operators, which is different from the dressed $X$ stabilizers in Eqs.~\eqref{eq:X_A} and \eqref{eq:A_i_sheaf}:
\begin{align}
	\Omega \coloneq \sum_{u_A, x_B, x_C} (-1)^{ \tr_{\mathbb{F}_q/\mathbb{F}_2}( \int_\xi u_A \smile x_B \smile x_C) } \ket{u_A,x_B,x_C} \bra{u_A,x_B,x_C}, 
\end{align}
and measure the ancilla qudits specified by the standard basis of $C^0(X,\F)$. The measurement is carried out with respect to the simultaneous eigenstates of $\{X_{\alpha v}\}$ where $v \in C^0(X,\F)$ is still a standard basis vector while $\alpha$ is taken from an $\mathbb{F}_2$-basis of $\mathbb{F}_q$. 
A measurement outcome is a string of $x_{\alpha v} = \pm 1$ as introduced in Section \ref{sec:F_q}. With this convention,
\begin{align}
	\prod_{v,\alpha} \frac{I + x_{\alpha v} X_{\alpha v}}{2} \propto \sum_{u \in C^0(X,\F)} x_u X_u, 
\end{align}
where $u = \sum_{v} \sum_\alpha \alpha v \in C^0(X,\F)$ is uniquely expressed by $\alpha$ and $v$, and $x_u$ is the product of $x_{\alpha v}$ with $\alpha$ and $v$ in the expansion of $u$.

Obviously, $U \Omega = \Omega U$ and
\begin{align}
	\ket{\Psi_0} \mapsto \ket{\Psi_1} \coloneq & \prod_{v,\alpha} \frac{I + x_{\alpha v} X_{\alpha v}}{2} \Omega \ket{\Psi_0} = \prod_{v,\alpha} \frac{I + x_{\alpha v} X_{\alpha v}}{2} \Omega U \ket{+_a,0_A} \otimes \ket{\overline{+}_B,\overline{+}_C} \notag \\
	= & \prod_{v,\alpha} \frac{I + x_{\alpha v} X_{\alpha v}}{2} U \Omega \ket{+_a,0_A} \otimes \ket{\overline{+}_B,\overline{+}_C} \\
	= & U \prod_{v,\alpha} \frac{I + x_{\alpha v} X_{\alpha v} X_{\alpha \delta v} }{2} \Omega \ket{+_a,0_A} \otimes \ket{\overline{+}_B,\overline{+}_C} \notag \\
	\propto \ & U \left( \sum_{u \in C^0(X,\F)} x_u X_u X_{\delta u} \right) \Omega \ket{+_a,0_A} \otimes \ket{\overline{+}_B,\overline{+}_C}. \notag
\end{align}


The next step is ungauging: we first apply $U^\dagger$ and then measure in the $Z$-eigenbasis of $\{Z_{\beta e} \}$. Here $Z_{\beta e}$ is a single-qudit Pauli $Z^\beta$ based on the standard basis vector $e \in C^1(X,\F)$ with $\beta$ being an $\mathbb{F}_2$-basis vector of $\mathbb{F}_q$. One possible outcome is
\begin{align}
	\ket{\Psi_1} \mapsto \ket{\Psi_2} \propto \prod_{e,\beta} \frac{I + z_{\beta e} Z_{\beta e} }{2} \left( \sum_{u \in C^0(X,\F)} x_u X_u X_{\delta u} \right) \Omega_{aBC} \ket{+_a,0_A} \otimes \ket{\overline{+}_B,\overline{+}_C}.
\end{align} 
Note that each measurement outcome $z_{\beta e} = (-1)^{\tr_{\mathbb{F}_q/\mathbb{F}_2} (\beta \gamma)}$ is defined with $\gamma$ being the value of the codeword at $e$. Accordingly, we define $y = \sum_e \gamma e \in C^1(X,\F)$ and let $X_y$ be the corresponding Pauli $X$ string. Then by definition,
\begin{align}
	X_y Z_{\beta e} = \begin{cases} Z_{\beta e} X_y, & z_{\beta e} = 1, \\ - Z_{\beta e} X_y, & z_{\beta e} = -1 \end{cases} 
\end{align} 
and 
\begin{align}
	\prod_{e,\beta} \frac{I + z_{\beta e} Z_{\beta e} }{2} = X_y^2 \prod_{e,\beta} \frac{I + z_{\beta e} Z_{\beta e} }{2} = X_y \prod_{e,\beta} \frac{I +  Z_{\beta e} }{2} X_y  
\end{align}
As a result
\begin{align}
	\ket{\Psi_2} \propto X_y \prod_{e,\beta} \frac{I +  Z_{\beta e} }{2} \left( \sum_{u \in C^0(X,\F)} x_u X_u X_{y + \delta u} \right) \Omega \ket{+_a,0_A} \otimes \ket{\overline{+}_B,\overline{+}_C}.
\end{align}
Since $X_{y + \delta u} \ket{0_A} = \ket{y + \delta u}$, the corresponding $Z$-projection is nonzero only when the state is a $+1$ eigenstate with respect to Pauli $Z$s. For every $\mathbb{F}_2$-basis vector $\beta$ and every component of $(-1)^{\tr_{\mathbb{F}_q/\mathbb{F}_2} (\beta \cdot (y + \delta u))}$ to be $1$, we have to have $\delta u + y = 0$. All the other components in the above projection vanish.

Let $y = \delta u_0$. Then $u = u_0$ plus elements in $\ker \delta^0$ and
\begin{align}
	\ket{\Psi_2} \propto X_{\delta u_0} \left( \sum_{\mu \in \ker \delta^0} x_{u_0 + \mu} X_{u_0 + \mu} \right) \Omega \ket{+_a,0_A} \otimes \ket{\overline{+}_B,\overline{+}_C}.
\end{align} 
We now apply $X_{u_0}^{-1}$ and obtain
\begin{align}
	\begin{aligned}
		\ket{\Psi_2} \mapsto \ket{\Psi_3} \propto \ & X_{\delta u_0} \left( \sum_{\mu \in \ker \delta^0} x_{u_0 + \mu} X_{\mu} \right) \Omega \ket{+_a,0_A} \otimes \ket{\overline{+}_B,\overline{+}_C} \\ 
		\propto & X_{\delta u_0} \prod_{\mu_i,\beta} \frac{I + x_{\beta \mu_i} X_{\beta \mu_i}}{2} \Omega \ket{+_a,0_A} \otimes \ket{\overline{+}_B,\overline{+}_C}
	\end{aligned}
\end{align} 
where $\{\mu_i\}$ is a basis of $\ker \delta^0$ and $\{\beta\}$ is the $\mathbb{F}_2$-basis. Actually, $\sum_{\mu \in \ker \delta^0} x_{u_0 + \mu} X_{\mu}$ is a partial sum of the original projector $\sum_{u \in C^0(X,\F)} x_u X_u$. With a proper normalization factor, the above identity holds because $\ker \delta^0$ is a subspace of $C^0(X,\F)$. Indeed,
\begin{align}
	\left( \frac{1}{|\ker \delta^0|} \sum_{\mu \in \ker \delta^0} x_{\mu} X_{\mu} \right)^2 = \frac{1}{|\ker \delta^0|^2} \left( \sum_{\mu,\mu'} x_{u_0 + \mu} x_{u_0 + \mu'} X_{\mu + \mu'} \right)
\end{align}
Recall that by definition, $x_{u_0 + \mu}$ is defined by the product of $x_{\alpha v}$ for $\alpha,v$ in the unique expansion of $u_0 + \mu$. As a result, $x_{u_0 + \mu} x_{u_0 + \mu'} = x_{\mu + \mu'}$. On the other hand, any $\mu + \mu' \in \ker \delta^0$ can be written as $|\ker \delta^0|$ distinct sums. Therefore, the above squared sum equals itself.


Finally, we discard the state $X_{\delta u_0} \ket{0_A}$ and apply $\Omega^{-1}$:
\begin{align}
	\ket{\Psi_3} \mapsto  \prod_{\mu_i,\beta} \Omega^\dagger \frac{I + x_{\beta \mu_i} X_{\beta \mu_i}}{2}\Omega \ket{+_a} \otimes \ket{\overline{+}_B,\overline{+}_C}.
\end{align}
We have
\begin{align}
	\begin{aligned}
		\Omega^\dagger X_\mu \Omega
		= & \sum_{u_A, x_B, x_C} (-1)^{-\tr_{\mathbb{F}_q/\mathbb{F}_2}( \int_\xi u_A \smile x_B \smile x_C) } \ket{u_A,x_B,x_C} \bra{u_A,x_B,x_C} \\
		& \cdot \sum_{u \in C^0(X,\F)} \ket{u + \mu}\bra{u} \cdot
		\sum_{u_A', x_B', x_C'} (-1)^{ \tr_{\mathbb{F}_q/\mathbb{F}_2}( \int_\xi u_A' \smile x_B' \smile x_C') } \ket{u_A',x_B',x_C'} \bra{u_A',x_B',x_C'} \\
		= & \sum_{u_A, x_B, x_C} (-1)^{-\tr_{\mathbb{F}_q/\mathbb{F}_2}( \int_\xi u_A \smile x_B \smile x_C) } \ket{u_A,x_B,x_C} \bra{u_A,x_B,x_C} \\
		& \cdot \sum_{u, x_B', x_C'} (-1)^{\tr_{\mathbb{F}_q/\mathbb{F}_2}( \int_\xi u \smile x_B' \smile x_C') } \ket{u+\mu,x_B',x_C'} \bra{u, x_B',x_C'} \\
		= & \sum_{u, x_B, x_C} (-1)^{\tr_{\mathbb{F}_q/\mathbb{F}_2}( \int_\xi \mu \smile x_B \smile x_C) } \ket{u + \mu,x_B,x_C} \bra{u,x_B,x_C}.
	\end{aligned}
\end{align}
The final result is simply $X_\mu \sum_{x_B, x_C} (-1)^{\tr_{\mathbb{F}_q/\mathbb{F}_2}( \int_\xi \mu \smile x_B \smile x_C) } \ket{x_B,x_C} \bra{x_B,x_C}$. Since $\int_\xi \mu \smile x_B \smile x_C = \int_{\mu \frown \xi} x_B \smile x_C$, the second part of the operator is the logical $\CZ$ gate, denoted by $\CZ_{\mu \frown \xi}$, on the second and third code blocks. Therefore, the final quantum state is
\begin{align}\label{eq:magic_projection}
	\begin{aligned}
		\prod_{\mu_i,\beta} \Omega^\dagger \frac{I + x_{\beta \mu_i} X_{\beta \mu_i}}{2} \Omega \ket{+_a} \otimes \ket{\overline{+}_B,\overline{+}_C}
		= & \prod_{\beta,\mu_i} \frac{I + x_{\beta \mu_i} X_{\beta \mu_i} \CZ_{\beta \mu_i \frown \xi} }{2} \ket{+_a} \otimes \ket{\overline{+}_B,\overline{+}_C} \\
		= & \ket{+_a} \otimes \prod_{\beta,\mu_i} \frac{I + x_{\beta \mu_i} \CZ_{\beta \mu_i \frown \xi} }{2} \ket{\overline{+}_B,\overline{+}_C}.
	\end{aligned}
\end{align}		
Discarding all the ancilla qudits, we obtain $\prod_{\beta,\mu_i} \frac{I + x_{\beta \mu_i} \CZ_{\beta \mu_i \frown \xi} }{2} \ket{\overline{+}_B,\overline{+}_C}$, which projects $\ket{\overline{+}_B,\overline{+}_C}$ into the eigenstates of $\CZ_{\beta \mu_i \frown \xi}$. 

Since $\CZ_{\beta \mu_i \frown \xi}$ is a generic composition of logical $\CZ$ gates, the projection yields a joint logical magic state in the following sense: for one qubit pair with $q = 2$, 
\begin{align}
	\frac{I + \CZ}{2} \ket{{+}{+}} = \frac{I + \CZ}{2} \frac{1}{2}\left(\ket{00} + \ket{01} + \ket{10} + \ket{11}\right) = \frac{\sqrt{3}}{2} \frac{(\ket{00} + \ket{01} + \ket{10}}{\sqrt{3}}
\end{align}
is a magic state because every pure stabilizer state has computational-basis support of size $2^k$ for some $k \geq 0$, with all nonzero amplitudes having the same magnitude \cite{DehaeneDeMoor2003}. Therefore, single $\CZ$ projection on $\ket{++}$ has probability $\frac{3}{4}$ to give a magic state. For two qubit pairs initialized in $\ket{{+}{+}{+}{+}}$, the two unnormalized outputs are
\begin{align}
	\frac{I+\CZ_1\CZ_2}{2}\ket{{+}{+}{+}{+}}, \qquad \frac{I-\CZ_1\CZ_2}{2}\ket{{+}{+}{+}{+}}.
\end{align}
It is easy to show that they have $10$ and $6$ computational basis terms and occur with probabilities $5/8$ and $3/8$, respectively. Both normalized outputs are magic and entangled across the two pairs. Similar results can be generalized to more qubit pairs and even for generic qudits with $q = 2^s$ \cite{Dehaene2005,BermejoVegaVanDenNest2014}. For gauging-based magic state preparation~\cite{Zhu2026,Christos2026}, a two-dimensional adaptation of the addressability construction~\cite{LLL2609} provides $\Theta(N)$ logical qubit pairs in good qLDPC blocks, with $\mu_i \in \ker\delta^0$ inducing the individual $\overline{\mathrm{CZ}}_i$ on the selected logical subspace. With these pairs initialized in logical $\ket{{+}{+}}$, the same gauging/ungauging protocol yields $\Theta(N)$ independent encoded magic states with high probability.


\section{Discussion}

In this work, we establish that non-Abelian gauging is compatible with constant encoding rate and linear distance, yielding asymptotically good non-Abelian qLDPC codes. Our construction uses cup products to characterize the coupled logical constraints imposed by gauging. We also provide a foundational treatment of code distance, formulating the problem directly in terms of the general Knill--Laflamme condition and developing new methods for establishing distance bounds through expansion and cleaning. The resulting framework reveals concrete connections between non-Abelian structure, error correction, and quantum magic through the algebraic properties of the underlying codes.

Our constructions and the proof of long-range magic suggest various further intriguing directions in quantum phases of matter and Hamiltonian complexity. In light of the intrinsic connection between long-range magic and non-Abelian topological order in two dimensions~\cite{Wei2025LRM,ZhangKimBaoVijay2026}, our results point toward a general notion of non-Abelian phases beyond geometric locality and provide an algebraic foundation for formulating and understanding such phases of many-body systems. Characterizing the excitations and stability under local perturbations would connect our coding and magic results to a phase description beyond geometric locality, extending recent studies of stabilizer qLDPC phases~\cite{YinLucas2025,DeRoeck2025Phases}.  From the perspective of Hamiltonian complexity, extending the magic obstruction from the code space to a wider low-energy sector would advance the proposed route toward no low-energy trivial magic (NLTM)~\cite{Wei2025LRM}. This requires additional control of low-energy states beyond code distance alone, potentially through suitable soundness or local testability properties of the gauged Hamiltonians. More immediate mathematical questions include whether full-code-space long-range magic can be achieved with strictly linear distance and whether quantitative robustness guarantees can be established under approximate preparation and noise.

For quantum computation, further study of the structure and properties of non-Abelian sheaf codes may lead to new approaches to fault-tolerant storage and logical operations. In particular, further investigation of the use of our magic state generation schemes in fault-tolerant protocols could be valuable. More broadly, assessing the practical advantages of our codes will require developing efficient decoders and fault-tolerant measurement circuits, as well as quantifying the associated resource costs. The flexibility of the sheaf code framework should also be used to design a wider range of codes tailored to practical scenarios, particularly for experimental platforms with reconfigurable connectivity, such as neutral atom arrays~\cite{Xu2024ConstantOverhead,Bluvstein2026Architecture}. This would open promising avenues for optimizing practical quantum computation with non-Abelian codes.

\section*{Acknowledgments} 
This work is supported in part by NSFC under Grant No.~12475023, Dushi Program, and a startup funding from YMSC.

The scientific content of this work was developed by human. Generative AI tools were used during the final stages of manuscript preparation for minor language and presentation edits.

\emph{Note added:} Shortly after this work appeared on arXiv, a related study that discusses non-Abelian sheaf codes~\cite{zhu2026nonAbelian} was posted. The two works offer complementary perspectives: their study develops a twisted gauge theory framework; our work establishes optimal asymptotic scaling in all  parameters. Their scheme achieves $\Theta(N^{1-\epsilon})$ magic state yield. In this update, we use a different scheme based on the linearly addressable cup product in Ref.~\cite{LLL2609} to directly achieve $\Theta(N)$ magic state yield, i.e., a constant magic generation rate.

\printbibliography[heading=bibintoc,title=References]	

@book{Milne_2017, 
place={Cambridge}, 
series={Cambridge Studies in Advanced Mathematics}, 
title={Algebraic Groups: The Theory of Group Schemes of Finite Type over a Field}, publisher={Cambridge University Press}, 
author={Milne, J. S.}, 
year={2017}, 
collection={Cambridge Studies in Advanced Mathematics}
}

@article{Bravyi2024Memory,
  author  = {Bravyi, Sergey and Cross, Andrew W. and Gambetta, Jay M. and Maslov, Dmitri and Rall, Patrick and Yoder, Theodore J.},
  title   = {High-threshold and low-overhead fault-tolerant quantum memory},
  journal = {Nature},
  volume  = {627},
  pages   = {778--782},
  year    = {2024},
  doi     = {10.1038/s41586-024-07107-7},
  url     = {https://www.nature.com/articles/s41586-024-07107-7}
}

@article{Bluvstein2026Architecture,
  author  = {Bluvstein, Dolev and Geim, Alexandra A. and Li, Sophie H. and Evered, Simon J. and Bonilla Ataides, J. Pablo and Baranes, Gefen and Gu, Andi and Manovitz, Tom and Xu, Muqing and Kalinowski, Marcin and Majidy, Shayan and Kokail, Christian and Maskara, Nishad and Trapp, Elias C. and Stewart, Luke M. and Hollerith, Simon and Zhou, Hengyun and Gullans, Michael J. and Yelin, Susanne F. and Greiner, Markus and Vuleti{\'c}, Vladan and Cain, Madelyn and Lukin, Mikhail D.},
  title   = {A fault-tolerant neutral-atom architecture for universal quantum computation},
  journal = {Nature},
  volume  = {649},
  pages   = {39--46},
  year    = {2026},
  doi     = {10.1038/s41586-025-09848-5},
  url     = {https://www.nature.com/articles/s41586-025-09848-5}
}

@article{EvraKaufmanZemor2024,
  author        = {Evra, Shai and Kaufman, Tali and Z{\'e}mor, Gilles},
  title         = {Decodable Quantum {LDPC} Codes beyond the $\sqrt{n}$ Distance Barrier Using High-Dimensional Expanders},
  journal       = {SIAM Journal on Computing},
  volume        = {53},
  number        = {6},
  pages         = {FOCS20-276--FOCS20-316},
  year          = {2024},
  doi           = {10.1137/20M1383689},
  eprint        = {2004.07935},
  archivePrefix = {arXiv},
  primaryClass  = {quant-ph},
  url           = {https://arxiv.org/abs/2004.07935}
}

@article{Gottesman2014,
  author        = {Gottesman, Daniel},
  title         = {Fault-Tolerant Quantum Computation with Constant Overhead},
  journal       = {Quantum Information and Computation},
  volume        = {14},
  number        = {15--16},
  pages         = {1338--1371},
  year          = {2014},
  doi           = {10.26421/QIC14.15-16-5},
  eprint        = {1310.2984},
  archivePrefix = {arXiv},
  primaryClass  = {quant-ph}
}

@misc{LuFuLiu2026Intrinsic,
  author        = {Lu, Yimin and Fu, Esther Xiaozhen and Liu, Zi-Wen},
  title         = {Intrinsic locality dimension of quantum codes},
  year          = {2026},
  eprint        = {2605.31441},
  archivePrefix = {arXiv},
  primaryClass  = {quant-ph},
  url           = {https://arxiv.org/abs/2605.31441}
}

@article{Nayak2008,
  author        = {Nayak, Chetan and Simon, Steven H. and Stern, Ady
                   and Freedman, Michael and Das Sarma, Sankar},
  title         = {{Non-Abelian} anyons and topological quantum computation},
  journal       = {Reviews of Modern Physics},
  volume        = {80},
  number        = {3},
  pages         = {1083--1159},
  year          = {2008},
  doi           = {10.1103/RevModPhys.80.1083},
  eprint        = {0707.1889},
  archivePrefix = {arXiv},
  primaryClass  = {cond-mat.str-el}
}

@article{Xu2024ConstantOverhead,
  author        = {Xu, Qian and Bonilla Ataides, J. Pablo and Pattison, Christopher A. and Raveendran, Nithin and Bluvstein, Dolev and Wurtz, Jonathan and Vasi{\'c}, Bane and Lukin, Mikhail D. and Jiang, Liang and Zhou, Hengyun},
  title         = {Constant-overhead fault-tolerant quantum computation with reconfigurable atom arrays},
  journal       = {Nature Physics},
  volume        = {20},
  pages         = {1084--1090},
  year          = {2024},
  doi           = {10.1038/s41567-024-02479-z},
  eprint        = {2308.08648},
  archivePrefix = {arXiv},
  primaryClass  = {quant-ph},
  url           = {https://www.nature.com/articles/s41567-024-02479-z}
}

@article{YinLucas2025,
  author        = {Yin, Chao and Lucas, Andrew},
  title         = {Low-Density Parity-Check Codes as Stable Phases of Quantum Matter},
  journal       = {PRX Quantum},
  volume        = {6},
  number        = {3},
  pages         = {030329},
  year          = {2025},
  doi           = {10.1103/361k-nj4b},
  eprint        = {2411.01002},
  archivePrefix = {arXiv},
  primaryClass  = {quant-ph},
  url           = {https://arxiv.org/abs/2411.01002}
}

@article{DeRoeck2025Phases,
  author        = {De Roeck, Wojciech and Khemani, Vedika and Li, Yaodong and O'Dea, Nicholas and Rakovszky, Tibor},
  title         = {Low-Density Parity-Check Stabilizer Codes as Gapped Quantum Phases: Stability under Graph-Local Perturbations},
  journal       = {PRX Quantum},
  volume        = {6},
  number        = {3},
  pages         = {030330},
  year          = {2025},
  doi           = {10.1103/7x71-8j7k},
  eprint        = {2411.02384},
  archivePrefix = {arXiv},
  primaryClass  = {quant-ph},
  url           = {https://arxiv.org/abs/2411.02384}
}

@article{AharonovAradVidick2013,
  author        = {Aharonov, Dorit and Arad, Itai and Vidick, Thomas},
  title         = {Guest Column: The Quantum {PCP} Conjecture},
  journal       = {ACM SIGACT News},
  volume        = {44},
  number        = {2},
  pages         = {47--79},
  year          = {2013},
  doi           = {10.1145/2491533.2491549},
  eprint        = {1309.7495},
  archivePrefix = {arXiv},
  primaryClass  = {quant-ph},
  url           = {https://arxiv.org/abs/1309.7495}
}

@book{Goodman2009,
	address = {New York, NY},
	series = {Graduate {Texts} in {Mathematics}},
	title = {Symmetry, {Representations}, and {Invariants}},
	volume = {255},
	isbn = {9780387798516 9780387798523},
	url = {https://link.springer.com/10.1007/978-0-387-79852-3},
	publisher = {Springer New York},
	author = {Goodman, Roe and Wallach, Nolan R.},
	year = {2009},
	doi = {10.1007/978-0-387-79852-3}
}

@book{Bekka2008, 
	place={Cambridge}, 
	series={New Mathematical Monographs}, 
	title={Kazhdan’s Property (T)}, 
	publisher={Cambridge University Press}, 
	author={Bekka, Bachir and de la Harpe, Pierre and Valette, Alain}, 
	year={2008}, 
	collection={New Mathematical Monographs}
}

@book{Hatcher2015AT,
	address = {Cambridge},
	edition = {14th printing 2015},
	title = {Algebraic topology},
	isbn = {9780521791601 9780521795401},
	publisher = {Cambridge University Press},
	author = {Hatcher, Allen},
	year = {2015}
}

@book{Gallier2022,
	address = {Singapore},
	title = {Homology, cohomology, and sheaf cohomology for algebraic topology, algebraic geometry, and differential geometry},
	isbn = {9789811245022 9789811245039},
	publisher = {World Scientific},
	author = {Gallier, Jean H. and Quaintance, Jocelyn},
	year = {2022}
}

@article{Glasner03,
  author = {Glasner, Yair},
  title = {{Ramanujan} Graphs with Small Girth},
  journal = {Combinatorica},
  volume = {23},
  number = {3},
  pages = {487--502},
  year = {2003},
  doi = {10.1007/s00493-003-0029-9},
  eprint = {math/0306196},
  archivePrefix = {arXiv},
  primaryClass = {math.CO},
  url = {https://link.springer.com/article/10.1007/s00493-003-0029-9}
}

@article{Marcus2015I,
	title={Interlacing Families {I}: Bipartite {Ramanujan} Graphs of All Degrees}, 
	author={Adam Marcus and Daniel A. Spielman and Nikhil Srivastava},
	volume={182},
	url={https://doi.org/10.4007/annals.2015.182.1.7},
	DOI={10.4007/annals.2015.182.1.7},
	journal={Annals of Mathematics},
	publisher={Princeton University},
	year={2015},
	pages={307–325} 
}

@article{Chandrasekaran2017,
  author = {Chandrasekaran, Karthekeyan and Velingker, Ameya},
  title = {Shift lifts preserving {Ramanujan} property},
  journal = {Linear Algebra and its Applications},
  volume = {529},
  pages = {199--214},
  year = {2017},
  doi = {10.1016/j.laa.2017.04.031},
  url = {https://www.sciencedirect.com/science/article/pii/S0024379517302689}
}

@article{Hall_2018,
	title={Ramanujan coverings of graphs},
	volume={323},
	ISSN={0001-8708},
	url={http://dx.doi.org/10.1016/j.aim.2017.10.042},
	DOI={10.1016/j.aim.2017.10.042},
	journal={Advances in Mathematics},
	publisher={Elsevier BV},
	author={Hall, Chris and Puder, Doron and Sawin, William F.},
	year={2018},
	month=jan, 
	pages={367–410} 
}

@article{Alon1986Expander,
  author = {Alon, Noga},
  title = {Eigenvalues and expanders},
  journal = {Combinatorica},
  volume = {6},
  number = {2},
  pages = {83--96},
  year = {1986},
  month = jun,
  doi = {10.1007/BF02579166},
  url = {https://link.springer.com/article/10.1007/BF02579166}
}

@article{Alon1994,
  author = {Alon, Noga and Roichman, Yuval},
  title = {Random {Cayley} graphs and expanders},
  journal = {Random Structures \& Algorithms},
  volume = {5},
  number = {2},
  pages = {271--284},
  year = {1994},
  month = apr,
  doi = {10.1002/rsa.3240050203},
  url = {https://onlinelibrary.wiley.com/doi/10.1002/rsa.3240050203}
}

@article{Agarwal2016,
  author = {Agarwal, Naman and Chandrasekaran, Karthekeyan and Kolla, Alexandra and Madan, Vivek},
  title = {On the Expansion of Group-Based Lifts},
  journal = {SIAM Journal on Discrete Mathematics},
  volume = {33},
  number = {3},
  pages = {1338--1373},
  year = {2019},
  doi = {10.1137/17M1141047},
  url = {https://doi.org/10.1137/17M1141047}
}

@InProceedings{Jeronimo2021,
	author =	{Jeronimo, Fernando Granha and Mittal, Tushant and O'Donnell, Ryan and Paredes, Pedro and Tulsiani, Madhur},
	title =	{{Explicit Abelian Lifts and Quantum {LDPC} Codes}},
	booktitle =	{13th Innovations in Theoretical Computer Science Conference (ITCS 2022)},
	pages =	{88:1--88:21},
	series =	{Leibniz International Proceedings in Informatics (LIPIcs)},
	ISBN =	{978-3-95977-217-4},
	ISSN =	{1868-8969},
	year =	{2022},
	volume =	{215},
	editor =	{Braverman, Mark},
	publisher =	{Schloss Dagstuhl -- Leibniz-Zentrum f{\"u}r Informatik},
	address =	{Dagstuhl, Germany},
	URL =		{https://drops.dagstuhl.de/entities/document/10.4230/LIPIcs.ITCS.2022.88},
	URN =		{urn:nbn:de:0030-drops-156846},
	doi =		{10.4230/LIPIcs.ITCS.2022.88}
}

@misc{PK2023RobustlyTestable,
	title={Two-sided Robustly Testable Codes}, 
	author={Gleb Kalachev and Pavel Panteleev},
	year={2023},
	eprint={2206.09973},
	archivePrefix={arXiv},
	primaryClass={cs.IT},
	url={https://arxiv.org/abs/2206.09973}, 
}

@misc{Panteleev2024,
      title={Maximally Extendable Sheaf Codes}, 
      author={Pavel Panteleev and Gleb Kalachev},
      year={2024},
      eprint={2403.03651},
      archivePrefix={arXiv},
      primaryClass={cs.IT},
      url={https://arxiv.org/abs/2403.03651}, 
}

@misc{KP2025Extendable,
	title={Maximally Extendable Product Codes are Good Coboundary Expanders}, 
	author={Gleb Kalachev and Pavel Panteleev},
	year={2025},
	eprint={2501.01411},
	archivePrefix={arXiv},
	primaryClass={cs.IT},
	url={https://arxiv.org/abs/2501.01411}, 
}

@INPROCEEDINGS{QuantumTanner2022,
	author={Leverrier, Anthony and Zémor, Gilles},
	booktitle={2022 IEEE 63rd Annual Symposium on Foundations of Computer Science (FOCS)}, 
	title={Quantum Tanner codes}, 
	year={2022},
	volume={},
	number={},
	pages={872-883},
	doi={10.1109/FOCS54457.2022.00117}
}

@article{Zemor2014,
	title={Quantum {LDPC} Codes With Positive Rate and Minimum Distance Proportional to the Square Root of the Blocklength},
	volume={60},
	ISSN={1557-9654},
	url={http://dx.doi.org/10.1109/TIT.2013.2292061},
	DOI={10.1109/tit.2013.2292061},
	number={2},
	journal={IEEE Transactions on Information Theory},
	publisher={Institute of Electrical and Electronics Engineers (IEEE)},
	author={Tillich, Jean-Pierre and Zémor, Gilles},
	year={2014},
	month=feb, pages={1193–1202}
}

@article{Zeng_2019,
  author = {Zeng, Weilei and Pryadko, Leonid P.},
  title = {Higher-dimensional quantum hypergraph-product codes with finite rates},
  journal = {Physical Review Letters},
  volume = {122},
  number = {23},
  eid = {230501},
  year = {2019},
  month = jun,
  doi = {10.1103/PhysRevLett.122.230501},
  url = {https://doi.org/10.1103/PhysRevLett.122.230501}
}

@misc{Fu2025nogo,
	title={No-go theorems for logical gates on product quantum codes}, 
	author={Esther Xiaozhen Fu and Han Zheng and Zimu Li and Zi-Wen Liu},
	year={2025},
	eprint={2507.16797},
	archivePrefix={arXiv},
	primaryClass={quant-ph},
	url={https://arxiv.org/abs/2507.16797}, 
}

@article{PK2021,
	title={Quantum {LDPC} Codes With Almost Linear Minimum Distance},
	volume={68},
	ISSN={1557-9654},
	url={http://dx.doi.org/10.1109/TIT.2021.3119384},
	DOI={10.1109/tit.2021.3119384},
	number={1},
	journal={IEEE Transactions on Information Theory},
	publisher={Institute of Electrical and Electronics Engineers (IEEE)},
	author={Panteleev, Pavel and Kalachev, Gleb},
	year={2022},
	month=jan, pages={213–229} 
}

@inproceedings{PK2022Good,
	author = {Panteleev, Pavel and Kalachev, Gleb},
	title = {Asymptotically good Quantum and locally testable classical {LDPC} codes},
	year = {2022},
	isbn = {9781450392648},
	publisher = {Association for Computing Machinery},
	address = {New York, NY, USA},
	url = {https://doi.org/10.1145/3519935.3520017},
	doi = {10.1145/3519935.3520017},
	booktitle = {Proceedings of the 54th Annual ACM SIGACT Symposium on Theory of Computing},
	pages = {375–388},
	numpages = {14},
	location = {Rome, Italy},
	series = {STOC 2022}
}

@inproceedings{DHLV2022,
	author = {Dinur, Irit and Hsieh, Min-Hsiu and Lin, Ting-Chun and Vidick, Thomas},
	title = {Good Quantum {LDPC} Codes with Linear Time Decoders},
	year = {2023},
	isbn = {9781450399135},
	publisher = {Association for Computing Machinery},
	address = {New York, NY, USA},
	url = {https://doi.org/10.1145/3564246.3585101},
	doi = {10.1145/3564246.3585101},
	booktitle = {Proceedings of the 55th Annual ACM Symposium on Theory of Computing},
	pages = {905–918},
	numpages = {14},
	location = {Orlando, FL, USA},
	series = {STOC 2023}
}

@INPROCEEDINGS{Dinur2024sheaf,
	author={Dinur, Irit and Lin, Ting-Chun and Vidick, Thomas},
	booktitle={2024 IEEE 65th Annual Symposium on Foundations of Computer Science (FOCS)}, 
	title={Expansion of High-Dimensional Cubical Complexes: with Application to Quantum Locally Testable Codes}, 
	year={2024},
	volume={},
	number={},
	pages={379-385},
	doi={10.1109/FOCS61266.2024.00031}
}

@misc{Li2025Poincare,
	title={Poincar\'e Duality and Multiplicative Structures on Quantum Codes}, 
	author={Yiming Li and Zimu Li and Zi-Wen Liu and Quynh T. Nguyen},
	year={2025},
	eprint={2512.21922},
	archivePrefix={arXiv},
	primaryClass={quant-ph},
	url={https://arxiv.org/abs/2512.21922}, 
}

@misc{LSWLL2026Theory,
      title={Theory of (Co)homological Invariants on Quantum {LDPC} Codes}, 
      author={Zimu Li and Yuguo Shao and Fuchuan Wei and Yiming Li and Zi-Wen Liu},
      year={2026},
      eprint={2603.25831},
      archivePrefix={arXiv},
      primaryClass={quant-ph},
      url={https://arxiv.org/abs/2603.25831}, 
}

@misc{LLL2026nontrivial,
      title={Transversal non-Clifford gates on almost-good quantum {LDPC} and quantum locally testable codes}, 
      author={Yiming Li and Zimu Li and Zi-Wen Liu},
      year={2026},
      eprint={2604.01874},
      archivePrefix={arXiv},
      primaryClass={quant-ph},
      url={https://arxiv.org/abs/2604.01874}, 
}

@article{Kitaev_2003,
	title={Fault-tolerant quantum computation by anyons},
	volume={303},
	ISSN={0003-4916},
	url={http://dx.doi.org/10.1016/S0003-4916(02)00018-0},
	DOI={10.1016/s0003-4916(02)00018-0},
	number={1},
	journal={Annals of Physics},
	publisher={Elsevier BV},
	author={Kitaev, A.Yu.},
	year={2003},
	month=jan, pages={2–30} 
}

@article{Wang_2024,
	title={Efficient fault-tolerant implementations of non-Clifford gates with reconfigurable atom arrays},
	volume={10},
	ISSN={2056-6387},
	url={http://dx.doi.org/10.1038/s41534-024-00945-3},
	DOI={10.1038/s41534-024-00945-3},
	number={1},
	journal={npj Quantum Information},
	publisher={Springer Science and Business Media LLC},
	author={Wang, Yifei and Wang, Yixu and Chen, Yu-An and Zhang, Wenjun and Zhang, Tao and Hu, Jiazhong and Chen, Wenlan and Gu, Yingfei and Liu, Zi-Wen},
	year={2024},
	month=dec 
}

@article{Breuckmann2024Cups,
  author = {Breuckmann, Nikolas P. and Davydova, Margarita and Eberhardt, Jens N. and Tantivasadakarn, Nathanan},
  title = {Cups and Gates {I}: Cohomology Invariants and Logical Quantum Operations},
  journal = {Communications in Mathematical Physics},
  volume = {407},
  eid = {86},
  date = {2026-04-04},
  doi = {10.1007/s00220-026-05570-z},
  eprint = {2410.16250},
  archivePrefix = {arXiv},
  primaryClass = {quant-ph},
  url = {https://link.springer.com/article/10.1007/s00220-026-05570-z}
}

@article{KnillLaflamme1997,
  author  = {Knill, Emanuel and Laflamme, Raymond},
  title   = {Theory of quantum error-correcting codes},
  journal = {Physical Review A},
  volume  = {55},
  number  = {2},
  pages   = {900--911},
  year    = {1997},
  doi     = {10.1103/PhysRevA.55.900}
}

@article{KL2000,
	title={Theory of Quantum Error Correction for General Noise},
	volume={84},
	ISSN={1079-7114},
	url={http://dx.doi.org/10.1103/PhysRevLett.84.2525},
	DOI={10.1103/physrevlett.84.2525},
	number={11},
	journal={Physical Review Letters},
	publisher={American Physical Society (APS)},
	author={Knill, Emanuel and Laflamme, Raymond and Viola, Lorenza},
	year={2000},
	month=mar, pages={2525–2528} 
}

@article{Kalachev2022Cleaning,
   title={A linear-algebraic and lattice-theoretical look at the Cleaning Lemma of quantum coding theory},
   volume={649},
   ISSN={0024-3795},
   url={http://dx.doi.org/10.1016/j.laa.2022.05.002},
   DOI={10.1016/j.laa.2022.05.002},
   journal={Linear Algebra and its Applications},
   publisher={Elsevier BV},
   author={Kalachev, Gleb and Sadov, Sergey},
   year={2022},
   month=Sept, pages={96–121} 
}

@article{Bravyi2009Cleaning,
	title={A no-go theorem for a two-dimensional self-correcting quantum memory based on stabilizer codes},
	volume={11},
	ISSN={1367-2630},
	url={http://dx.doi.org/10.1088/1367-2630/11/4/043029},
	DOI={10.1088/1367-2630/11/4/043029},
	number={4},
	journal={New Journal of Physics},
	publisher={IOP Publishing},
	author={Bravyi, Sergey and Terhal, Barbara},
	year={2009},
	month=apr, pages={043029} 
}

@misc{FreedmanHastings2021,
	title={Building manifolds from quantum codes}, 
	author={Michael Freedman and Matthew B. Hastings},
	year={2021},
	volume={34},
	ISSN={1420-8970},
	url={https://doi.org/10.1007/s00039-021-00567-3},
	DOI={10.1007/s00039-021-00567-3},
	journal={Geometric and Functional Analysis},
	publisher={Springer Nature},
}

@incollection{KB1993,
	title={On the Realization of Networks in Three-Dimensional Space},
	author={A. N. Kolmogorov and Y. M. Barzdin},
	year={1993},
	booktitle = {Selected Works of Kolmogorov},
	publisher = {Kluwer, Dordrecht, 3:194–202}
}

@article{GG2012,
  author = {Gromov, Misha and Guth, Larry},
  title = {Generalizations of the {Kolmogorov--Barzdin} embedding estimates},
  journal = {Duke Mathematical Journal},
  volume = {161},
  number = {13},
  pages = {2549--2603},
  year = {2012},
  doi = {10.1215/00127094-1812840},
  url = {https://doi.org/10.1215/00127094-1812840}
}

@article{Wills2024magic,
  author = {Wills, Adam and Hsieh, Min-Hsiu and Yamasaki, Hayata},
  title = {Constant-overhead magic state distillation},
  journal = {Nature Physics},
  volume = {21},
  number = {11},
  pages = {1842--1846},
  year = {2025},
  doi = {10.1038/s41567-025-03026-0},
  url = {https://www.nature.com/articles/s41567-025-03026-0}
}

@inproceedings{Nguyen2025CCZ,
	author = {Nguyen, Quynh T.},
	title = {Good Binary Quantum Codes with Transversal CCZ Gate},
	year = {2025},
	isbn = {9798400715105},
	publisher = {Association for Computing Machinery},
	address = {New York, NY, USA},
	url = {https://doi.org/10.1145/3717823.3718186},
	doi = {10.1145/3717823.3718186},
	booktitle = {Proceedings of the 57th Annual ACM Symposium on Theory of Computing},
	pages = {697–706},
	numpages = {10},
	location = {Prague, Czechia},
	series = {STOC '25}
}

@misc{He2025addressable,
	title={Quantum Codes with Addressable and Transversal Non-Clifford Gates}, 
	author={Zhiyang He and Vinod Vaikuntanathan and Adam Wills and Rachel Yun Zhang},
	year={2025},
	eprint={2502.01864},
	archivePrefix={arXiv},
	primaryClass={quant-ph},
	url={https://arxiv.org/abs/2502.01864}, 
}

@inproceedings{Anshu_2023, 
	series={STOC ’23},
	title={NLTS Hamiltonians from Good Quantum Codes},
	url={http://dx.doi.org/10.1145/3564246.3585114},
	DOI={10.1145/3564246.3585114},
	booktitle={Proceedings of the 55th Annual ACM Symposium on Theory of Computing},
	publisher={ACM},
	author={Anshu, Anurag and Breuckmann, Nikolas P. and Nirkhe, Chinmay},
	year={2023},
	month=jun, pages={1090–1096},
	collection={STOC ’23} 
}

@misc{Davydova2025,
	title={Universal fault tolerant quantum computation in 2D without getting tied in knots}, 
	author={Margarita Davydova and Andreas Bauer and Julio C. Magdalena de la Fuente and Mark Webster and Dominic J. Williamson and Benjamin J. Brown},
	year={2025},
	eprint={2503.15751},
	archivePrefix={arXiv},
	primaryClass={quant-ph},
	url={https://arxiv.org/abs/2503.15751}, 
}

@misc{Zhu2026,
	title={Non-Abelian qLDPC: TQFT Formalism, Addressable Gauging Measurement and Application to Magic State Fountain on 2D Product Codes}, 
	author={Guanyu Zhu and Ryohei Kobayashi and Po-Shen Hsin},
	year={2026},
	eprint={2601.06736},
	archivePrefix={arXiv},
	primaryClass={quant-ph},
	url={https://arxiv.org/abs/2601.06736}, 
}

@misc{Christos2026,
	      title={Non-Abelian Quantum Low-Density Parity Check Codes and Non-Clifford Operations from Gauging Logical Gates via Measurements}, 
	author={Maine Christos and Chiu Fan Bowen Lo and Vedika Khemani and Rahul Sahay},
	year={2026},
	eprint={2602.12228},
	archivePrefix={arXiv},
	primaryClass={quant-ph},
	url={https://arxiv.org/abs/2602.12228}, 
}

@misc{Wei2025LRM,
      title={Long-range nonstabilizerness and quantum codes, phases, and complexity}, 
      author={Fuchuan Wei and Zi-Wen Liu},
      year={2026},
      eprint={2503.04566},
      archivePrefix={arXiv},
      primaryClass={quant-ph},
}

@article{DehaeneDeMoor2003,
	title = {Clifford group, stabilizer states, and linear and quadratic operations over GF(2)},
	author = {Dehaene, Jeroen and De Moor, Bart},
	journal = {Phys. Rev. A},
	volume = {68},
	issue = {4},
	pages = {042318},
	numpages = {10},
	year = {2003},
	month = {Oct},
	publisher = {American Physical Society},
	doi = {10.1103/PhysRevA.68.042318},
	url = {https://link.aps.org/doi/10.1103/PhysRevA.68.042318}
}

@article{Dehaene2005,
	title = {Stabilizer states and Clifford operations for systems of arbitrary dimensions and modular arithmetic},
	author = {Hostens, Erik and Dehaene, Jeroen and De Moor, Bart},
	journal = {Phys. Rev. A},
	volume = {71},
	issue = {4},
	pages = {042315},
	numpages = {9},
	year = {2005},
	month = {Apr},
	publisher = {American Physical Society},
	doi = {10.1103/PhysRevA.71.042315},
	url = {https://link.aps.org/doi/10.1103/PhysRevA.71.042315}
}

@article{BermejoVegaVanDenNest2014,
  author = {Bermejo-Vega, Juan and {Van den Nest}, Maarten},
  title = {Classical simulations of {Abelian}-group normalizer circuits with intermediate measurements},
  journal = {Quantum Information and Computation},
  volume = {14},
  number = {3-4},
  pages = {181--216},
  year = {2014},
  doi = {10.26421/QIC14.3-4-1},
  url = {https://doi.org/10.26421/QIC14.3-4-1}
}

@misc{LLL2609,
  author        = {Li, Yiming and Li, Zimu and Han, Zhengyi and Liu, Zi-Wen},
  title         = {Parallelizable and addressable transversal non-{Clifford} gates on good quantum {LDPC} codes},
  year          = {2026},
  eprint        = {2609.26691},
  archivePrefix = {arXiv},
  primaryClass  = {quant-ph},
  doi           = {10.48550/arXiv.2609.26691},
  url           = {https://arxiv.org/abs/2609.26691}
}

@misc{GJ2026,
      title={Asymptotically Good Quantum Locally Testable Codes}, 
      author={William Gay and Fernando Granha Jeronimo},
      year={2026},
      eprint={2609.20780},
      archivePrefix={arXiv},
      primaryClass={quant-ph},
      url={https://arxiv.org/abs/2609.20780}, 
}

@misc{BLN2026,
      title={Good Quantum Locally Testable Codes from Product Expansion}, 
      author={Mitali Bafna and Anqi Li and Quynh T. Nguyen},
      year={2026},
      eprint={2609.26735},
      archivePrefix={arXiv},
      primaryClass={quant-ph},
      url={https://arxiv.org/abs/2609.26735}, 
}

@misc{CHLT2026,
      title={Cubical Sheaf Complexes with Constant Expansion with Applications to Asymptotically Good qLTCs}, 
      author={Yeyuan Chen and Miryam Mi-Ying Huang and Yinchen Liu and Er-Cheng Tang},
      year={2026},
      eprint={2609.28028},
      archivePrefix={arXiv},
      primaryClass={cs.IT},
      url={https://arxiv.org/abs/2609.28028}, 
}

@misc{ZhangKimBaoVijay2026,
  author        = {Zhang, Yuzhen and Kim, Isaac H. and Bao, Yimu and Vijay, Sagar},
  title         = {Extensive long-range magic in non-{Abelian} topological orders},
  year          = {2026},
  eprint        = {2605.15150},
  archivePrefix = {arXiv},
  primaryClass  = {quant-ph},
  url           = {https://arxiv.org/abs/2605.15150}
}

@article{twisted_quantum_double,
  author = {Hu, Yuting and Wan, Yidun and Wu, Yong-Shi},
  title = {Twisted quantum double model of topological phases in two dimensions},
  journal = {Physical Review B},
  volume = {87},
  number = {12},
  eid = {125114},
  date = {2013-03-11},
  doi = {10.1103/PhysRevB.87.125114},
  url = {https://journals.aps.org/prb/abstract/10.1103/PhysRevB.87.125114}
}

@article{RSV2019,
  author = {Nithi Rungtanapirom and Jakob Stix and Alina Vdovina},
  title = {Infinite series of quaternionic 1-vertex cube complexes, the doubling construction, and explicit cubical Ramanujan complexes},
  journal = {International Journal of Algebra and Computation},
  volume = {29},
  year = {2019},
  eprint = {1808.03290},
  archivePrefix = {arXiv}
}

@misc{zhu2026nonAbelian,
title={{Sipser}-{Spielman} meets {Dijkgraaf}-{Witten}: non-{Abelian} qLDPC codes via twisted sheaf gauge theory and almost-constant-overhead magic state fountain},
author={Guanyu Zhu and Shi Jie Samuel Tan and Ryohei Kobayashi and Po-Shen Hsin},
year={2026},
eprint={2609.31541},
archivePrefix={arXiv},
primaryClass={quant-ph},
url={https://arxiv.org/abs/2609.31541} }

\end{document}